\documentclass[12pt,a4paper,english]{article}
\usepackage{amsmath}
\usepackage{amsfonts}
\usepackage{amssymb}
\usepackage{amsthm}
\usepackage[mathscr]{euscript}
\usepackage{bbm}
\usepackage{bm}
\usepackage{mathrsfs}
\usepackage{prodint}
\usepackage{accents}

\usepackage{graphicx}
\usepackage[left=1in,right=1in]{geometry}
\usepackage{xcolor}

\usepackage{authblk}
\usepackage[hypertexnames=false]{hyperref}

\usepackage{paralist}
\usepackage{natbib}

\usepackage{tikz}
\usetikzlibrary{shapes,decorations,arrows,calc,arrows.meta,fit,positioning}
\tikzset{
    -Latex,auto,node distance =1 cm and 1 cm,semithick,
    state/.style ={ellipse, draw, minimum width = 0.7 cm},
    point/.style = {circle, draw, inner sep=0.04cm,fill,node contents={}},
    bidirected/.style={Latex-Latex,dashed},
    el/.style = {inner sep=2pt, align=left, sloped}
}

\usepackage{caption}
\usepackage{subcaption}

\title{Doubly Robust Proximal Synthetic Controls}

\date{}
\author[1]{Hongxiang Qiu}
\author[2]{Xu Shi}
\author[3]{Wang Miao}
\author[4]{Edgar Dobriban}
\author[4]{Eric Tchetgen Tchetgen\footnote{Author e-mail addresses: 
		\texttt{qiuhongx@msu.edu},
		\texttt{shixu@umich.edu},
		\texttt{mwfy@pku.edu.cn},
		\texttt{dobriban@wharton.upenn.edu},
		\texttt{ett@wharton.upenn.edu}
}}

\affil[1]{Department of Epidemiology and Biostatistics, Michigan State University}
\affil[2]{Department of Biostatistics, University of Michigan}
\affil[3]{Department of Probability and Statistics, Peking University}
\affil[4]{Department of Statistics and Data Science, University of Pennsylvania}

\allowdisplaybreaks
\newtheorem{theorem}{Theorem}

\theoremstyle{definition}
\newtheorem{remark}{Remark}
\newtheorem{condition}{Condition}
\newtheorem{example}{Example}

\DeclareMathOperator*{\argmin}{argmin}

\DeclareMathOperator{\rank}{rank}

\newcommand{\real}{{\mathbb{R}}}

\newcommand{\ind}{{\mathbbm{1}}}

\newcommand{\cexpect}{{E}}
\newcommand{\Var}{{\mathrm{var}}}
\newcommand{\Prob}{{\mathrm{pr}}}

\newcommand{\intd}{{\mathrm{d}}}

\newcommand{\smallo}{{\mathrm{o}}}

\definecolor{ed}{RGB}{225,0,100}

\newcommand\independent{\protect\mathpalette{\protect\independenT}{\perp}}
\def\independenT#1#2{\mathrel{\rlap{$#1#2$}\mkern2mu{#1#2}}}

\begin{document}

\maketitle

\begin{abstract}
	To infer the treatment effect for a single treated unit using panel data, synthetic control methods construct a linear combination of control units' outcomes that mimics the treated unit's pre-treatment outcome trajectory. 
	This linear combination is subsequently used to impute the counterfactual outcomes of the treated unit had it not been treated in the post-treatment period, and used to estimate the treatment effect.
	Existing synthetic control methods rely on correctly modeling certain aspects of the counterfactual outcome generating mechanism and may require near-perfect matching of the pre-treatment trajectory. 
	Inspired by proximal causal inference, we obtain two novel nonparametric identifying formulas for the average treatment effect for the treated unit: one is based on weighting, and the other combines models for the counterfactual outcome and the weighting function. We introduce the concept of covariate shift to synthetic controls to obtain these identification results conditional on the treatment assignment. We also develop two treatment effect estimators based on these two formulas and the generalized method of moments. One new estimator is doubly robust: it is consistent and asymptotically normal if at least one of the outcome and weighting models is correctly specified. 
	We demonstrate the performance of the methods via
simulations and apply them to evaluate the effectiveness of a Pneumococcal conjugate vaccine on the risk of all-cause pneumonia in Brazil.
\end{abstract}

\section{Introduction} \label{sec: intro}

\subsection{Background}

Interventions such as policies are often implemented in a single unit such as a state, a city, or a school.
Causal inference in these cases is challenging due to the small number of treated units, and due to the lack of randomization and independence.
In various fields including economics, public health, and biometry, synthetic control (SC) methods \citep{Abadie2003,Abadie2010,Abadie2015,Doudchenko2016} are a common tool to estimate the intervention (or treatment) effect for the treated unit in time series from a single treated unit and multiple untreated units in both pre- and post-treatment periods.
For example, SC methods have been used to estimate the effects of terrorist conflicts on GDP \citep{Abadie2003}, 
tobacco control program on tobacco consumption \citep{Abadie2010}, 
Kansas's tax cut on GDP \citep{Ben-Michael2021augsynth,Rickman2018}, 
Florida’s ``stand your ground'' law on homicide rates \citep{Bonander2021}, 
and 
pneumococcal conjugate vaccines on pneumonia \citep{Bruhn2017}.

Classical SCs are linear combinations of control units that mimic the treated unit before the treatment. 
Outcome differences between the treated unit and the SC in the post-treatment period are used to make inferences about the treatment effect for the treated unit.
In \cite{Abadie2003} and \cite{Abadie2010}, a SC is a weighted average of a pool of control units, called the \textit{donors}. 
The weights are obtained by minimizing a distance between the SC and the treated unit in the pre-treatment period, under the constraint that the weights are non-negative and sum to unity. 
Many extensions have been proposed. For example, 
\cite{Abadie2021} proposed methods for multiple treated units, and \cite{Ben-Michael2021} further considered the case where these treated units initiate treatment at different time points; \cite{Doudchenko2016} and \cite{Ben-Michael2021,Ben-Michael2021augsynth} introduced penalization to improve performance; \cite{Athey2021} and \cite{Bai2021} used techniques from matrix completion; \cite{Li2020} studied statistical inference for SC methods; \cite{Chernozhukov2021} and \cite{Cattaneo2021} considered prediction intervals for treatment effects. Among these extensions, some also incorporate the idea that, similarly to the control units' outcomes in the post-treatment period, the treated unit's outcomes in the pre-treatment period can be used to impute the counterfactual outcome had it not been treated \citep{Ben-Michael2021augsynth,Arkhangelsky2021}.

Existing methods often rely on assuming linear models and on the existence of near-perfectly matching weights 
in the observed data. 
Under such assumptions, valid SCs are linear combinations, often weighted averages, of donors. 
However, if such assumptions do not hold, these methods may not produce a valid SC. 
This may happen if the outcomes in the donors have a different measuring scale from the treated unit, or if the treated unit's and the donors’ outcomes have a nonlinear relationship.

To relax these assumptions, \cite{Shi2021} viewed SCs from the proximal causal inference perspective.
For independent and identically distributed (i.i.d.) observations, \citet{Miao2018,Deaner2018,Deaner2021,Cui2020,Tchetgen2020} derived nonparametric identification using proxies, variables capturing the effect of the unmeasured confounders.
\citet{Shi2021} viewed control units' outcomes as proxies and obtained nonparametric identification results for the potential outcome of the treated unit had it not been treated as well as the treatment effect in a general setting, beyond the common linear factor model \citep[e.g., ][]{Abadie2010}. 
They assumed the existence of a function of these proxies, termed \textit{confounding bridge function}, that captures the (possibly nonlinear) effects of unobserved confounders.
With this function, they imputed the expected counterfactual outcomes for the treated unit. 
Estimation of, and inference about, the average treatment effect for the treated unit (ATT) followed from this identification result.
Instrumental variables have also been used. For example, \citet{Holtz-Eakin1988} considered a linear model with interactive fixed effects and showed how to identify it using appropriate instruments. The solution to this problem relies on a particular differencing strategy, which may be viewed as an application of a confounding bridge function. \citet{Cunha2010,Freyberger2018} and references therein considered general nonparametric models with interactive effects, showing how to identify it using appropriate instruments. 
While treatment confounding proxies in proximal causal inference are sometimes described as instruments, it is crucial to note that they are more general than instrumental variables (IVs), 
in the sense that valid IVs are valid treatment confounding proxies, but invalid IVs dependent on hidden confounders are also valid treatment confounding proxies \citep{Tchetgen2020}. 
In addition, while IVs require a form of homogeneity condition for nonparametric identification (e.g., separable errors or monotonicity), proxies do not require such a condition.

\subsection{Our contribution}

Existing methods rely on correctly specifying an \emph{outcome model}, based on which 
one can impute 
the counterfactual outcome trajectory of the treated unit, had it not been treated, after treatment. 
This outcome bridge function model may be difficult to specify correctly, or may not exist.
In this paper, we 
relax this requirement by
leveraging the proximal causal inference framework as in \cite{Shi2021}.
We develop two novel methods to estimate the ATT. One method relies on weighting and is a building block to a second method which we rigorously prove is doubly robust \citep{Bang2005,Scharfstein1999}.
It is consistent and asymptotically normal if either the outcome model or the weighting function is correctly specified, without requiring that both are. 
An advantage of the doubly robust method compared to existing methods is that it allows for misspecifing one of the two models, without the user necessarily knowing which might be misspecified.

We observed that our estimand of interest, the ATT, is closely related to the average treatment effect on the treated for i.i.d. data \cite[e.g., ][]{Hahn1998,Imbens2004,Chen2008,Shu2018}. The method in \cite{Shi2021} corresponds to using an outcome confounding bridge function \citep{Miao2018}, 
which is the proximal causal inference counterpart of G-computation, or an outcome regression-based approach in causal inference under unconfoundedness \citep{Robins1986}. Our proposed methods are motivated by the existing identification results in proximal causal inference in the i.i.d. setting \citep{Cui2020}: one result is based on weighting and the other is based on the influence function. 

Despite these similarities, it remains challenging to adapt these ideas from the i.i.d. setting to panel data. 
Since treatment assignment is often viewed
as fixed in SC problems, a key concept from the i.i.d. setting, the propensity score \citep{Rosenbaum1983}, is undefined.
Thus, existing results for the i.i.d. setting cannot be directly applied to SC problems. 
We leverage the notion of \emph{covariate shift} \citep[e.g.,][]{quinonero2009dataset} to circumvent this issue.
We also find a relaxed version of the i.i.d. assumption to allow for serial correlation, while still obtaining identification via weighting.
We illustrate our proposed methods in simulations and three empirical examples: 
	two examples concern public health outcomes, one studying the effect of the PCV10 vaccine in Brazil on pneumonia \citep{Bruhn2017}, and the other studying the effect of Florida's ``stand your ground'' law on homicide rates \citep{Bonander2021}; the third example concerns economic outcomes, studying the effect of Kansas's tax cut on GDP \citep{Rickman2018}.

Both our doubly robust method and the method in \cite{Ben-Michael2021augsynth} take the form of augmented weighted moment equations, but the known robustness properties of these methods differ. 
In \cite{Ben-Michael2021augsynth}, the treated unit's counterfactual outcome had it not been treated after treatment can be imputed with two approaches.
One is a weighted average of control units' outcomes, identical to classical SC methods \citep{Abadie2010}; the other is a prediction model obtained using the treated unit's outcomes before treatment. By combining these two approaches, \cite{Ben-Michael2021augsynth} developed a SC method with improved performance.
\citet{Arkhangelsky2021} proposed a method combining two imputation approaches based on similar ideas.
Nevertheless, to date, neither method has formally been shown to be doubly robust. In contrast, we formally establish double robustness, inherited from the influence function of the ATT in i.i.d cases.

\section{Problem setup} \label{sec: setup}

We observe data over $T$ time periods.
The first $T_0$ time periods are the pre-treatment periods, and the last $T-T_0$ time periods are the post-treatment periods. 
For the treated unit, at each time period $t=1,\ldots,T$, let $Y_t(0), Y_t(1) \in \real$ be the counterfactual outcome corresponding to no treatment and treatment, respectively, and $Y_t = \ind(t \leq T_0) Y_t(0) + \ind(t > T_0) Y_t(1)$ be the observed outcome. 
At each time period $t$, other variables such as other control units' outcomes are observed. We provide more details about these variables below.
We treat $T$, $T_0$ and the treated unit as deterministic, and treat other variables such as the treated unit's potential outcomes $Y_t(1)$ and $Y_t(0)$ as random. In other words, our proposed methods are conditional on the study design.
We study the ATT causal estimand, that is,
\begin{equation*}
    \phi^*(t) = \cexpect \{ Y_t(1) - Y_t(0) \}
\end{equation*}
in a post-treatment period $t > T_0$.
We treat times, namely $T$ and $T_0$, and all units as deterministic, and treat the potential outcomes as stochastic \citep{Greenland1987,Robins1989,Robins2000,VanderWeele2012}; that is, $Y_t(0)$ and $Y_t(1)$ are both stochastic processes indexed by $t$ that are randomly generated over time, rather than fixed unknown scalar sequences.
In the frequentist interpretation, under repeated sampling, the times and units are all fixed and hence identical for all samples, but the outcomes are randomly generated from a fixed unknown joint distribution and hence may differ across samples.
Stochastic counterfactuals are commonly assumed in the SC literature, implicitly in the random noise or residuals of a linear latent factor model or an autoregressive model \citep[e.g.,][]{Abadie2010,Abadie2021, Ben-Michael2021augsynth,Ben-Michael2021,Athey2021}.
This notion of stochastic counterfactuals $Y_t(0)$ and $Y_t(1)$ as time series is required in our paper because the expectation in $\phi^*(t)$ is taken over the joint distribution of $(Y_t(0),Y_t(1))$.

In the main text, we focus on the case without covariates, and discuss using covariates in Web Appendix~\ref{sec: use covariates}. 
We assume that
all unmeasured confounding
is captured
by a latent factor $U_t$, with assumptions stated in later sections. 
We use $t_-$ and $t_+$ to denote generic times before and after treatment, respectively; that is, $t_- \leq T_0$ and $t_+ > T_0$.
When stating asymptotic results, we consider the asymptotic regime where $T \rightarrow \infty$ with $T_0/T \rightarrow \gamma \in (0,1)$.
This asymptotic regime may be interpreted as the number of observations in both pre- and post-treatment periods growing to infinity, collecting more data before and after the treatment time.
Beyond this asymptotic regime, finite-sample results concerning the error in pre-treatment fitting or treatment effect estimation have been established in previous works \citep[e.g.,][]{Abadie2021,Athey2021,Ben-Michael2021,Ben-Michael2021augsynth}.

\section{Review of identification via outcome modeling} \label{sec: identify ATT outcome}

	In classical SC methods, a weighted average of a pool of control units called \emph{donors} forms the SC \citep{Abadie2010}.
	The motivation for using control units' outcomes to learn about $Y_{t_+}(0)$, despite the presence of potential unmeasured confounder $U_t$, is that these control units may be affected by, and thus contains information about, $U_t$.
	This characteristic resembles that of proxies in proximal causal inference.
	In an i.i.d. setting, proxies capture the effect of the unmeasured confounder on the outcome or treatment assignment. 
	They may also be viewed as noisy observations of the unmeasured confounder.
	In a panel data setting, for each time $t$, we use $W_t \in \mathcal{W}$ to denote a proxy, the vector of the donors' outcomes in classical SC settings. 
Under commonly assumed data-generating assumptions such as the linear factor model \citep{Abadie2010}, donors' outcomes are ideal proxies of $U_t$: any variation in $U_t$ induces some variation in $W_t$.
	We use $Z_t \in \mathcal{Z}$ to denote a general supplemental proxy.
	We next state the causal conditions required and discuss the role of and the choice of $Z_t$.

\begin{condition} \label{cond: proxy independence}
	For all pre-treatment time points $t_-$, 
		the supplemental proxies are independent of the outcomes and the proxies, conditional on the confounders: $Z_{t_-} \independent (Y_{t_-},W_{t_-}) \mid U_{t_-}$.
\end{condition}

Conditions similar to Condition~\ref{cond: proxy independence} are common in proximal causal inference literature \citep[e.g.,][]{Miao2016,Miao2018,Cui2020,Tchetgen2020}. As we will show later, $W_t$ is used to model the outcome $Y_t$ with $Z_t$ supplementing for identification, while $Z_t$ is used to model the weighting process introduced in the following section with $W_t$ supplementing for identification. %
The conditional independence in Condition~\ref{cond: proxy independence} is implied by the factor model in classical SC \citep[e.g.,][]{Abadie2010}, so such assumptions are commonly made implicitly. Causal graphs of Condition~\ref{cond: proxy independence} are in Figure~\ref{fig: DAG}.

\begin{figure}
    \centering
    
    \begin{subfigure}{0.4\textwidth}
    \centering
    \begin{tikzpicture}
    \node[state] (z) at (0,0) {$Z_t$};
    \node[state] (w) at (1*1.5,0) {$W_t$};
    \node[state] (y) at (2*1.5,0) {$Y_t$};
    \node[state,dashed] (u) at (1*1.5,1*1.5) {$U_t$};
    
    \path (u) edge (z);
    \path (u) edge (w);
    \path (u) edge (y);
    
    \path[bidirected] (w) edge[bend left=-60] (y);
    \end{tikzpicture}
    \caption{}
    \label{fig: DAG1}
    \end{subfigure}
    \begin{subfigure}{0.4\textwidth}
    \centering
    \begin{tikzpicture}
    \node[state] (z) at (0,0) {$Z_t$};
    \node[state] (w) at (1*1.5,0) {$W_t$};
    \node[state] (y) at (2*1.5,0) {$Y_t$};
    \node[state,dashed] (u) at (1*1.5,1*1.5) {$U_t$};
    
    \path (u) edge (z);
    \path (u) edge (w);
    \path (u) edge (y);
    \end{tikzpicture}
    \caption{}
    \label{fig: DAG2}
    \end{subfigure}
    \caption{Causal graphs satisfying Condition~\ref{cond: proxy independence} at each time period $t$. The variable $U_t$ is the unobserved confounder. In Figure~\ref{fig: DAG1}, additional unmeasured confounding between proxy $W_t$ and the treated unit's outcome $Y_t$ may be present. In Figure~\ref{fig: DAG2}, $Z_t$, $W_t$ and $Y_t$ are mutually independent conditional on $U_t$, which is often sensible when $(W_t,Z_t)$ are control units' outcomes}.
    \label{fig: DAG}
\end{figure}
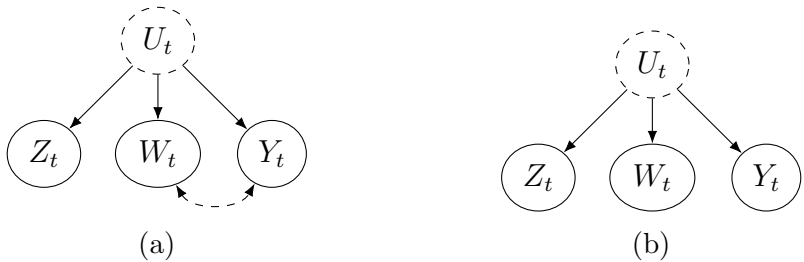

\begin{condition} \label{cond: outcome bridge}
	There exists a function $h^*: \mathcal{W} \rightarrow \real$, $h^*(W_t)$ that captures the conditional mean of $Y_t(0)$ given $U_t$: $\cexpect \{ h^*(W_t) \mid U_t \}=\cexpect \{ Y_t(0) \mid U_t \}$, for all time points $t$.
\end{condition}
The function $h^*$ is called an \emph{outcome confounding bridge function} in proximal causal inference, a terminology we adopt. Condition~\ref{cond: outcome bridge} states that the transformation $h^*(W_t)$ of the observable proxy $W_t$
matches the unobservable $Y_t(0)$ in expectation conditional on $U_t$.
Sufficient conditions for the existence of $h^*$ and examples of $h^*$ can be found in \cite{Shi2021}. The most popular model for $h^*$ is the linear model in classical SC methods \citep{Abadie2010}---that is, the average of $Y_t(0)$ may be imputed by a linear combination of the donors' outcomes $W_t$---but $h^*$ may also be nonlinear. This assumption substantially generalizes the form of SCs by allowing more flexible models. In this condition, we implicitly assume that $W_t$ is not causally impacted by the treatment because of the constant relationship over time.

Throughout, proxy $W_t$ will consist of donors' outcomes \textit{a priori} known not to be causally impacted by the treatment, to make Conditions~\ref{cond: proxy independence}--\ref{cond: outcome bridge} plausible.
When many donors are viable, a subset may be selected based on, for example, the similarity of their outcome trajectories to the treated unit's trajectory before treatment.
Typical choices of $Z_t$ that may satisfy Condition~\ref{cond: proxy independence} include (i) outcomes of control units that are not valid donors, (ii) outcomes of donors excluded from the model $h^*$ to impute $Y_t(0)$, 
and (iii) covariates of donors that are contemporaneous with $(Y_t,W_t)$. The proxy $Z_t$ may be impacted by the treatment. \citet{Shi2021} Section~2.2. contains more discussion about how to choose proxies $W_t$ and $Z_t$.
\cite{Shi2021} showed that, under Condition~\ref{cond: outcome bridge}, for all post-treatment time points $t_+ > T_0$,
\begin{equation}
	\cexpect \{ Y_{t_+}(0) \} = \cexpect \{ h^*(W_{t_+}) \}
	\label{eq: identify EY0 outcome}
\end{equation}
and thus the ATT %
$\phi^*(t_+) = \cexpect \{ Y_{t_+} - h^*(W_{t_+}) \}$. Additionally under Condition~\ref{cond: proxy independence}, $h^*$ solves
\begin{equation}
	\cexpect \{ Y_{t_-}-h(W_{t_-}) \mid Z_{t_-} \}=0 \quad \text{for all pre-treatment time points $t_- \leq T_0$}
	\label{eq: identify h}
\end{equation}
in
$h: \mathcal{W} \rightarrow \real$. Further, under Condition~\ref{cond: outcome complete} in Web Appendix~\ref{sec: completeness}, \eqref{eq: identify h} has a unique solution.

\begin{remark}
In principle, it may be possible to allow $h^*$ to depend on $t$, but it may be challenging to obtain a consistent estimator because only one observation is available at each time point; assumptions such as smoothness may be necessary. 
		Our subsequent confounding bridge functions can also depend on $t$, but we do not pursue this direction.
\end{remark}

\section{Weighted and doubly robust identification of ATT} \label{sec: identify ATT}

The method reviewed above is solely based on the treated unit's outcome process model, similar to G-computation under unconfoundedness \citep{Robins1986} and proximal G-computation \citep{Tchetgen2020}.
Now we introduce our first novel SC method, which is based on weighting and similar to inverse probability weighting under unconfoundedness \citep{Robins1994} and proximal inverse probability weighting \citep{Cui2020}.

To illustrate the idea, for the moment, suppose that the observations are i.i.d. across time. 
Then, the estimand $\phi^*(t_+)=\cexpect \{ Y_{t_+}(1)-Y_{t_+}(0) \}$ corresponds to an average treatment effect on the treated, for which several identification formulas exist \cite[e.g., ][]{Hahn1998,Imbens2004,Chen2008}, 
including those based on outcome regression and weighting. Therefore, one may identify $\cexpect \{ Y_{t_+}(1)-Y_{t_+}(0) \}$ in a SC setting via weighting. To avoid the issue that propensity scores are undefined conditional on the design, we use the concept of covariate shift \citep[e.g., ][]{quinonero2009dataset} and likelihood ratio weighting. 
We next describe our needed causal conditions to identify the ATT via weighting for panel data.

\begin{condition} \label{cond: covariate shift}
	The joint conditional distribution of the counterfactual outcome $Y_t(0)$ and proxy $W_t$ given $U_t$ is identical for all time points $t$.
\end{condition}
Intuitively,
this condition states that once the process $U_t$ is generated, 
$(Y_t(0),W_t)$ are then generated in the same way for all $t$.
The invariance of the conditional distribution of $W_t$ given $U_t$ rules out a causal effect of the treatment on the proxies $W_t$, analogously to the exclusion restriction property of negative control outcomes \citep{Lipsitch2010,Miao2018}.

The next condition states that the marginal distribution of the confounder $U_{t_+}$ is identical for all $t_+$, which is implied by stationarity of $U_{t_+}$. 
Moreover, this condition ensures that only one approach is needed to impute post-treatment counterfactual outcomes $Y_{t_+}(0)$. 

\begin{condition} \label{cond: posttreatment identical distribution}
	The distribution of the unobserved confounder $U_{t_+}$ is identical for all post-treatment time points $t_+ > T_0$.
\end{condition}

Conditions~\ref{cond: covariate shift} and \ref{cond: posttreatment identical distribution} together imply that the distribution of
$(Y_{t_+}(0),W_{t_+},U_{t_+})$
is identical for all $t_+$.
This implication holds under stationarity after treatment, without necessarily requiring stationarity before treatment.
Even if stationarity also holds before treatment, these two conditions hold if the distributions in these two periods differ. Stationarity or similar are often used in time series analysis and other areas, including SCs \citep[e.g.,][]{Hsiao2012,Hahn2017,Li2020,Cattaneo2021,Chernozhukov2021,Ferman2021}. Conditions~\ref{cond: covariate shift}--\ref{cond: posttreatment identical distribution} may thus be plausible in certain applications.
Stationarity in Condition~\ref{cond: posttreatment identical distribution} might be implausible when the number $T-T_0$ of post-treatment time periods is large.
We assume Condition~\ref{cond: posttreatment identical distribution} to facilitate the presentation and present an identification approach applicable to non-stationary cases in Web Appendix~\ref{sec: identify ATT treatment nonstationary}.
These two conditions allow an instantaneous distributional shift in the unobserved confounder $U_t$ at treatment $T_0$. Such an instantaneous shift is a source of confounding. 
Therefore, in general, pre-treatment outcomes $Y_{t_-}$ cannot be directly used to impute post-treatment counterfactual outcomes $Y_{t_+}(0)$.
We also note that, although these two conditions appear to require instantaneous distributional shift, in practice one may specify a window around the treatment in which the transition may not be instantaneous, and restrict the pre- and post-treatment periods before and after the specified window, respectively, 
so that Conditions~\ref{cond: covariate shift} and \ref{cond: posttreatment identical distribution} are plausible.

The next condition states the existence of a \textit{treatment confounding bridge function} \citep{Cui2020}, which models the likelihood ratio, namely a Radon–Nikodym derivative, $\intd P_{U_{t_+}}/$ $\intd P_{U_{t_-}}$
via a regression of the supplemental proxy $Z_{t_-}$ on the unobserved confounder $U_{t_-}$.
\begin{condition} \label{cond: treatment bridge}
	For all times $t_- \leq T_0$ and $t_+ > T_0$, the distribution $P_{U_{t_+}}$ of the unobserved confounder $U_{t_+}$ after treatment is dominated by that before treatment, namely $P_{U_{t_-}}$, and there exists a function $q^*: \mathcal{Z} \rightarrow \real$ capturing the distributional shift in $U_t$; that is,
	\begin{equation}
	    \cexpect \{ q^*(Z_{t_-}) \mid U_{t_-}=u \} = \frac{\intd P_{U_{t_+}}}{\intd P_{U_{t_-}}}(u). \label{eq: define q}
	\end{equation}
\end{condition}

The left-hand side of \eqref{eq: define q} does not depend on $t_+$, due to Condition~\ref{cond: posttreatment identical distribution}. 
We call $q^*$
the treatment confounding bridge function, encoding information on how pre-treatment outcomes $Y_{t_-}$ can be used to impute post-treatment counterfactual outcomes $Y_{t_+}(0)$.
The function $q^*$ is closely related to that for estimating the ATT in i.i.d settings, such as in Theorem~5.1 of \citet{Cui2020}.
One consequence of Condition~\ref{cond: treatment bridge} is that the distribution of $U_{t_+}$ must be dominated by that of $U_{t_-}$, for all $t_-\le T_0$ and $t_+ > T_0$. 
	As shown in Theorem~\ref{thm: identify ATT treatment} Equation~\ref{eq: identify q} below, under Conditions~\ref{cond: proxy independence} and \ref{cond: covariate shift}--\ref{cond: treatment bridge}, this further implies that the distribution of $W_{t_+}$ is dominated by that of  $W_{t_-}$ for all $t_-\le T_0$ and $t_+ > T_0$, which is a testable condition. 
	Thus, for time series data subject to a significant secular trend, especially a monotone trend, we recommend a pre-processing step to remove, to the extent possible, any significant secular trends to make Condition~\ref{cond: treatment bridge} as plausible as possible. We provide concrete approaches for removing time trends in Web Appendix~\ref{sec: data analysis2} and the corresponding simulation results in Web Appendix~\ref{sec: sim detrend}. Although detrending with our approach may lead to slightly to moderately anti-conservative inference, failing to correct for such a trend will likely compromise one's ability to implement the proposed weighted approach successfully.
    It is still an open question how to account for time trends appropriately to obtain asymptotically valid inference in our proposed methods.
Sufficient conditions for the existence of $q^*$ can be found in Web Appendix~\ref{sec: bridge exists}.
We list a few examples of treatment confounding bridge functions below.

\begin{example} \label{ex: normal}

    Suppose that, for all $t_- \leq T_0$, $t_+ > T_0$ and $t$, 
    $U_{t_-} \sim \mathrm{N}(0,\sigma_-^2)$, 
    $U_{t_+} \sim \mathrm{N}(0,\sigma_+^2)$, 
    $Z_t \mid U_t \sim \mathrm{N}(a U_t, \sigma^2)$
    for some $a \neq 0$. If $\sigma^{-2} a^2 - \sigma_+^{-2} + \sigma_-^{-2} > 0$, then $q^*: z \mapsto \exp(\alpha + \beta z^2)$ for $\beta = \sigma^{-2}(\sigma_-^{-2}-\sigma_+^{-2})/\{2 (\sigma^{-2} a^2 - \sigma_+^{-2} + \sigma_-^{-2}) \}$ and some $\alpha\in \mathbb{R}$ satisfies Condition~\ref{cond: treatment bridge}. 
    In this example, we only specify the marginal distributions of $(U_t,Z_t)$ for each $t$ and allow for serial correlation. For instance, $U_t$ could be generated from an autoregressive model that is stationary before and after treatment, respectively.
    This might occur if the treatment is implemented at a random time point $T_0$ due to an abrupt change in $U_t$ in a time window around $T_0$; with the observations in this window excluded from analysis, $U_t$ can be stationary before and after treatment with possibly different distributions.
    Moreover, if $U_t$ is stationary, Condition~\ref{cond: treatment bridge} holds with the constant bridge function $q^*: z \mapsto 1$ even if $T_0$ is random.
    Similar results hold for data marginally distributed as multivariate normal at each time $t$.
\end{example}

\begin{example} \label{ex: exponential}
    Suppose that $U_t=(U_{t,1},\ldots,U_{t,K})$ has all coordinates mutually independent at any time $t$, while there may be serial correlations for each element over time. 
    Let $U_{t_-,k} \sim \mathrm{Exponential}(\lambda_{-})$ and $U_{t_+,k} \sim \mathrm{Exponential}(\lambda_{+})$, for $k=1,\ldots,K$, with $\lambda_{+} \geq \lambda_{-}$. Suppose that $Z_t=(Z_{t,1},\ldots,Z_{t,K})$ with mutually independent entries distributed as, conditional on $U_t$, $Z_{t,k} \mid U_t \sim \mathrm{Poisson}(\alpha + \beta U_{t,k})$ for some scalars $\alpha \geq 0$ and $\beta>0$. Then, direct calculation shows that $q^*: z \mapsto \prod_{k=1}^K \exp(a + b z_k)$ satisfies Condition~\ref{cond: treatment bridge} for some scalars $a$ and $b$.
\end{example}

\begin{remark}
    Condition~\ref{cond: treatment bridge} may fail when $T_0$ is random and depends on the unobserved confounder $U_t$, even when we condition on $T_0$ to mimic a fixed treatment time design. Consider the following counterexample. Suppose that $U_t, t\ge 0$ are i.i.d. across $t$, with support being $\real$, and $T_0 = \max \{t: U_t < a\}$ for an unknown fixed number $a$. 
    This corresponds to the case where the treatment initiates immediately when $U_t$ crosses the threshold $a$. We assume i.i.d. and unbounded $U_t$s only for simplicity. 
    Conditional on $T_0$, since $\Prob(U_{t-} < a)=1$ but $\Prob(U_{t_+} \geq a) > 0$, the Radon-Nikodym derivative $\intd P_{U_{t_+}}/\intd P_{U_{t_-}}$ in Condition~\ref{cond: treatment bridge} does not exist. Thus, this condition fails. Therefore, a fixed treatment time is crucial to the weighted approach to ATT, and conditioning on a random $T_0$ might not suffice.
\end{remark}

The above conditions lead to our first formal identification result.

\begin{theorem}[Identification of ATT with $q^*$] \label{thm: identify ATT treatment}
	Let $f:\real \rightarrow \real$ be any square-integrable function. Under Conditions~\ref{cond: proxy independence} and \ref{cond: covariate shift}--\ref{cond: treatment bridge}, for all $t_- \leq T_0$ and $t_+ > T_0$,
	\begin{equation}
	    \cexpect [ f \{ Y_{t_+}(0) \} ] = \cexpect \{ q^*(Z_{t_-}) f(Y_{t_-}) \}.
	    \label{eq: identify EfY0 treatment}
	\end{equation}
	In particular, taking $f$ to be the identity function, it holds that
	\begin{equation}
        \cexpect \{ Y_{t_+}(0) \} = \cexpect \{ q^*(Z_{t_-}) Y_{t_-} \}
		\label{eq: identify EY0 treatment}
	\end{equation}
	and thus the ATT is identified as
	$\phi^*(t_+)=\cexpect \{ Y_{t_+}-q^*(Z_{t_-}) Y_{t_-} \}$
	for all $t_- \leq T_0$ and $t_+ > T_0$. In addition, 
	$P_{W_{t_+}}$ is dominated by $P_{W_{t_-}}$; the treatment confounding bridge function $q^*$ solves
	\begin{equation}
		\cexpect \{ q(Z_{t_-}) \mid W_{t_-}=w \} = \frac{\intd P_{W_{t_+}}}{\intd P_{W_{t_-}}}(w) \quad \text{for all $t_- \leq T_0$}
		\label{eq: identify q}
	\end{equation}
	in
    $q: \mathcal{Z} \rightarrow \real$.
    Further, under Condition~\ref{cond: treatment complete} in Web Appendix~\ref{sec: completeness}, \eqref{eq: identify q} has a unique solution.
\end{theorem}

In contrast to SC methods based on outcome modeling,
the identifying expression \eqref{eq: identify EY0 treatment} cannot be directly interpreted as an outcome trajectory. 
Indeed, given a treatment confounding bridge function $q^*$, the right-hand side of \eqref{eq: identify EY0 treatment} only depends on observations before treatment but not after treatment. 
We also obtain a novel doubly robust identification result, which is motivated by doubly robust estimation of the average treatment effect on the treated \citep{Cui2020} and is the basis of doubly robust inference.

\begin{theorem}[Doubly robust identification] \label{thm: identify ATT DR}
	Let
	$h: \mathcal{W} \rightarrow \real$
	and
	$q: \mathcal{Z} \rightarrow \real$
	be square-integrable. Under Conditions~\ref{cond: proxy independence}, \ref{cond: covariate shift} and \ref{cond: posttreatment identical distribution}, 	for all $t_- \leq T_0$ and $t_+ > T_0$, we have
    \begin{align}
	    \begin{split}
	        \phi^*(t_+) &= \cexpect \left[ Y_{t_+} - q(Z_{t_-}) \{Y_{t_-} - h(W_{t_-})\} - h(W_{t_+}) \right],
	    \end{split} \label{eq: DR identify ATT}
	\end{align}
if (i) Condition~\ref{cond: outcome bridge} holds and $h=h^*$,
	or (ii) Condition~\ref{cond: treatment bridge} holds and $q=q^*$.
\end{theorem}

Theorem~\ref{thm: identify ATT DR} states that the ATT $\phi^*(t_+)$ is identified by a single formula if \emph{at least} one of the two nuisance functions $h^*$ or $q^*$ is known, and thus doubly robust estimation \citep{Bang2005,Robins1995a,Robins1995,Scharfstein1999} is possible. 
This result holds even if one of $h^*$ and $q^*$ exists but the other does not; that is, Conditions~\ref{cond: outcome bridge} and \ref{cond: treatment bridge} need not hold simultaneously.

\section{Doubly robust inference about ATT} \label{sec: estimate ATT}

We assume that one specifies parametric models for the confounding bridge functions $h^*$ and $q^*$. We use $h_\alpha$ and $q_\beta$ to denote the models for confounding bridge functions parameterized by $\alpha \in \mathcal{A} \subseteq \real^{d_\alpha}$ and $\beta \in \mathcal{B} \subseteq \real^{d_\beta}$, respectively. For example, in classical SC methods,
it is typically assumed that $h^*(w)=w^\top \alpha_0$ for a vector $\alpha_0$ of non-negative numbers that sum up to unity \citep{Abadie2010}. In this case, we may take $h_\alpha$ to be $w \mapsto w^\top \alpha$. 
In some cases, e.g., in Example~\ref{ex: exponential}, we may take $q_\beta(z)=\exp(\beta_0 + z^\top \beta_1)$ where $\beta=(\beta_0,\beta_1^\top)^\top$.

We assume that the function $t_+ \mapsto \phi^*(t_+)$ encoding the potentially time-varying ATT is correctly parameterized by $\lambda \in \Lambda \subseteq \real^{d_\lambda}$. We use $\phi_\lambda$ to denote this model. For example, the ATT is commonly assumed to be constant overtime, which holds under stationarity of $\{Y_{t_+}\}_{t_+ > T_0}$ and Conditions~\ref{cond: covariate shift}--\ref{cond: posttreatment identical distribution}. In this case, we may set $\phi_\lambda$ to be constant $\lambda \in \real$.

By \eqref{eq: identify h} and \eqref{eq: identify q} respectively, we have that
\begin{equation}
    \cexpect[ \{ Y_{t_-} - h^*(W_{t_-}) \} g_Z(Z_{t_-})] = 0 \quad \text{and} \quad \cexpect \{ q^*(Z_{t_-}) g_W(W_{t_-}) - g_W(W_{t_+}) \} = 0
    \label{eq: implied moment equation}
\end{equation}
for any functions $g_Z: \mathcal{Z} \rightarrow \real$
and $g_W: \mathcal{W} \rightarrow \real$.
We propose to use the generalized method of moments (GMM) to estimate $h^*$, $q^*$, and the ATT \citep[e.g., ][]{Hansen1982,Wooldridge1994,Hall2007}. 
This method involves a parameter vector $\theta=(\alpha,\beta,\lambda,\psi,\psi_-) \in \Theta$ including nuisance parameters $\psi$ and $\psi_-$, a moment equation $G_t$ for each time $t$, and a weight matrix $\Omega_T$ that may depend on sample size $T$.
Details of this method are presented in Web Appendix~\ref{sec: technical conditions}.
We need some additional conditions to obtain valid inferences about $\phi^*$.

\begin{condition} \label{cond: misspecify h or q}
	(i) There is a unique parameter value $\theta_\infty = (\alpha_\infty,\beta_\infty,\lambda_\infty,\psi_\infty,\psi_{-,\infty}) \in \Theta$ such that $\cexpect \{ G_t(\theta_\infty) \}=0$ for all time points $t$. (ii) The function $h^*=h_{\alpha_\infty}$ is a valid outcome confounding bridge function satisfying Condition~\ref{cond: outcome bridge}, or the function $q^*=q_{\beta_\infty}$ is a valid treatment confounding bridge function satisfying Condition~\ref{cond: treatment bridge}.
\end{condition}

Part~(i) of Condition~\ref{cond: misspecify h or q} is an identifying condition to ensure a unique solution to the population moment equation, which is standard for GMM. 
If both $h^*$ and $q^*$ are correctly specified, part~(i) would hold if they are both unique.
Part~(ii) requires correct parametric specification of at least one of, but not necessarily both of, $h^*$ or $q^*$, without necessarily knowing \textit{a priori} which model might be incorrect.
Under Condition~\ref{cond: misspecify h or q}, Theorem~\ref{thm: identify ATT DR} implies that $\phi_{\lambda_\infty}$ equals the target estimand, ATT $\phi^*$. 
By standard asymptotic theory for GMM \citep[e.g., Theorems~7.1 and 7.2 in][]{Wooldridge1994} along with Theorem~\ref{thm: identify ATT DR}, we have that $\hat{\theta}_T$ is consistent for $\theta_\infty$ and asymptotically normal under conditions, as stated in Theorem~\ref{thm: DR ATT asymptotic normality} below.
If the data are i.i.d., as we argued in Section~\ref{sec: identify ATT}, the estimand reduces to the average treatment effect on the treated and our estimator is locally efficient \citep{Cui2020}.

\begin{theorem} \label{thm: DR ATT asymptotic normality}
	Under Conditions~\ref{cond: proxy independence}, \ref{cond: covariate shift}, \ref{cond: posttreatment identical distribution}, \ref{cond: misspecify h or q}, as well as \ref{cond: positive weight matrix}--\ref{cond: reg conditions for estimating function} and \ref{cond: uniform law of large numbers} in Web Appendix~\ref{sec: technical conditions}, with the estimator $\hat{\theta}_T$ from \eqref{eq: thetahat} and $\theta_\infty$ in Condition~\ref{cond: misspecify h or q}, it holds that, 
	as $T\to\infty$,
	$\hat{\theta}_T$ is consistent for $\theta_\infty$. 
	Additionally, under Conditions~\ref{cond: full rank matrix}, \ref{cond: uniform law of large numbers derivative} and \ref{cond: CLT weighted moment}, as $T\to\infty$, $\sqrt{T} (\hat{\theta}_T-\theta_\infty) \overset{d}{\to} \mathrm{N}(0,A^{-1} B A^{-1})$,
    where $A = R^\top \Omega R$, $\Omega$ is from Condition~\ref{cond: positive weight matrix},
	\begin{align*}
		& R=\lim_{T \rightarrow \infty} \frac{1}{T} \sum_{t=1}^T \cexpect \{ \nabla_\theta G_t(\theta) |_{\theta=\theta_\infty} \},\,\,
         B = R^\top \Omega \left[ \lim_{T \rightarrow \infty} \Var \left\{ T^{-1/2} \sum_{t=1}^T G_t(\theta_\infty) \right\} \right] \Omega R.
	\end{align*}
    Suppose that the ATT function $\phi_\lambda$ is differentiable with respect to $\lambda$ and let $\dot{\phi}_\lambda$ denote this partial derivative. Thus, with $\Pi := (0_{d_\lambda \times (d_\alpha + d_\beta)}, I_{d_\lambda},0_{d_\lambda \times (d_{\beta'}+1}))$ being the matrix consisting of zeros and ones with dimensions denoted in the subscript, for every $t_+ > T_0$, it holds that $\phi_{\lambda_\infty}$ is the ATT at time $t_+$ and $\sqrt{T} (\phi_{\hat{\lambda}_T}(t_+)-\phi_{\lambda_\infty}(t_+)) \overset{d}{\to} \mathrm{N}(0,\dot{\phi}_{\lambda_\infty}(t_+)^\top \Pi^\top A^{-1} B A^{-1} \Pi \dot{\phi}_{\lambda_\infty}(t_+))$.
\end{theorem}

One can in principle use any GMM implementation, and we use the standard \texttt{R} package \texttt{gmm} for our simulation and data analyses. 
Confidence intervals of the ATT follows from standard outputs of \texttt{gmm}. In particular, when $\phi_\lambda$ is a constant function, since $\lambda_\infty$ is a component of $\theta_\infty$, a Wald test or confidence interval about the ATT $\phi_{\lambda_\infty}(t_+) = \lambda_\infty$ follows immediately from inference about $\theta_\infty$.
We have noted some numerical instabilities with this implementation in our simulations, and provide our empirical suggestions to alleviate numerical issues in Web Appendix~\ref{sec: implementation of GMM}. We expect these issues to be alleviated by using an improved GMM software implementation with, for example, more numerically stable optimization algorithms that are more capable to handle nonconvex problems, potentially under constraints.

Our asymptotic results in Theorem~\ref{thm: DR ATT asymptotic normality} rely on the total number $T$ of time periods tending to infinity. Thus, our proposed doubly robust method is applicable to cases with a large number of time periods and a relatively small number of model parameters. With many control units, expertise might be required to reduce the number of parameters before analysis using our method. Methods with different theoretical results that allow for a large number of control units include \citet{Athey2021,Ben-Michael2021,Ben-Michael2021augsynth}, among others.

The method and results corresponding to the identification results in
Theorem~\ref{thm: identify ATT treatment} based on weighting alone are
similar and can be found in
Web Appendix~\ref{sec: weighting GMM}. This is largely based on the outcome modeling-based approach developed by \cite{Shi2021}.
Compared to these two methods, the doubly robust method has the advantage that it only requires correct specification of one of $h^*$ and $q^*$ in a parametric model, but not necessarily both. One potential drawback of the doubly robust method, however, is that more parameters need to be estimated in GMMs compared to the other two methods. This issue of dimensionality might limit the usage of complicated models for $h^*$, $q^*$ and $\phi^*$ when the time series is short.
In particular, if the number of parameters is comparable to the total number $T$ of time periods, the doubly robust method might lead to numerical instability due to too many parameters being estimated and might be impractical.
Our methods have no guarantee in such scenarios.
However, the number of parameters is much smaller than $T$ in many applications, for example, in our data analyses in Section~\ref{sec: vaccine analysis} and Web Appendix~\ref{sec: data analysis2}.
The number of control units is often highly related to the number of parameters, and the number of control units is commonly small compared to the number of time periods in SC applications, including the motivating ones \citep[e.g.,][]{Abadie2003,Abadie2010}.
Models with many variables might possibly be used by utilizing methods in \citet{Deaner2021} under linearity and sparsity.
\citet{Chamberlain1992,Hansen1982,Hansen1985} showed that GMM attains the efficiency bound under conditional moment restrictions with a suitably chosen weighting matrix $\Omega_T$ in i.i.d settings as well as panel data settings where asymptotics are in the number of units. We do not pursue high-dimensional models or efficiency further in this work and restrict ourselves to the setting of a bounded number of units (a single treated unit and a finite number of control units) while relying on large $T$ asymptotics; to the best of our knowledge, none of these prior results can directly apply to our more challenging setting and existing semiparametric efficiency theory does not appear to directly apply without additional restrictions (e.g., Markov restrictions).

\section{Simulations} \label{sec: sim}

We investigate the performance of our methods to estimate the constant ATT $\phi^* \equiv \phi^*(t_+)$ in several simulations. 
Here we present the first simulation where
the moment equations in the GMM are just identified.
We compare the following methods: \texttt{correct.DR}, the doubly robust method with correctly specified parametric $h^*$ and $q^*$; \texttt{correct.h}, the outcome confounding bridge method from \cite{Shi2021} with correctly specified $h^*$; \texttt{correct.q}, the treatment confounding bridge method described in Web Appendix~\ref{sec: weighting GMM} with correctly specified $q^*$; \texttt{mis.h.DR}, the doubly robust method with misspecified $h^*$ and correctly specified $q^*$; \texttt{mis.q.DR}, the doubly robust method with misspecified $q^*$ and correctly specified $h^*$; \texttt{mis.h}, the outcome confounding bridge method with misspecified $h^*$; \texttt{mis.q}, the treatment confounding bridge method with misspecified $q^*$; \texttt{OLS}, the method based on unconstrained ordinary least squares and similar to the method from \cite{Abadie2010}. 
\texttt{OLS} finds the linear combination of donors' outcomes that best fits the treated unit's outcome before treatment, and uses this combination as the SC.
All methods except \texttt{OLS} are based on the proximal causal inference perspective.
We let the number of latent confounders, the number of donors and of the other control units, all equal to $K$, range in \{2,3,4,5\}.
Details of the data-generating mechanism are presented in Wed Appendix~\ref{sec: sim DGP}.

We next present the simulation results. 
We have run 16,000 GMM involving weighting via a treatment confounding bridge function.
Among them, 
only one run had numerical errors. 
The sampling distributions of the estimated ATT is presented in Figure~\ref{fig: sim phi hat}. In all settings, the \texttt{OLS} estimator is biased. 
When at least one of $h^*$ or $q^*$ is correctly specified, our proposed doubly robust method appears consistent and asymptotically normal, aligning with Theorem~\ref{thm: DR ATT asymptotic normality}. 
The other two methods, based on only one of $h^*$, $q^*$, are not doubly robust. 
They are consistent and asymptotically normal only when the bridge function they rely on is correctly specified, but are biased otherwise. The 95\%-Wald confidence interval coverage of all above methods is presented in Figure~\ref{fig: sim CI cover}. In large samples, confidence intervals based on consistent and asymptotically normal estimators have coverage close to the nominal level; otherwise, the confidence interval coverage is much lower than the nominal level.

\begin{figure}[bt]
    \centering
    \includegraphics[scale=0.7]{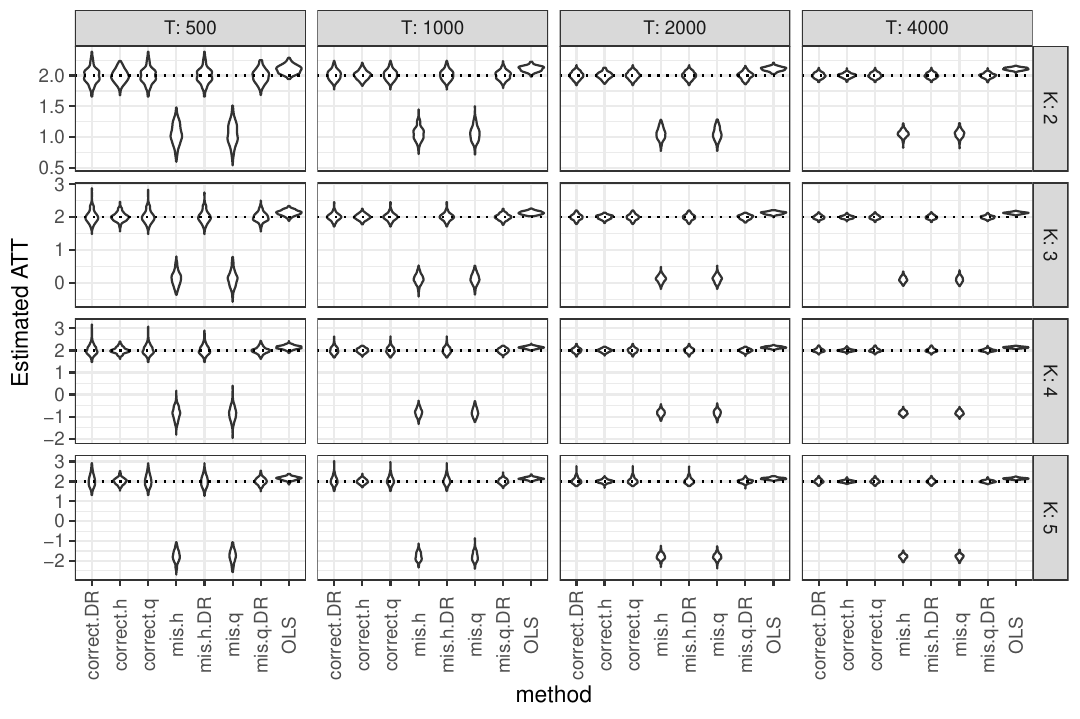}
    \caption{Sampling distribution of estimated average treatment effect. The horizontal dotted line is the true average treatment effect for the treated unit.}
    \label{fig: sim phi hat}
\end{figure}

\begin{figure}[bt]
    \centering
    \includegraphics[scale=0.7]{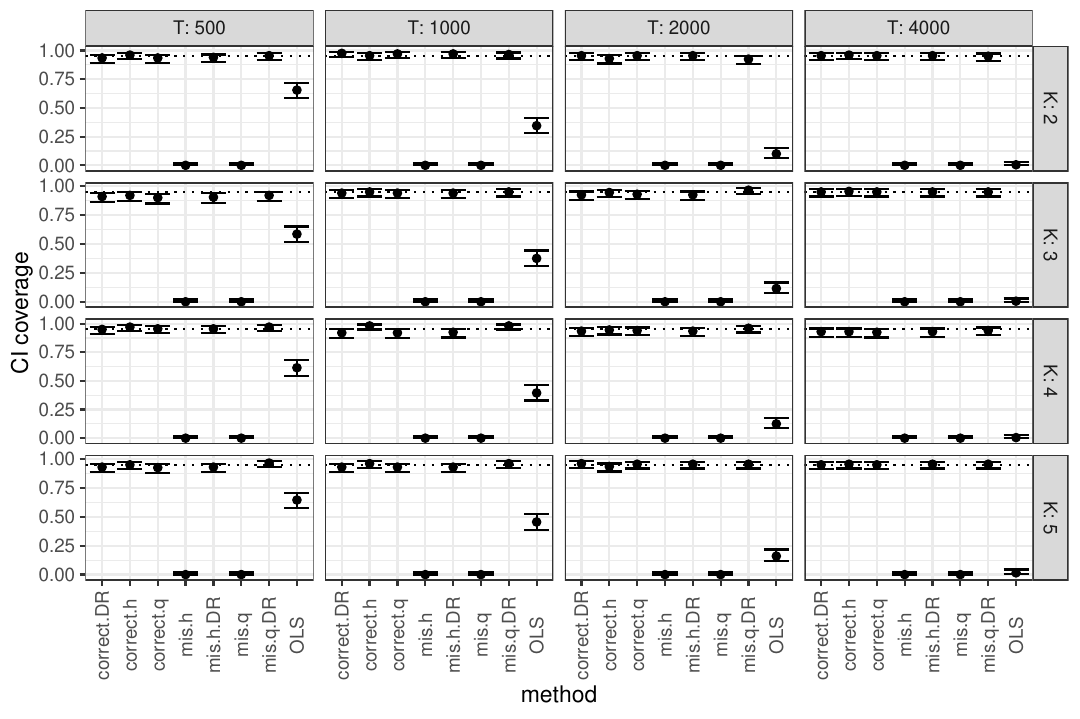}
    \caption{95\%-Wald confidence interval coverage of the average treatment effect for the treated unit for all methods in the simulation under the just-identified setting. The error bars represent 95\%-Wilson score intervals for the confidence interval coverage. The horizontal dotted line is the nominal coverage 95\%. We set the confidence interval to be the empty set when numerical errors occur to obtain a conservative Monte Carlo estimate of confidence interval coverage.}
    \label{fig: sim CI cover}
\end{figure}

\section{Analysis of Brazil all-cause pneumonia hospitalization} \label{sec: vaccine analysis}
	
We study the effect of the introduction of the PCV10 vaccine in Brazil on the number of hospitalizations due to all-cause pneumonia \citep{Bruhn2017}. Monthly hospitalizations and their causes were collected from 2003 to 2013. We focus on the subpopulation of children less than 12 months old in this analysis. PCV10 was introduced to Brazil in January 2010. Following the analysis in \citet{Bradley2017}, we allow two years for the introduction of the vaccine to take effect and set the evaluation period to be 2012--2013.

We view each group of causes of hospitalization as a unit and dismiss units with missing data.
The time series data does not have a clear monotone trend; thus, Condition~\ref{cond: treatment bridge} may be plausible.
To alleviate numerical issues due to non-linearity in GMM and to reduce differences in scaling between units, 
we scale the numbers of hospitalizations due to each group of causes to the unit interval before analysis (see Web Appendix~\ref{sec: implementation of GMM} for more details).
We model the outcome bridge function as a linear function $h_\alpha: w \mapsto (1,w^\top)^\top \alpha$, and select the number of hospitalizations due to the following three groups of causes (units) as the proxies $W_t$: (i) bronchitis, bronchiolitis and unspecified acute lower respiratory infection, (ii) endocrine, nutritional, metabolic disorders, and (iii) malnutrition.
This choice is motivated by \citet{Bradley2017} and the prior knowledge about their relation to pneumonia.

We use all other groups of causes as control units in the supplemental proxies $Z_t$.
We consider three parametrizations of the treatment confounding bridge function with increasing sets of units included in the model: $q_\beta(Z_t) = \exp ( \beta_0 + \sum_{j=1}^J \beta_j Z_j )$ for $J=1,2,3$, where $Z_1$, $Z_2$ and $Z_3$, respectively, are the numbers of hospitalizations due to (i) certain infectious and parasitic diseases, except intestinal, (ii) diseases of blood and blood-forming organs and certain disorders involving the immune mechanism, and (iii) premature delivery and low birth weight. 
The motivation for choosing these supplemental proxies $Z_j$ is to capture the effect of unmeasured confounders $U_t$ on general health issues related to infections and the immune system, which are both associated with the outcome of interest, all-cause pneumonia hospitalizations, among children less than 12 months old.
We refer to the corresponding doubly robust (resp., weighted) estimators as \texttt{DR}, \texttt{DR2} and \texttt{DR3} (resp., \texttt{treatment bridge}, \texttt{treatment bridge2} and \texttt{treatment bridge3}).
As shown in the simulation in Web Appendix~\ref{sec: sim post-select}, choosing $W$ and $Z$ completely based on the data might lead to estimates with increased uncertainty
        due to the possibility of model misspecification. Such estimates may not be informative.  
        When selecting units empirically, we recommend, when possible, using prior knowledge and not solely relying on the data, in which case this issue might be avoided.
Ideally, proxies should be selected \textit{a priori} to avoid potential post-selection inference issues.

The GMM estimator outlined in Section~\ref{sec: estimate ATT} is implemented with the user-specified functions being $g_h: z \mapsto (1,z^\top)^\top$, where we recall that $z$ is the collection of all supplemental proxies, namely hospitalizations due to non-donor causes, 
and $g_q: w \mapsto (1,w^\top)^\top$. We set the  weight matrix to equal the identity matrix.

Besides proximal SC estimators, we also report results for the standard \texttt{OLS} estimator described in Section~\ref{sec: sim}. 
We also consider the regression-based SC method (\texttt{Abadie's SC}) proposed by \cite{Abadie2010} with the same three donors forming the proxy $W_t$ as above.
The point estimates and 95\% confidence intervals for the ATT are presented in Table~\ref{tab: vaccine results}. The trajectories of SCs and the actual number of hospitalizations are presented in Figure~\ref{fig: vaccine main analysis}.

\begin{table}
    \centering
    \begin{tabular}{l|c|c}
        Method & PCV10 & placebo \\
        \hline\hline
        \texttt{Abadie's SC} & 409 & 3092 \\
        \texttt{OLS} & -3533 (-4137, -2930) & 253 (-287, 794) \\
        \texttt{DR} & -2745 (-3559, -1931) & 1192 (501, 1884) \\
        \texttt{DR2} & -3527 (-4663, -2392) & 317 (-407, 1042) \\
        \texttt{DR3} & -3548 (-6036, -1061) & 260 (-246, 767) \\
        \texttt{Outcome bridge} & -3646 (-4693, -2598) & 565 (-224, 1355) \\
        \texttt{Treatment bridge} & -3989 (-4373, -3605) & -532 (-1638, 574) \\
        \texttt{Treatment bridge2} & -3814 (-4941, -2688) & -205 (-1542, 1133) \\
        \texttt{Treatment bridge3} & -3895 (-6401, -1388) & 97 (-502, 695)
    \end{tabular}
    \caption{Estimate of the average treatment effect of PCV10 and placebo treatment on the number of hospitalizations due to all-cause pneumonia among children less than 12 months old in Brazil for various methods with 95\%-Wald confidence intervals. \texttt{Abadie's SC} \citep{Abadie2010} does not readily provide confidence intervals.}
    \label{tab: vaccine results}
\end{table}

\begin{figure}
    \centering
    \begin{subfigure}{0.495\textwidth}
    \includegraphics[width=\textwidth]{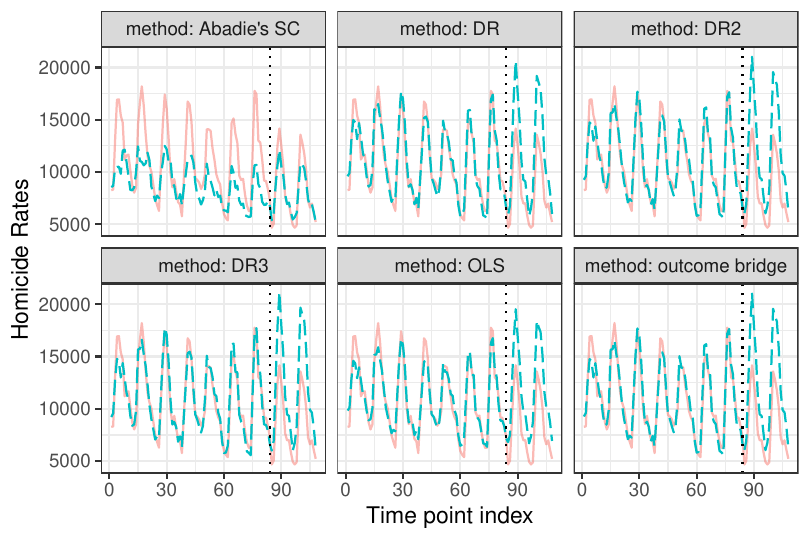}
    \caption{PCV10}
    \label{fig: vaccine main analysis}
    \end{subfigure}
    \hfill
    \begin{subfigure}{0.495\textwidth}
        \includegraphics[width=\textwidth]{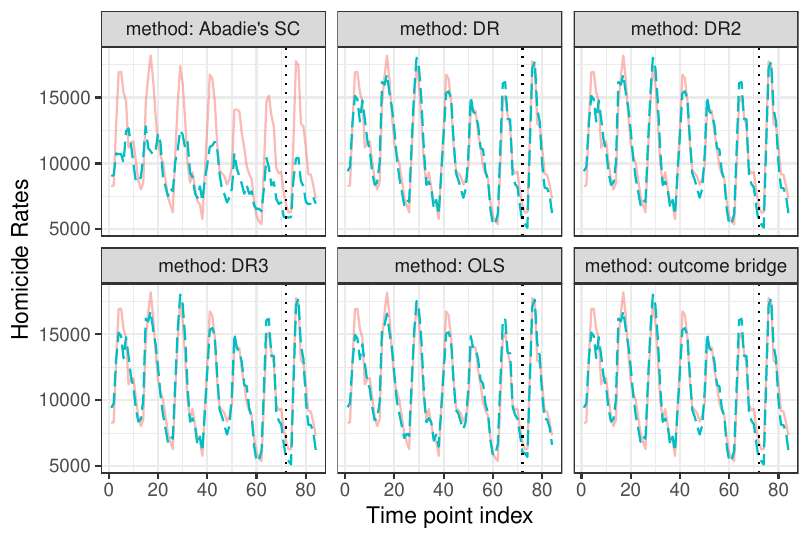}
        \caption{Placebo}
        \label{fig: vaccine placebo analysis}
    \end{subfigure}
    \caption{Trajectories of synthetic controls (green dashed) and the number of hospitalizations due to all-cause pneumonia (red solid) among children less than 12 months old in Brazil for various methods. The vertical line is the last time point (month) before implementing the treatment. Methods based on weighting only (\texttt{treatment bridge}, \texttt{treatment bridge2} and \texttt{treatment bridge3}) do not have synthetic control trajectories.}
    \label{fig: vaccine analysis}
\end{figure}

Although we have scaled all outcomes to fall in the unit interval, \texttt{Abadie's SC} does not output a good pre-treatment fit and its ATT estimate appears unreliable. Because the original scales of the outcomes across units differ substantially, the constraints in \texttt{Abadie's SC} (i) that the intercept vanishes, and (ii) that the weights are non-negative and sum to one, might not be appropriate in this application \citep{Doudchenko2016}; scaling to unit interval might still fail to justify the adequacy of these constraints. All other methods conclude a significant decrease in hospitalizations due to all-cause pneumonia after the introduction of PCV10, as expected. The estimate from \texttt{DR} is somewhat different from other proximal methods involving a treatment confounding bridge function with at least two units, suggesting a model misspecification affecting \texttt{DR} in finite samples. Though the theory suggests that \texttt{DR} should be consistent, when at least one nuisance function is correctly specified, a larger sample size $T$ than this data set might be needed for \texttt{DR} to be close to the truth.

We also conduct a falsification analysis. We consider a hypothetical placebo treatment in January 2009 and estimate its effect in the year 2009.
	The analysis results are presented in Table~\ref{tab: vaccine results} and Figure~\ref{fig: vaccine placebo analysis}. 
	Similarly to the main analysis, the pre-treatment fit from \texttt{Abadie's SC} appears unreliable in this setting; therefore confirming the recommendation that SCs might perform poorly when the pre-treatment fit is suboptimal \citep{Abadie2010}. Our proposed methods perform better and offer alternative justification for SC methods in such settings. In fact, 95\%-confidence intervals from most proximal SC methods cover zero, correctly indicating a non-significant effect due to the placebo treatment; so does \texttt{OLS} in this case. The only exception \texttt{DR} echoes the poor performance from the main analysis. We conclude that, in this application, when the nuisance functions are approximately correctly specified, the proximal methods are among the best.

We also study the effect of Florida's ``stand your ground'' law on homicide rates and the effect of a tax cut in Kansas on economic outcomes (see Web Appendix~\ref{sec: data analysis2}). In these analyses, our proposed doubly robust method outperforms \texttt{Abadie's SC} or \texttt{OLS} in several cases.

\section{Discussion}

To estimate the ATT in panel data settings, classical SC methods often require correct specification of nuisance functions. Our proposed doubly robust methods involve estimating two nuisance functions but allow for misspecification of one, without knowledge about which is misspecified. Our identification results may enable the development of new SC methods, especially nonparametric or semiparametric ones that rely on weaker conditions. Such methods can extend the idea of SCs and proximal causal inference to more general applications.

\section*{Acknowledgement}
This work was supported in part by the U.S. National Science Foundation and National Institutes of Health, and by Analytics at Wharton.
We are grateful to an anonymous reviewer for a counterexample of Condition~\ref{cond: treatment bridge}.

\clearpage

\setcounter{page}{1}
\setcounter{section}{0}
\renewcommand{\thesection}{S\arabic{section}}%
\setcounter{table}{0}
\renewcommand{\thetable}{S\arabic{table}}%
\setcounter{figure}{0}
\renewcommand{\thefigure}{S\arabic{figure}}%
\setcounter{equation}{0}
\renewcommand{\theequation}{S\arabic{equation}}%
\setcounter{condition}{0}
\renewcommand{\thecondition}{S\arabic{condition}}%
\setcounter{lemma}{0}
\renewcommand{\thelemma}{S\arabic{lemma}}%
\setcounter{theorem}{0}
\renewcommand{\thetheorem}{S\arabic{theorem}}%
\setcounter{corollary}{0}
\renewcommand{\thecorollary}{S\arabic{corollary}}%

\begin{center}
    \LARGE Supplementary Materials for ``Doubly Robust Proximal Synthetic Controls'' by Hongxiang Qiu, Xu Shi, Wang Miao, Edgar Dobriban, and Eric Tchetgen Tchetgen
\end{center}

\section{Completeness conditions} \label{sec: completeness}

The following completeness conditions are sufficient to ensure the uniqueness of the solutions to the equations \eqref{eq: identify h} and \eqref{eq: identify q} used to estimate confounding bridge functions $h^*$ and $q^*$. We emphasize that these conditions are not necessary for valid estimation of the average treatment effect for the treated unit, as shown in the simulation in Section~\ref{sec: sim2}.

\begin{condition} \label{cond: outcome complete}
	Let $g: \mathcal{W} \rightarrow \real$ be any square-integrable function. For all time points $t_- \leq T_0$ before treatment, the following two statements are equivalent: (i) $g(W_{t_-})=0$, (ii) $\cexpect \{ g(W_{t_-}) \mid Z_{t_-} \} = 0$.
\end{condition}

\begin{condition} \label{cond: treatment complete}
	Let $g: \mathcal{Z} \rightarrow \real$ be any square-integrable function. For all time points $t_- \leq T_0$ before treatment, the following two statements are equivalent: (i) $g(Z_{t_-})=0$, (ii) $\cexpect \{ g(Z_{t_-}) \mid W_{t_-} \} = 0$.
\end{condition}

We next discuss Condition~\ref{cond: outcome complete} in more detail; Condition~\ref{cond: treatment complete} is similar. Intuitively, Condition~\ref{cond: outcome complete} requires that $Z_t$ captures all variation in $W_t$, because any transformation $g(W_t)$ of $W_t$ whose regression on $Z_t$ is zero must be zero.
Equivalent conditions to Condition~\ref{cond: outcome complete} have been found. Lemma~2.1 in \cite{Severini2006} and Proposition~1 in \cite{Andrews2017} imply that Condition~\ref{cond: outcome complete} is equivalent to the following: 
for every nonconstant random variable $a(W_{t_-})$,
there exists a random variable $b(Z_{t_-})$ that is correlated with $a(W_{t_-})$. According to this equivalent condition, it is evident that the completeness condition \ref{cond: outcome complete} can be interpreted as that any variation in $W_{t_-}$ can be captured by a variation in $Z_{t_-}$.
We refer readers to Supplemental Section~D in \cite{Ying2021} for more details on completeness conditions.

\cite{Shi2021} recommended measuring a rich set of proxies $Z_t$ to make this condition plausible.
Similarly, to make Condition~\ref{cond: treatment complete} plausible, we recommend measuring a rich set of proxies $W_t$.
In the common setting where both $W_t$ and $Z_t$ are control units' outcomes, we thus recommend measuring contemporary outcomes in many control units and using many of them as proxies, either in $W_t$ or $Z_t$, to make Conditions~\ref{cond: outcome complete} and \ref{cond: treatment complete} plausible.

\section{Technical details and formal results for generalized method of moments} \label{sec: technical conditions}

We first describe our estimation method based on GMM.
Let
$g_h: \mathcal{Z} \rightarrow \real^{d_\alpha'}$, with
$d_\alpha' \geq d_\alpha$, and 
$g_q: \mathcal{W} \rightarrow \real^{d_\beta'}$,
with $d_\beta' \geq d_\beta$, be two user-specified functions.
In practice, these functions can be
chosen from known classes of basis 
functions used in non-parametric estimation. 
For instance, they can be a truncated polynomial basis \citep{Hastie2009}: with the support $\mathcal{Z}$ of $Z_t$ being $\real^{d_Z}$, $g_h(z)$ may consist of monomials up to order $K$, namely $(\prod_{k=1}^K z_{i_k}^{j_k})$ with $i_k \in \{1,\ldots,d_Z\}$, $j_k$ being non-negative integers, and $\sum_{k=1}^\alpha j_k \leq K$, such that the dimension of the output of $g_h$ is at least $\alpha$.
Let $\alpha \in \mathcal{A}$, $\beta \in \mathcal{B}$, $\lambda \in \Lambda$, $\psi \in \Psi \subseteq \real^{d_{\beta'}}$, and $\psi_- \in \Psi_- \subseteq \real$ be GMM parameters, $\theta = (\alpha,\beta,\lambda,\psi,\psi_-)$ be the collection of parameters, and $\Theta = \mathcal{A} \times \mathcal{B} \times \Lambda \times \Psi \times \Psi_-$ be the parameter space. Define the moment equation $G_t:\Theta \to \real^{2d_\beta'+3}$ as
the vertical stacking of
\begin{align*}
	& G_{t1}: \theta \mapsto \ind(t \leq T_0) \left[ \{ Y_t-h_\alpha(W_t) \} g_h(Z_t) \right], \quad G_{t2}: \theta \mapsto \ind(t > T_0) \left\{ \psi - g_q(W_t) \right\} \\
	& G_{t3}: \theta \mapsto \ind(t \leq T_0) \left\{ q_\beta(Z_t) g_q(W_t) - \psi \right\}, \quad G_{t4}: \theta \mapsto
	\ind(t>T_0) \left[ \phi_\lambda(t) - \{ Y_t - h_\alpha(W_t) \} + \psi_- \right] \\
	& \text{and} \quad G_{t5}: \theta \mapsto \ind(t \leq T_0) \left[ \psi_- - q_\beta(Z_t) \{ Y_t - h_\alpha(W_t) \} \right].
\end{align*}
Let $\Omega_T$ be a user-specified symmetric positive semi-definite $(d_\alpha'+2d_\beta'+2) \times (d_\alpha'+2d_\beta'+2)$ matrix, for example, the identity. Consider the GMM estimator
\begin{align}
    &\hat{\theta}_T = \left( \hat{\alpha}_T,\hat{\beta}_T,\hat{\lambda}_T,\hat{\psi}_T,,\hat{\psi}_{-,T} \right) = \argmin_{\theta \in \Theta} \left\{ \frac{1}{T} \sum_{t=1}^T G_t(\theta)\right\}^\top \Omega_T \left\{ \frac{1}{T} \sum_{t=1}^T G_t(\theta) \right\}. \label{eq: thetahat}
\end{align}
We use $\phi_{\hat{\lambda}_T}$ as the doubly robust estimator of the ATT $\phi^*$.

\begin{remark}
    Since \eqref{eq: DR identify ATT} and \eqref{eq: implied moment equation} involve expectations over different time periods, directly using these equations for estimation leads to moment equations that might not have mean zero for each time period $t$, but only when averaged over all time periods. 
This may invalidate the corresponding GMM procedure. 
To circumvent this, 
we introduce two centering parameters $\psi$ and $\psi_-$ to decouple different time periods in \eqref{eq: DR identify ATT} and \eqref{eq: implied moment equation}. 
The second moment equation in \eqref{eq: implied moment equation} is decoupled via the introduction of $\psi$ and represented by the second and third elements in $G_t$. Further, \eqref{eq: DR identify ATT} is decoupled via the introduction of $\psi_-$ and split into the fourth and fifth elements in $G_t$.
\end{remark}

We next present the additional technical conditions we rely on for the consistency and asymptotic normality results for the generalized method of moments estimators. We put conditions and their explanation under stationarity of confounders (Condition~\ref{cond: posttreatment identical distribution}) in the text and those without stationarity in square brackets.

\begin{condition} \label{cond: positive weight matrix}
    For the doubly robust method, $\Omega_T$ [$\undertilde{\Omega}_T$] converges in probability to $\Omega$ [$\undertilde{\Omega}$] for some fixed symmetric positive-definite matrix $\Omega$ [$\undertilde{\Omega}$]; for the weighting method, $\Omega^q_T$ [$\undertilde{\Omega}^q_T$] converges in probability to $\Omega^q$ [$\undertilde{\Omega}^q$] for some fixed symmetric positive-definite matrix $\Omega^q$ [$\undertilde{\Omega}^q$].
\end{condition}

Condition~\ref{cond: positive weight matrix} can be easily satisfied by taking $\Omega_T$ or $\Omega^q_T$ [$\undertilde{\Omega}_T$ or $\undertilde{\Omega}^q_T$] as a fixed positive-definite matrix.

\begin{condition} \label{cond: param compact}
For the doubly robust method, $\Theta$ [$\undertilde{\Theta}$] is compact; for the weighting method, $\Theta^q$ [$\undertilde{\Theta}^q$] is compact.
\end{condition}

Condition~\ref{cond: param compact} is a regularity compactness condition to establish theoretical results. In computations, the parameter space is often taken to be unrestricted.

\begin{condition} \label{cond: reg conditions for estimating function}
For the doubly robust method, $\cexpect\{G_t(\theta)\}$ [$\lim_{T \rightarrow \infty} \sum_{t=1}^T \cexpect\{\undertilde{G}_{T,t}(\undertilde{\theta})\}/T$] exists and is finite for all $\theta \in \Theta$ [$\undertilde{\theta} \in \undertilde{\Theta}$]; $G_t$ [$\undertilde{\theta} \mapsto \lim_{T \rightarrow \infty} \sum_{t=1}^T \cexpect\{\undertilde{G}_{T,t}(\undertilde{\theta})\}/T$] is continuous; for the weighting method, $\cexpect\{G^q_t(\theta^q)\}$ [$\lim_{T \rightarrow \infty} \sum_{t=1}^T \cexpect\{\undertilde{G}^q_{T,t}(\undertilde{\theta}^q)\}/T$] exists and is finite for all $\theta^q \in \Theta^q$ [$\undertilde{\theta}^q \in \undertilde{\Theta}^q$]; $G^q_t$ [$\undertilde{\theta}^q \mapsto \lim_{T \rightarrow \infty} \sum_{t=1}^T \cexpect\{\undertilde{G}^q_{T,t}(\undertilde{\theta}^q)\}/T$] is continuous.
\end{condition}

Condition~\ref{cond: reg conditions for estimating function} is also a regularity condition, which ensures that the population expectation of the moment equation exists.

\begin{condition} \label{cond: full rank matrix}
For the doubly robust method, $\rank(R)=d_\alpha+d_\beta+1$ [$\rank(\undertilde{R})=d_\alpha+d_\beta+1$]; for the weighting method, $\rank(R_q)=d_\beta+1$ [$\rank(\undertilde{R}_q)=d_\beta+1$].
\end{condition}

Condition~\ref{cond: full rank matrix} is a full-rank condition to ensure that the parameter in the  is locally identified.

We next list uniform weak laws of large numbers and central limit theorems as Conditions~\ref{cond: uniform law of large numbers}, \ref{cond: uniform law of large numbers derivative} and \ref{cond: CLT weighted moment}.
In turn, these hold under various sets of assumptions.

\begin{condition} \label{cond: uniform law of large numbers}
For the doubly robust method, with 
$\bar V_{T'}(\theta) = \sum_{t=1}^{T'} \cexpect\{\undertilde{G}_{T',t}(\theta)\}/T'$,
$$\sup_{\theta \in \Theta} \left\| \frac{1}{T} \sum_{t=1}^T G_t(\theta) - \lim_{T' \rightarrow \infty} \frac{1}{T'} \sum_{t=1}^{T'} \cexpect\{G_t(\theta)\} \right\|, \qquad  
\left[\sup_{\undertilde{\theta} \in \undertilde{\Theta}} \left\| \frac{1}{T} \sum_{t=1}^T G_t(\theta) - \lim_{T' \rightarrow \infty} \bar V_{T'}(\undertilde{\theta}) \right\|\right] $$
converge to zero in probability. 
For the weighting method,
$$\sup_{\theta^q \in \Theta^q} \left\| \frac{1}{T} \sum_{t=1}^T G^q_t(\theta^q) - \lim_{T' \rightarrow \infty} \frac{1}{T'} \sum_{t=1}^{T'} \cexpect\{G^q_t(\theta^q)\} \right\|,
$$
$$\left[ \sup_{\undertilde{\theta}^q \in \undertilde{\Theta}^q} \left\| \frac{1}{T} \sum_{t=1}^T \undertilde{G}^q_{T,t}(\undertilde{\theta}^q) - \lim_{T' \rightarrow \infty} \frac{1}{T'} \sum_{t=1}^{T'} \cexpect\{\undertilde{G}^q_{T',t}(\undertilde{\theta}^q)\} \right\| \right]$$
converge to zero in probability.
\end{condition}

\begin{condition} \label{cond: uniform law of large numbers derivative}
For the doubly robust method, $\theta_\infty$ [$\undertilde{\theta}_\infty$] is in the interior of $\Theta$ [$\undertilde{\Theta}$], $G_t$ is continuously differentiable [$\undertilde{G}_{T,t}$ is continuously differentiable and the derivative of $\sum_{t=1}^T \undertilde{G}_{T,t}/T$ is uniformly bounded over all $T \in \{1,2,\ldots\}$], and
\begin{align*}
    & \sup_{\theta \in \Theta} \left\| \frac{1}{T} \sum_{t=1}^T \nabla_{\theta} G_t(\theta) - \lim_{T' \rightarrow \infty} \frac{1}{T'} \sum_{t=1}^{T'} \cexpect \left\{ \nabla_{\theta} G_t(\theta) \right\} \right\| \\
    & \left[ \sup_{\undertilde{\theta} \in \undertilde{\Theta}} \left\| \frac{1}{T} \sum_{t=1}^T \nabla_{\undertilde{\theta}} \undertilde{G}_{T,t}(\undertilde{\theta}) - \lim_{T' \rightarrow \infty} \frac{1}{T'} \sum_{t=1}^{T'} \cexpect \left\{ \nabla_{\undertilde{\theta}} \undertilde{G}_{T',t}(\undertilde{\theta}) \right\} \right\| \right]
\end{align*}
converges to zero in probability; for the weighting method, $\theta^q_\infty$ [$\undertilde{\theta}^q_\infty$] lies in the interior of $\Theta^q$ is continuously differentiable [$\undertilde{G}^q_{T,t}$ is continuously differentiable and the derivative of $\sum_{t=1}^T \undertilde{G}^q_{T,t}/T$ is uniformly bounded over all $T \in \{1,2,\ldots\}$], and
\begin{align*}
    & \sup_{\theta^q \in \Theta^q} \left\| \frac{1}{T} \sum_{t=1}^T \nabla_{\theta^q} G^q_t(\theta^q) - \lim_{T' \rightarrow \infty} \frac{1}{T'} \sum_{t=1}^{T'} \cexpect \left\{ \nabla_{\theta^q} G^q_t(\theta^q) \right\} \right\| \\
    & \left[ \sup_{\undertilde{\theta}^q \in \undertilde{\Theta}^q} \left\| \frac{1}{T} \sum_{t=1}^T \nabla_{\undertilde{\theta}^q} \undertilde{G}^q_{T,t}(\undertilde{\theta}^q) - \lim_{T' \rightarrow \infty} \frac{1}{T'} \sum_{t=1}^{T'} \cexpect \left\{ \nabla_{\undertilde{\theta}^q} \undertilde{G}^q_{T',t}(\undertilde{\theta}^q) \right\} \right\| \right].
\end{align*}
converges to zero in probability.
\end{condition}

\begin{condition} \label{cond: CLT weighted moment}
For the doubly robust method,
$$T^{-1/2} \sum_{t=1}^T R^\top \Omega G_t(\theta_\infty) \qquad \left[ T^{-1/2} \sum_{t=1}^T \undertilde{R}^\top \undertilde{\Omega} \undertilde{G}_{T,t}(\undertilde{\theta}_\infty) \right]$$
is asymptotically $\mathrm{N}(0,B)$ [$\mathrm{N}(0,\undertilde{B})$]; for the weighting method,
$$T^{-1/2} \sum_{t=1}^T R_q^\top \Omega^q G^q_t(\theta^q_\infty) \qquad \left[ T^{-1/2} \sum_{t=1}^T \undertilde{R}_q^\top \undertilde{\Omega}^q \undertilde{G}^q_{T,t}(\undertilde{\theta}^q_\infty) \right]$$
is asymptotically $\mathrm{N}(0,B_q)$ [$\mathrm{N}(0,\undertilde{B}_q)$].
\end{condition}

In Condition~\ref{cond: uniform law of large numbers derivative}, the assumption that the probability limit of the estimator lies in the interior of the parameter space and the differentiability of the moment equation are usually satisfied under stationarity (Condition~\ref{cond: posttreatment identical distribution}). Without stationarity, we require a slightly stronger version of differentiability, namely differentiability of the average moment equation with uniformly bounded derivative. This difference appears a necessary (and low) price for non-stationairty to establish asymptotic normality in the generalized method of moments.

The uniform weak laws of large numbers, Conditions~\ref{cond: uniform law of large numbers} and \ref{cond: uniform law of large numbers derivative}, can be verified by using, for example, Theorem~5.1 and 5.2 in \cite{Potscher1997}, Theorems~1 and 2 in \cite{Andrews1988}, or Corollary~1.4 in \cite{VanHandel2013}, among others. For example, if the underlying process is strictly stationary and the involved function class is uniformly bounded with a finite bracketing number for every size of the brackets, then Corollary~1.4 in \cite{VanHandel2013} implies Conditions~\ref{cond: uniform law of large numbers} and \ref{cond: uniform law of large numbers derivative}. 
Uniform laws of large numbers can also be obtained without stationarity assumptions for strongly mixing (e.g., $\alpha$-mixing) processes, or $\phi$-mixing processes \citep[see, e.g., Chapter~5 in][]{Potscher1997}.

The central limit theorem, Condition~\ref{cond: CLT weighted moment}, can be verified by using, for example, results in \cite{Rosenblatt1956}, Theorems~5 and 6 in \cite{Philipp1969}, Corollary~2.11 in \cite{McLeish1977}, Corollary~1 in \cite{Herrndorf1984}, Theorem~3.6 in \cite{Davidson1992}, Theorem~2.1 and Corollary~2.1 in \cite{Arcones1994}, Theorems~10.1 and 10.2 in \cite{Potscher1997}, or Theorem~3.23 in \cite{Dehling2002}, Theorem~1.1 and 6.1 in \cite{Bradley2017}, among others. For example, define $\| A \|_q = (\cexpect |A|^q)^{1/q}$ for any random variable $A$. For the doubly robust method, suppose that the underlying process is strongly mixing with coefficients $\alpha(k)$, $k\ge 1$ \citep{Rosenblatt1956}. If, for some $s \in (2,\infty]$ and with $r=2/s$, $\sum_{k=1}^\infty \alpha(k)^{1-r} < \infty$ and $\limsup_{t \rightarrow \infty} \| R^\top \Omega \nabla_{\theta} G_t(\theta) |_{\theta=\theta_\infty} \|_s < \infty$, then Condition~\ref{cond: CLT weighted moment} holds by Corollary~1 in \cite{Herrndorf1984}.

\section{Relaxing stationarity of confounders} \label{sec: nonstationary}

\subsection{Estimand and identification} \label{sec: identify ATT treatment nonstationary}

As mentioned in Section~\ref{sec: identify ATT}, stationarity of confounders, namely Condition~\ref{cond: posttreatment identical distribution}, may be too strong in certain conditions.
In this section, we drop this condition and present more general identification results. 
We will generally add a tilde symbol under various symbols for objects similar to those from Section~\ref{sec: identify ATT}.

Without Condition~\ref{cond: posttreatment identical distribution}, it is challenging to identify the average treatment effect for the treated unit $\phi^*_{t_+}$ at each post-treatment time period $t_+$. We therefore consider another causal estimand, the time-averaged average treatment effect for the treated unit. 
We first consider the case where $T$ and $T_0$ are fixed. Let $\ell_T: t_+ \mapsto \ell_T(t_+) \geq 0$ be a user-specified importance weighting function that satisfies $\sum_{t_+=T_0+1}^T \ell_T(t_+) = 1$, and define $\undertilde{\phi}^*_T$ to be $\sum_{t_+=T_0+1}^T \ell_T(t_+) \phi^*(t_+) = \sum_{t_+=T_0+1}^T \ell_T(t_+) \cexpect \{ Y_{t_+}(1) - Y_{t_+}(0) \}$. Thus, $\undertilde{\phi}^*_T$ is a weighted average of time-specific ATTs $\phi^*_{t_+}$ over post-treatment time periods with averaging weights given by $\ell_T(\cdot)$. We refer to $\ell_T$ as an importance weight because it encodes the importance of each post-treatment time period in the average treatment effect for the treated unit $\undertilde{\phi}^*_T$ averaged over time. In the notation, we suppress the dependence of $\undertilde{\phi}^*_T$ on $\ell_T$ and $T_0$ for conciseness.

Our approach to this problem via weighting is similar to the case with stationarity in Section~\ref{sec: estimate ATT}. We again use a treatment bridge function to capture the covariate shift of  confounders between pre- and post-treatment periods and subsequently use this function to weight pre-treatment observations to impute post-treatment outcomes. A key difference is that the treatment bridge function in this case also needs to incorporate the importance weight $\ell_T(\cdot)$ and therefore is a proxy of a Radon-Nikodym derivative between two mixtures of distributions.
Specifically, we rely on the following condition on the existence of a treatment confounding bridge function $\undertilde{q}^*$.

\begin{condition} \label{cond: treatment bridge nonstationary}
    There exists a function $\undertilde{q}^*_T : \mathcal{Z} \rightarrow \real$ such that, for any square-integrable function $g: \mathcal{U} \rightarrow \real$:
    \begin{equation}
        \frac{1}{T_0} \sum_{t_-=1}^{T_0} \cexpect \{ \undertilde{q}^*_T(Z_{t_-}) g(U_{t_-}) \} = \sum_{t_+=T_0+1}^T \ell_T(t_+) \cexpect \{ g(U_{t_+}) \}. \label{eq: define q nonstationary}
    \end{equation}
\end{condition}

We have dropped the dependence of the treatment confounding bridge function $\undertilde{q}^*_T$ on the importance weighting function $\ell_T$ from the notation for conciseness. Condition~\ref{cond: treatment bridge nonstationary}, namely \eqref{eq: define q nonstationary}, might appear non-intuitive. We next rewrite \eqref{eq: define q nonstationary} in an integral form and argue that $\undertilde{q}^*$ can be interpreted as a proxy of a Radon-Nikodym derivative, similarly to $q^*$ in Condition~\ref{cond: treatment bridge}. Let $\mathscr{U}_T$ be the uniform law over pre-treatment time periods $\{1,\ldots,T_0\}$ and $\mathscr{L}_T$ be the law over post-treatment time periods $\{T_0+1,T_0+2,\ldots,T\}$ defined by $\mathscr{L}_T(\{t_+\}) = \ell_T(t_+)$ for $t_+ > T_0$. With these definitions, \eqref{eq: define q nonstationary} is equivalent to
$$\int_{\{1,\ldots,T\}} \cexpect \{ \undertilde{q}^*_T(Z_{t_-}) g(U_{t_-}) \} \mathscr{U}_T(\intd t_-) = \int_{\{T_0+1,\ldots,T\}} \cexpect \{ g(U_{t_+}) \} \mathscr{L}_T(\intd t_+),$$
which is also equivalent to
$$\int_{\{1,\ldots,T\}} \cexpect[ \cexpect \{ \undertilde{q}^*_T(Z_{t_-}) \mid U_{t_-} \} g(U_{t_-}) ] \mathscr{U}_T(\intd t_-) = \int_{\{T_0+1,\ldots,T\}} \cexpect \{ g(U_{t_+}) \} \mathscr{L}_T(\intd t_+).$$
Since $g$ is an arbitrary square-integrable function, 
clearly $u \mapsto \cexpect \{ \undertilde{q}^*_T(Z_{t_-}) \mid U_{t_-}=u \}$ may be viewed as a Radon-Nikodym derivative, of the law of $U_{t_+}$ ($t_+ \sim \mathscr{L}_T$),
with respect to the law of $U_{t_-}$ ($t_- \sim \mathscr{U}_T$), namely
$$u \mapsto \frac{\int_{\{T_0+1,\ldots,T\}} f_{U_{t_+}}(u) \mathscr{L}_T(\intd t_+)}{\int_{\{1,\ldots,T_0\}} f_{U_{t_-}}(u) \mathscr{U}_T(\intd t_-)},$$
where we have informally used $f_A$ to denote the density of a random variable $A$.

The next condition is a completeness condition that is similar to Condition~\ref{cond: treatment complete}.
\begin{condition} \label{cond: treatment complete nonstationary}
	Let $g: \mathcal{Z} \rightarrow \real$ be any square-integrable function. The following two statements are equivalent: (i) $g(Z_{t_-})=0$ for all $t_- \leq T_0$, (ii) $\sum_{t_-=1}^{T_0} \cexpect \{ g(Z_{t_-})$ $f(W_{t_-}) \}/T_0$ $= 0$ for any square-integrable function $f$.
\end{condition}
The second statement of Condition~\ref{cond: treatment complete nonstationary} 
is %
a generalization of
the second statement of Condition~\ref{cond: treatment complete}.
Indeed, the latter is equivalent to the following: $\cexpect \{ g(Z_{t_-}) f(W_{t_-}) \} = 0$ for any square-integrable function $f$.

With the above two conditions, we have the following identification result for $\undertilde{\phi}^*_T$ without stationarity for $U_{t_+}$ in Condition~\ref{cond: posttreatment identical distribution}.

\begin{theorem}[Identification of average treatment effect with $\undertilde{q}^*_T$] \label{thm: identify ATT treatment nonstationary}
	Let $f:\real \rightarrow \real$ be any square-integrable function. Under Conditions~\ref{cond: proxy independence}, \ref{cond: covariate shift}, and \ref{cond: treatment bridge nonstationary}, it holds that
	\begin{equation}
	    \sum_{t_+=T_0+1}^T \ell_T(t_+) \cexpect \{ f(Y_{t_+}(0)) \} = \frac{1}{T_0} \sum_{t_-=1}^{T_0} \cexpect \{ \undertilde{q}^*_T(Z_{t_-}) f(Y_{t_-}) \}.
	    \label{eq: identify EfY0 treatment nonstationary}
	\end{equation}
	In particular, taking $f$ to be the identity function, it holds that
	\begin{equation}
        \sum_{t_+=T_0+1}^T \ell_T(t_+) \cexpect \{ Y_{t_+}(0) \} = \frac{1}{T_0} \sum_{t_-=1}^{T_0} \cexpect \{ \undertilde{q}^*_T(Z_{t_-}) Y_{t_-} \}.
		\label{eq: identify EY0 treatment nonstationary}
	\end{equation}
	and thus
	$\undertilde{\phi}^*_T$
	is identified as
	$\undertilde{\phi}^*_T=\sum_{t_+=T_0+1}^T \cexpect \{ Y_{t_+} \}- \sum_{t_-=1}^{T_0} \cexpect \{ q^*(Z_{t_-}) Y_{t_-} \}/T_0$. In addition, the treatment confounding bridge function $\undertilde{q}^*_T$ is a solution to
	\begin{equation}
		\frac{1}{T_0} \sum_{t_-=1}^{T_0} \cexpect \{ \undertilde{q}(Z_{t_-}) g(W_{t_-}) \} = \sum_{t_+=T_0+1}^T \ell_T(t_+) \cexpect \{ g(W_{t_+}) \} \quad \text{for all square-integrable $g$}
		\label{eq: identify q nonstationary}
	\end{equation}
	in
    $\undertilde{q}: \mathcal{Z} \rightarrow \real$.
    Further, under Condition~\ref{cond: treatment complete nonstationary}, \eqref{eq: identify q nonstationary} has a unique solution almost surely.
\end{theorem}

We also have the following doubly robust identification result similar to Theorem~\ref{thm: identify ATT DR}.

\begin{theorem}[Doubly robust identification with $h^*$ and $\undertilde{q}^*_T$] \label{thm: identify ATT DR nonstationary}
	Let
	$h: \mathcal{W} \rightarrow \real$
	and
	$\undertilde{q}: \mathcal{Z} \rightarrow \real$
	be any square-integrable functions. Under Conditions~\ref{cond: proxy independence} and \ref{cond: covariate shift}, 
	 if either (i) Condition~\ref{cond: outcome bridge} holds and $h=h^*$,
	or (ii) Condition~\ref{cond: treatment bridge nonstationary} holds and $\undertilde{q}=\undertilde{q}^*_T$,
	then 
	\begin{align}
	        &\sum_{t_+=T_0+1}^T \ell_T(t_+) \cexpect[Y_{t_+}(0)] = \frac{1}{T_0} \sum_{t_-=1}^{T_0} \cexpect \left[ \undertilde{q}(Z_{t_-}) \{Y_{t_-} - h(W_{t_-})\} \right] + \sum_{t_+=T_0+1}^T \ell_T(t_+) \cexpect \{ h(W_{t_+}) \}, \nonumber\\
	        &\undertilde{\phi}^*_T = \sum_{t_+=T_0+1}^T \ell_T(t_+) \cexpect \left\{ Y_{t_+} - h(W_{t_+}) \right\} - \frac{1}{T_0} \sum_{t_-=1}^{T_0} \cexpect \left[ \undertilde{q}(Z_{t_-}) \{Y_{t_-} - h(W_{t_-})\} \right]. \label{eq: DR identify ATT nonstationary}
	\end{align}
\end{theorem}

Since our estimators and associated theoretical results are asymptotic as $T \rightarrow \infty$, we also present identification results under the same asymptotic regime. The causal estimand of interest is $\undertilde{\phi}^*$ defined as $\lim_{T \rightarrow \infty} \undertilde{\phi}^*_T = \lim_{T \rightarrow \infty} \sum_{t_+=T_0+1}^T \ell_T(t_+) \cexpect \{ Y_{t_+}(1)-Y_{t_+}(0) \}$. The associated identification results are similar to Theorem~\ref{thm: identify ATT treatment nonstationary} with the key difference being the limit as $T \rightarrow \infty$, so we abbreviate our presentation.
Recall that we assume that $T_0/T \rightarrow \gamma \in (0,1)$ as $T \rightarrow \infty$.

\begin{condition} \label{cond: treatment bridge nonstationary limit}
    There exists a function $\undertilde{q}^*: \mathcal{Z} \rightarrow \real$ such that, for any square-integrable function $g: \mathcal{U} \rightarrow \real$, it holds that
    \begin{equation}
        \lim_{T \rightarrow \infty} \frac{1}{T_0} \sum_{t_-=1}^{T_0} \cexpect \{ \undertilde{q}^*(Z_{t_-}) g(U_{t_-}) \} = \lim_{T \rightarrow \infty} \sum_{t_+=T_0+1}^T \ell_T(t_+) \cexpect \{ g(U_{t_+}) \}. \label{eq: define q nonstationary limit}
    \end{equation}
\end{condition}

\begin{condition} \label{cond: treatment complete nonstationary limit}
	Let $g: \mathcal{Z} \to  \real$ be any square-integrable function. The following two statements are equivalent: (i) $g(Z_{t})=0$ for all $t \in \{1,2,\ldots\}$, (ii) $\lim_{T \rightarrow \infty}\sum_{t_-=1}^{T_0} \cexpect \{ g(Z_{t_-})$ $f(W_{t_-}) \}/T_0 = 0$ for any square-integrable function $f$.
\end{condition}

\begin{theorem}[Identification of average treatment effect with $\undertilde{q}^*$] \label{thm: identify ATT treatment nonstationary limit}
	Let $f:\real \rightarrow \real$ be any square-integrable function. Under Conditions~\ref{cond: proxy independence}, \ref{cond: covariate shift}, and \ref{cond: treatment bridge nonstationary limit}, it holds that
	\begin{equation}
	    \lim_{T \rightarrow \infty} \sum_{t_+=T_0+1}^T \ell_T(t_+) \cexpect \{ f(Y_{t_+}(0)) \} = \lim_{T \rightarrow \infty} \frac{1}{T_0} \sum_{t_-=1}^{T_0} \cexpect \{ \undertilde{q}^*(Z_{t_-}) f(Y_{t_-}) \}.
	    \label{eq: identify EfY0 treatment nonstationary limit}
	\end{equation}
	In particular, taking $f$ to be the identity function, it holds that
	\begin{equation}
        \lim_{T \rightarrow \infty} \sum_{t_+=T_0+1}^T \ell_T(t_+) \cexpect \{ Y_{t_+}(0) \} = \lim_{T \rightarrow \infty} \frac{1}{T_0} \sum_{t_-=1}^{T_0} \cexpect \{ \undertilde{q}^*(Z_{t_-}) Y_{t_-} \}.
		\label{eq: identify EY0 treatment nonstationary limit}
	\end{equation}
	and thus
	$\undertilde{\phi}^*=\lim_{T \rightarrow \infty} \left[ \sum_{t_+=T_0+1}^T \cexpect \{ Y_{t_+} \} - \sum_{t_-=1}^{T_0} \cexpect \{ q^*(Z_{t_-}) Y_{t_-} \}/T_0 \right]$. In addition, the treatment confounding bridge function $\undertilde{q}^*_T$ is a solution to
    \begin{align}
        \begin{split}
            &\lim_{T \rightarrow \infty} \frac{1}{T_0} \sum_{t_-=1}^{T_0} \cexpect \{ \undertilde{q}(Z_{t_-}) g(W_{t_-}) \} = \lim_{T \rightarrow \infty} \sum_{t_+=T_0+1}^T \ell_T(t_+) \cexpect \{ g(W_{t_+}) \} \\
            &\hspace{1in} \text{for all square-integrable function $g$}
        \end{split}\label{eq: identify q nonstationary limit}
    \end{align}
	in
    $\undertilde{q}: \mathcal{Z} \rightarrow \real$.
    Further, under Condition~\ref{cond: treatment complete nonstationary limit}, \eqref{eq: identify q nonstationary limit} has a unique solution almost surely.
\end{theorem}

\begin{theorem}[Doubly robust identification with $h^*$ and $\undertilde{q}^*$] \label{thm: identify ATT DR nonstationary limit}
	Let
	$h: \mathcal{W} \rightarrow \real$
	and
	$\undertilde{q}: \mathcal{Z} \rightarrow \real$
	be any square-integrable functions. Under Conditions~\ref{cond: proxy independence} and \ref{cond: covariate shift}, 
if either (i) Condition~\ref{cond: outcome bridge} holds and $h=h^*$,
	or (ii) Condition~\ref{cond: treatment bridge nonstationary limit} holds and $\undertilde{q}=\undertilde{q}^*$, then
    \begin{align}
	        &\lim_{T \rightarrow \infty} \sum_{t_+=T_0+1}^T \ell_T(t_+) \cexpect \{ Y_{t_+}(0) \} \nonumber \\
         &= \lim_{T \rightarrow \infty} \left( \frac{1}{T_0} \sum_{t_-=1}^{T_0} \cexpect \left[ \undertilde{q}(Z_{t_-}) \{Y_{t_-} - h(W_{t_-})\} \right] + \sum_{t_+=T_0+1}^T \ell_T(t_+) \cexpect \{ h(W_{t_+}) \} \right), \nonumber\\
	        &\undertilde{\phi}^*_T = \lim_{T \rightarrow \infty} \left( \sum_{t_+=T_0+1}^T \ell_T(t_+) \cexpect \left\{ Y_{t_+} - h(W_{t_+}) \right\} - \frac{1}{T_0} \sum_{t_-=1}^{T_0} \cexpect \left[ \undertilde{q}(Z_{t_-}) \{Y_{t_-} - h(W_{t_-})\} \right] \right). \label{eq: DR identify ATT nonstationary limit}
	\end{align}
\end{theorem}

\begin{remark}
    We use causal conditions on the treatment confounding bridge that connect to the unobserved confounder $U_t$, namely Conditions~\ref{cond: treatment bridge}, \ref{cond: treatment bridge nonstationary} and \ref{cond: treatment bridge nonstationary limit}. With such conditions, we do not rely on the existence of the outcome confounding bridge function in Condition~\ref{cond: outcome bridge} to obtain the doubly robust identification formulas \eqref{eq: DR identify ATT}, \eqref{eq: DR identify ATT nonstationary} and \eqref{eq: DR identify ATT nonstationary limit}. 
    An alternative set of conditions that leads to the same identification formulas is the existence of the treatment confounding bridge function in terms of $W_t$ (namely \eqref{eq: identify q}, \eqref{eq: identify q nonstationary} and \eqref{eq: identify q nonstationary limit}) as well as the existence of the outcome confounding bridge function in Condition~\ref{cond: outcome bridge}. Thus, these two sets of causal conditions do not imply each other. 
    One potential drawback of the alternative set of conditions is that, for example, to identify $\cexpect[f \{ Y_{t_+}(0) \}]$ for some given function $f$ under stationarity of $U_{t_+}$ (Condition~\ref{cond: posttreatment identical distribution}), a different outcome confounding bridge function that can identify the same causal estimand is required to exist, but the approach we present in this paper does not require this (see \eqref{eq: identify EfY0 treatment}, \eqref{eq: identify EfY0 treatment nonstationary} and \eqref{eq: identify EfY0 treatment nonstationary limit}). Identification of $\cexpect[f\{ Y_{t_+}(0) \}]$ can be useful, for example, in constructing prediction intervals for the actual treatment effect for the treated unit. Since our main focus is the average treatment effect for the treated unit and both approaches lead to the same nonparametric identification formulas and therefore the same estimation procedures, we do not present this alternative approach in detail.
\end{remark}

\subsection{Doubly robust estimation} \label{sec: estimate ATT nonstationary}

In this section, we describe a doubly robust method to estimate the average treatment effect for the treated unit based on Theorem~\ref{thm: identify ATT DR nonstationary limit}, along with its theoretical properties. 
Due to the similarity between Theorems~\ref{thm: identify ATT DR} and \ref{thm: identify ATT DR nonstationary limit}, this method is also similar to that described in Section~\ref{sec: estimate ATT}.
We focus on the estimation of $\undertilde{\phi}^*$, the limit (as the number of time periods tends to infinity) of an average of time-specific average treatment effect for the treated unit $\phi^*_t$ over post-treatment time periods. We propose to use a generalized method of moments method that is similar to that described in Section~\ref{sec: estimate ATT}. We therefore abbreviate our presentation with emphasis on key differences.

We still use $\alpha$ and $\beta$ to parameterize $h^*$ and $\undertilde{q}^*$, and will use $h_\alpha$ and $\undertilde{q}_\beta$ to denote the models of confounding bridge functions with parameter $\alpha$ and $\beta$, respectively. Key changes of the method without stationarity of confounders are (i) that we no longer parameterize the average treatment effect for the treated unit with $\lambda$ since the estimand $\undertilde{\phi}^*$ is a scalar, and (ii) that we need to incorporate the user-specified importance weighting function $\ell_T$. We still let
$g_h: \mathcal{Z} \rightarrow \real^{d_\alpha'}$, with
$d_\alpha' \geq d_\alpha$, and 
$g_q: \mathcal{W} \rightarrow \real^{d_\beta'}$,
with $d_\beta' \geq d_\beta$, be two user-specified functions. Let $\alpha \in \mathcal{A}$, $\beta \in \mathcal{B}$, $\phi \in \Phi \subseteq \real$, $\psi \in \Psi \subseteq \real^{d_{\beta'}}$, and $\psi_- \in \Psi_- \subseteq \real$ be parameters in the generalized method of moments procedure, $\undertilde{\theta} = (\alpha,\beta,\phi,\psi,\psi_-)$ be the collection of parameters, and $\undertilde{\Theta} = \mathcal{A} \times \mathcal{B} \times \Phi \times \Psi \times \Psi_-$ be the parameter space. For each $t=1,\ldots,T$, define $\undertilde{G}_{T,t}:\undertilde{\Theta} \mapsto \real^{2d_\beta'+3}$ as
\begin{align}
    \undertilde{G}_{T,t}: \undertilde{\theta} \mapsto
    \quad \begin{pmatrix}
	\ind(t \leq T_0) \left[ \{Y_t-h_\alpha(W_t)\} g_h(Z_t) \right] \\
	\ind(t > T_0) \left\{ \psi - g_q(W_t) \right\} \\
	\ind(t \leq T_0) \left\{ \undertilde{q}_\beta(Z_t) g_q(W_t) - \psi \right\} \\
	\ind(t>T_0) \left[ \phi - (T-T_0) \ell_T(t) \{Y_t - h_\alpha(W_t)\} + \psi_- \right] \\
	\ind(t \leq T_0) \left[ \psi_- - \undertilde{q}_\beta(Z_t) \{Y_t - h_\alpha(W_t)\} \right]
\end{pmatrix}. \label{eq: Gt nonstationary}
\end{align}
Let $\undertilde{\Omega}_T$ be a user-specified symmetric positive semi-definite $(d_\alpha'+2d_\beta'+2) \times (d_\alpha'+2d_\beta'+2)$ matrix, for example, the identity matrix. Consider the generalized method of moments estimator
\begin{align}
    &\hat{\undertilde{\theta}}_T = \left( \hat{\undertilde{\alpha}}_T,\hat{\undertilde{\beta}}_T,\hat{\undertilde{\phi}}_T,\hat{\undertilde{\psi}}_T,,\hat{\undertilde{\psi}}_{-,T} \right) = \argmin_{\undertilde{\theta} \in \undertilde{\Theta}} 
    \left\{\frac{1}{T} \sum_{t=1}^T \undertilde{G}_{T,t}
    (\undertilde{\theta})\right\}^\top \undertilde{\Omega}_T 
    \left\{\frac{1}{T} \sum_{t=1}^T \undertilde{G}_{T,t}(\undertilde{\theta})\right\}. \label{eq: thetahat nonstationary}
\end{align}
The entry $\hat{\undertilde{\phi}}_T$ in the generalized method of moments estimator is the estimator of the average treatment effect for the treated unit $\undertilde{\phi}^*$ averaged over the post-treatment time periods. This generalized method of moments estimator coincides with the estimator $\phi_{\hat{\lambda}_T}(t)$ in Section~\ref{sec: estimate ATT} if the time-varying model $\phi_{\lambda}$ for the average treatment effect for the treated unit is a constant function $t \mapsto \lambda$ and the user-specified importance weight $\ell_T(t)=1/(T-T_0)$. This is the case in all our simulations in Sections~\ref{sec: sim} and \ref{sec: sim2}, as well as in our analysis of Kansas GDP in Section~\ref{sec: Kansas analysis}.

Similarly to Condition~\ref{cond: misspecify h or q}, we make the following regularity conditions to obtain consistency and valid inference about $\undertilde{\phi}^*$.

\begin{condition} \label{cond: misspecify h or q nonstationary}
	(i) There exists a unique $\undertilde{\theta}_\infty = (\undertilde{\alpha}_\infty,\undertilde{\beta}_\infty,\undertilde{\phi}_\infty,\undertilde{\psi}_\infty,\undertilde{\psi}_{-,\infty}) \in \undertilde{\Theta}$ such that $\lim_{T \rightarrow \infty} \sum_{t=1}^T \cexpect \{ \undertilde{G}_{T,t}(\undertilde{\theta}_\infty) \}/T=0$, (ii) $h^*=h_{\undertilde{\alpha}_\infty}$ satisfies Condition~\ref{cond: outcome bridge}, or $\undertilde{q}^*=\undertilde{q}_{\undertilde{\beta}_\infty}$ satisfies Condition~\ref{cond: treatment bridge nonstationary}.
\end{condition}

Though similar to Condition~\ref{cond: misspecify h or q}, Condition~\ref{cond: misspecify h or q nonstationary} is different.
In part~(i) of Condition~\ref{cond: misspecify h or q}, we require that the moment equation $G_t$ at the true parameter has mean zero for all time periods; that is, $\cexpect \{ G_t(\theta_\infty) \}=0$ for all $t=1,\ldots,T$. This is a standard assumption for the generalized method of moments, and therefore Theorem~\ref{thm: DR ATT asymptotic normality} follows immediately from standard theory for the generalized method of moments. In contrast, part~(i) of Condition~\ref{cond: misspecify h or q nonstationary} only requires that the average of the moment equation means over all time periods has approximately mean zero, namely $\lim_{T \rightarrow \infty} \sum_{t=1}^T \cexpect \{ \undertilde{G}_{T,t}(\undertilde{\theta}_\infty) \}/T=0$. However, the moment equation might not have mean zero for all time periods; that is, we do not require that $\cexpect \{ \undertilde{G}_{T,t}(\undertilde{\theta}_\infty) \} =0$ for all $t=1,\ldots,T$. This relaxed condition deviates from the standard generalized method of moments assumptions, and therefore the standard theory for the generalized method of moments does not apply. 

Another deviation from standard assumptions for the generalized method of moments is that the moment equation $\undertilde{G}_{t,T}$ depends on the sample size $T$. Therefore, when deriving the asymptotic normality of the estimator $\undertilde{\theta}_T$, in order to apply an argument based on Taylor series, we require a slightly different differentiability condition (see Condition~\ref{cond: uniform law of large numbers derivative} in Web Appendix~\ref{sec: technical conditions}) on the moment equation from standard conditions for the generalized method of moments. We obtain the following theoretical result for the above generalized method of moments method after carefully modifying the proof for the standard theory of the generalized method of moments.

\begin{theorem} \label{thm: DR ATT asymptotic normality nonstationary}
	Under Conditions~\ref{cond: proxy independence}, \ref{cond: covariate shift}, \ref{cond: misspecify h or q nonstationary}, \ref{cond: positive weight matrix}--\ref{cond: reg conditions for estimating function} and \ref{cond: uniform law of large numbers}, with the estimator $\hat{\undertilde{\theta}}_T$ from \eqref{eq: thetahat nonstationary} and $\undertilde{\theta}_\infty$ in Condition~\ref{cond: misspecify h or q nonstationary}, it holds that $\undertilde{\phi}_\infty = \undertilde{\phi}^*$ and, 
	as $T\to\infty$,
	$\hat{\undertilde{\theta}}_T$ is consistent for $\undertilde{\theta}_\infty$
	and
	$\hat{\undertilde{\phi}}_T$ is consistent for $\undertilde{\phi}^*$. 
	Additionally, under Conditions~\ref{cond: full rank matrix}, \ref{cond: uniform law of large numbers derivative} and \ref{cond: CLT weighted moment}, it holds that, as $T\to\infty$, $\sqrt{T} (\hat{\undertilde{\theta}}_T-\undertilde{\theta}_\infty)$ is asymptotically distributed as $\mathrm{N}(0,\undertilde{A}^{-1} \undertilde{B} \undertilde{A}^{-1})$,
	where, with $\undertilde{G}_{T,t}$ in \eqref{eq: Gt nonstationary} and $\undertilde{\Omega}$ being the probability limit of $\undertilde{\Omega}_T$ in Condition~\ref{cond: positive weight matrix}, we use the following quantities, whose existence follows based on our assumptions:
	\begin{align*}
		& \undertilde{R}=\lim_{T \rightarrow \infty} \frac{1}{T} \sum_{t=1}^T \cexpect \{ \nabla_{\undertilde{\theta}} \undertilde{G}_{T,t}(\undertilde{\theta}) |_{\undertilde{\theta}=\undertilde{\theta}_\infty} \}, 
		\qquad \undertilde{A} = \undertilde{R}^\top \undertilde{\Omega} \undertilde{R},
		\\
		& \undertilde{B} = \undertilde{R}^\top \undertilde{\Omega} \left[ \lim_{T \rightarrow \infty} \Var \left\{ T^{-1/2} \sum_{t=1}^T \undertilde{G}_{T,t}(\undertilde{\theta}_\infty) \right\} \right] \undertilde{\Omega} \undertilde{R}.
	\end{align*}
\end{theorem}

Theorem~\ref{thm: DR ATT asymptotic normality nonstationary} appears similar to Theorem~\ref{thm: DR ATT asymptotic normality}, but there is a key difference in their practical implications in data analysis. 
Under stationarity of confounders, the moment equation for $G_t$ has mean zero at the limiting parameter value $\theta_\infty$ at each time period (see Condition~\ref{cond: misspecify h or q}). In this case, the matrix $B$ in Theorem~\ref{thm: DR ATT asymptotic normality} can be consistently estimated with heteroskedasticity and autocorrelation consistent estimators. In contrast, without this stationarity, the moment equation $\undertilde{G}_{T,t}$ might not have mean zero at any parameter value at all time periods; only the average of the moment equation over all time periods is zero at the true parameter value (see Condition~\ref{cond: misspecify h or q nonstationary}).

Therefore, it is challenging, if possible at all, to consistently estimate $\undertilde{B}$ in Theorem~\ref{thm: DR ATT asymptotic normality nonstationary}, which is a key component of obtaining a consistent variance estimator. Indeed, standard generalized method of moments software outputs standard errors under the assumption that the moment equation has mean zero at all time periods, and the limit of the variance in the ``meat'' of $\undertilde{B}$ is typically estimated as a weighted average of squares of $\undertilde{G}_{T,t}$. This standard error is conservative, because $\cexpect(X^2) \geq \Var(X)$ for any random variable $X$. Therefore, implementing the generalized method of moments without assuming the stationarity would lead to conservative statistical inference.

\section{Estimation of average treatment effect based on weighting} \label{sec: weighting GMM}

In this section, we describe the generalized method of moments estimators of the average treatment effect for the treated unit $\phi^*(t)$ under stationarity of confounders and $\undertilde{\phi}^*$ without stationarity. Since the estimators are similar to the doubly robust generalized method of moments estimators, we abbreviate our presentation.

\subsection{Estimation under stationarity} \label{sec: weighting GMM stationary}

Let $g_q: \mathcal{W} \rightarrow \real^{d_\beta'}$, with $d_\beta' \geq d_\beta$, be a user-specified function. 
For $\beta \in \mathcal{B}$, $\lambda \in \Lambda$, $\psi \in \Psi$, and $\psi_- \in \Psi_-$, define $\theta^q = (\beta,\lambda,\psi,\psi_-)$, $\Theta^q = \mathcal{B} \times \Lambda \times \Psi \times \Psi_-$, and
$$G^q_t: \theta^q \mapsto \begin{pmatrix}
    \ind(t>T_0) \left\{ \psi - g_q(W_t) \right\} \\
    \ind(t \leq T_0) \left\{ q_\beta(Z_t) g_q(W_t) - \psi \right\} \\
    \ind(t>T_0) \left\{ \phi_\lambda(t) - Y_t + \psi_- \right\} \\
	\ind(t \leq T_0) \left\{ \psi_- - q_\beta(Z_t) Y_t \right\}
\end{pmatrix}.$$
In case of non-convexity, the interpretation is as for $\hat{\theta}_T$ from \eqref{eq: thetahat}.
Let $\Omega^q_T$ be a user-specified symmetric positive semi-definite $(2d_\beta'+2) \times (2d_\beta'+2)$ matrix. Consider the generalized method of moments estimator
$$\hat{\theta}^q_T = \left( \hat{\beta}^q_T,\hat{\lambda}^q_T,\hat{\psi}^q_T,\hat{\psi}^q_{-,T} \right) = \argmin_{\theta \in \Theta^q} \left\{ \frac{1}{T} \sum_{t=1}^T G^q_t(\theta)\right\}^\top \Omega^q_T \left\{ \frac{1}{T} \sum_{t=1}^T G^q_t(\theta) \right\}.$$
We propose to use $\phi_{\hat{\lambda}^q_T}(t_+)$ as the estimator of the average treatment effect for the treated unit $\phi^*(t_+)$ at time period $t_+>T_0$ based on the treatment confounding bridge function. 
We require the following assumption, similarly to condition~\ref{cond: misspecify h or q}.

\begin{condition} \label{cond: correctly specify q}
    (i) There exists a unique $\theta^q_\infty = (\beta^q_\infty,\lambda^q_\infty,\psi^q_\infty,\psi^q_{-,\infty},\phi^q_\infty) \in \Theta$ such that $\cexpect \{ G^q_t(\theta^q_\infty) \}=0$ for all $t=1,\ldots,T$. (ii) $q^*=q_{\beta^q_\infty}$ satisfies Condition~\ref{cond: treatment bridge}.
\end{condition}

Under Condition~\ref{cond: correctly specify q}, Theorem~\ref{thm: identify ATT treatment} implies that $\phi_{\lambda^q_\infty}(t_+)$ equals the average treatment effect for the treated unit $\phi^*(t_+)$ at time period $t_+>T_0$.
Similarly to Theorem \ref{thm: DR ATT asymptotic normality}, we have the following asymptotic result about the generalized method of moments estimator $\hat{\theta}^q_T$, under the additional assumptions \ref{cond: positive weight matrix}--\ref{cond: uniform law of large numbers}.

\begin{theorem} \label{thm: treatment ATT asymptotic normality}
	Under Conditions~\ref{cond: proxy independence}, \ref{cond: covariate shift}, \ref{cond: posttreatment identical distribution}, \ref{cond: correctly specify q}, \ref{cond: positive weight matrix}--\ref{cond: reg conditions for estimating function} and \ref{cond: uniform law of large numbers}, 
	as $T\to\infty$,
	we have that $\hat{\theta}^q_T$ is consistent for $\theta^q_\infty$. 
	Additionally under Conditions~\ref{cond: full rank matrix}, \ref{cond: uniform law of large numbers derivative} and \ref{cond: CLT weighted moment}, as $T\to\infty$, $\sqrt{T} (\hat{\theta}^q_T - \theta^q_\infty)$ is asymptotically $\mathrm{N}(0,A_q^{-1} B_q A_q^{-1})$,
	where we define
	\begin{align*}
		& R_q=\lim_{T \rightarrow \infty} \frac{1}{T} \sum_{t=1}^T \cexpect \{ \nabla_{\theta^q} G^q_t(\theta^q) |_{\theta^q=\theta^q_\infty} \}, 
		\qquad
		A_q = R_q^\top \Omega^q R_q,
		\\
		& B_q = R_q^\top \Omega^q \left[ \lim_{T \rightarrow \infty} \Var \left\{ T^{-1/2} \sum_{t=1}^T G^q_t(\theta^q_\infty) \right\} \right] \Omega^q R_q.
	\end{align*}
\end{theorem}

\subsection{Estimation without stationarity} \label{sec: weighting GMM nonstationary}

Let $g_q: \mathcal{W} \rightarrow \real^{d_\beta'}$, with $d_\beta' \geq d_\beta$, be a user-specified function. 
For $\beta \in \mathcal{B}$, $\phi \in \Phi$, $\psi \in \Psi$, and $\psi_- \in \Psi_-$, define $\undertilde{\theta}^q = (\beta,\phi,\psi,\psi_-)$, $\undertilde{\Theta}^q = \mathcal{B} \times \Phi \times \Psi \times \Psi_-$, and
$$\undertilde{G}^q_{T,t}: \undertilde{\theta}^q \mapsto \begin{pmatrix}
    \ind(t>T_0) \left\{ \psi - g_q(W_t) \right\} \\
    \ind(t \leq T_0) \left\{ \undertilde{q}_\beta(Z_t) g_q(W_t) - \psi \right\} \\
    \ind(t>T_0) \left\{ \phi - (T-T_0) \ell_T(t) Y_t + \psi_- \right\} \\
	\ind(t \leq T_0) \left\{ \psi_- - \undertilde{q}_\beta(Z_t) Y_t \right\}
\end{pmatrix}.$$
Let $\undertilde{\Omega}^q_T$ be a user-specified symmetric positive semi-definite $(2d_\beta'+2) \times (2d_\beta'+2)$ matrix. Consider the generalized method of moments estimator
$$\hat{\undertilde{\theta}}^q_T = \left( \hat{\undertilde{\beta}}^q_T,\hat{\undertilde{\phi}}^q_T,\hat{\undertilde{\psi}}^q_T,\hat{\undertilde{\psi}}^q_{-,T} \right) = \argmin_{\undertilde{\theta} \in \undertilde{\Theta}^q} 
\left\{ \frac{1}{T} \sum_{t=1}^T \undertilde{G}^q_{T,t}(\theta)\right\}^\top \undertilde{\Omega}^q_T \left\{ \frac{1}{T} \sum_{t=1}^T \undertilde{G}^q_t(\undertilde{\theta}) \right\}.$$
We propose to use $\hat{\undertilde{\phi}}_T$ as the estimator of the average treatment effect for the treated unit $\undertilde{\phi}^*$ based on the treatment confounding bridge function. 
We require the following assumption, similarly to condition~\ref{cond: misspecify h or q nonstationary}.

\begin{condition} \label{cond: correctly specify q nonstationary}
    (i) There exists a unique $\undertilde{\theta}^q_\infty = (\undertilde{\beta}^q_\infty,\undertilde{\phi}^q_\infty,\undertilde{\psi}^q_\infty,\undertilde{\psi}^q_{-,\infty},\undertilde{\phi}^q_\infty) \in \undertilde{\Theta}$ such that $\lim_{T \rightarrow \infty} \sum_{t=1}^T \cexpect \{ \undertilde{G}^q_{T,t}(\undertilde{\theta}^q_\infty) \}=0$. (ii) $\undertilde{q}^*=\undertilde{q}_{\undertilde{\beta}^q_\infty}$ satisfies Condition~\ref{cond: treatment bridge}.
\end{condition}

Under Condition~\ref{cond: correctly specify q nonstationary}, Theorem~\ref{thm: identify ATT treatment nonstationary limit} implies that $\undertilde{\phi}^q_\infty$ equals the average treatment effect for the treated unit $\undertilde{\phi}^*$ averaged over post-treatment periods.
Similarly to Theorem \ref{thm: DR ATT asymptotic normality nonstationary}, we have the following asymptotic result about the generalized method of moments estimator $\hat{\undertilde{\theta}}^q_T$, under the additional assumptions \ref{cond: positive weight matrix}--\ref{cond: uniform law of large numbers}.

\begin{theorem} \label{thm: treatment ATT asymptotic normality nonstationary}
	Under Conditions~\ref{cond: proxy independence}, \ref{cond: covariate shift}, \ref{cond: correctly specify q nonstationary}, \ref{cond: positive weight matrix}--\ref{cond: reg conditions for estimating function} and \ref{cond: uniform law of large numbers}, it holds that $\undertilde{\phi}^q_\infty = \undertilde{\phi}^*$, and
	as $T\to\infty$,
	we have that $\hat{\undertilde{\theta}}^q_T$ is consistent for $\undertilde{\theta}^q_\infty$, which implies that $\hat{\undertilde{\phi}}^q_T$ is consistent for $\undertilde{\phi}^*$. 
	Additionally under Conditions~\ref{cond: full rank matrix}, \ref{cond: uniform law of large numbers derivative} and \ref{cond: CLT weighted moment}, as $T\to\infty$, $\sqrt{T} (\hat{\undertilde{\theta}}^q_T - \undertilde{\theta}^q_\infty)$ is asymptotically $\mathrm{N}(0,\undertilde{A}_q^{-1} \undertilde{B}_q \undertilde{A}_q^{-1})$,
	where we define
	\begin{align*}
		& \undertilde{R}_q=\lim_{T \rightarrow \infty} \frac{1}{T} \sum_{t=1}^T \cexpect \{ \nabla_{\undertilde{\theta}^q} \undertilde{G}^q_{T,t}(\undertilde{\theta}^q) |_{\undertilde{\theta}^q=\undertilde{\theta}^q_\infty} \}, 
		\qquad
		\undertilde{A}_q = \undertilde{R}_q^\top \undertilde{\Omega}^q \undertilde{R}_q,
		\\
		& \undertilde{B}_q = \undertilde{R}_q^\top \undertilde{\Omega}^q \left[ \lim_{T \rightarrow \infty} \Var \left\{ T^{-1/2} \sum_{t=1}^T \undertilde{G}^q_{T,t}(\undertilde{\theta}^q_\infty) \right\} \right] \undertilde{\Omega}^q \undertilde{R}_q.
	\end{align*}
\end{theorem}

\section{Incorporating covariates} \label{sec: use covariates}

For some set $\mathcal{X}$, let $X_t\in\mathcal{X}$ be the observed covariates in set at time period $t$. 
A natural way to incorporate covariate $X_t$ is to condition on $X_t$, treating $X_t$ similarly to $U_t$, but as observed.
Below we list the causal assumptions and identification results with a one-to-one correspondence to those in the main text under stationarity of covariates in the post-treatment periods (similar to Condition~\ref{cond: posttreatment identical distribution}). In fact, the identification result Theorem~\ref{thm: identify ATT DR} in the main text is a special case of Theorem~\ref{thm: identify ATT DR2} below with $X_t=\emptyset$. The proof of Theorem~\ref{thm: identify ATT DR2} below is also similar with the only modification being to condition on $X_t$ throughout. We therefore omit their proofs. 
The results without stationarity are strikingly similar and thus omitted.
The generalized method of moments estimation methods corresponding to the results below can also be derived similarly to those without covariates.

An issue with the results below is that the covariates $X_t$ must be used in the post-treatment period in the bridge functions. We therefore need to assume that $X_t$ is exogenous.
However, this assumption may well be violated, since distributional shift from $X_{t_-}$ to $X_{t_+}$ may be a consequence of the treatment,
thus being endogenous. 
This is in contrast to the classical SC method from \cite{Abadie2010}, where only covariates in the pre-treatment period are used.

Another possible way to use covariates is to view them as proxies and concatenate them to $W_t$ or $Z_t$. The user can decide which set of proxies each covariate belongs to based on Conditions~\ref{cond: proxy independence}--\ref{cond: treatment complete}. For example, covariates of donors and the treated unit may be part of $W_t$, while covariates of the other control units may be part of $Z_t$.

Our assumptions are as follows:
\begin{condition}\label{cond: proxy independence2}
	For all $t_- \leq T_0$, $Z_{t_-} \independent (Y_{t_-},W_{t_-}) \mid (X_{t_-},U_{t_-})$.
\end{condition}

\begin{condition}
\label{cond: outcome bridge2}
    There exists a function $h^*: \mathcal{W} \times \mathcal{X} \rightarrow \real$ such that $\cexpect \{ h^*(W_t,X_t) \mid X_t,U_t \}=\cexpect \{ Y_t(0) \mid X_t,U_t \}$ for all $t$.
\end{condition}

\begin{condition}
\label{cond: outcome complete2}
    Let $g: \mathcal{W} \rightarrow \real$ be any square-integrable function. For all $t_- \leq T_0$ and $P_{X_{t_-}}$-a.e. $x \in \mathcal{X}$, the following two statements are equivalent: (i) $g(W_{t_-})=0$, $P_{W_{t_-} \mid X_{t_-}=x}$-almost surely, (ii) $\cexpect \{ g(W_{t_-}) \mid Z_{t_-},X_{t_-}=x \} = 0$, $P_{Z_{t_-} \mid X_{t_-}=x}$-almost surely.
\end{condition}

\begin{condition}
\label{cond: covariate shift2}
    The conditional distribution $(Y_t(0),W_t) \mid (X_t,U_t)$ is identical for all $t$.
\end{condition}

\begin{condition}
\label{cond: posttreatment identical distribution2}
    The distribution of $(X_{t_+},U_{t_+})$ is identical for all $t_+ > T_0$.
\end{condition}

\begin{condition}
\label{cond: treatment bridge2}
    Suppose that $P_{X_{t_+},U_{t_+}}$ is dominated by $P_{X_{t_-},U_{t_-}}$ and there exists a function $q^*: \mathcal{Z} \times \mathcal{X} \rightarrow \real$ such that, 	for $P_{X_{t_-},U_{t_-}}$-a.e. $(x,u)$ and for all $t_- \leq T_0$ and $t_+ > T_0$,
	$$\cexpect \{ q^*(Z_{t_-},x) \mid X_{t_-}=x,U_{t_-}=u \} = \frac{\intd P_{X_{t_+},U_{t_+}}}{\intd P_{X_{t_-},U_{t_-}}}(x,u).$$
\end{condition}

\begin{condition}
\label{cond: treatment complete2}
    Let $g: \mathcal{Z} \rightarrow \real$ be any square-integrable function. For all $t_- \leq T_0$ and $P_{X_{t_-}}$-a.e. $x \in \mathcal{X}$, the following two statements are equivalent: (i) $g(Z_{t_-})=0$, $P_{Z_{t_-} \mid X_{t_-}=x}$-almost surely, (ii) $\cexpect \{ g(Z_{t_-}) \mid W_{t_-},X_{t_-}=x \} = 0$, $P_{W_{t_-} \mid X_{t_-}=x}$-almost surely.
\end{condition}
The resulting theorems are analogous to those from the main text.
\begin{theorem}
\label{thm: identify ATT treatment2}
    Let $f:\real \rightarrow \real$ be any square-integrable function. Under Conditions~\ref{cond: proxy independence2} and \ref{cond: covariate shift2}--\ref{cond: treatment bridge2}, it holds that
	\begin{equation}
	    \cexpect [ f\{Y_{t_+}(0)\} ] = \cexpect \{ q^*(Z_{t_-},X_{t_-}) f(Y_{t_-}) \}.
	    \label{eq: identify EfY0 treatment2}
	\end{equation}
	In particular, taking $f$ to be the identity function, it holds that
	\begin{equation}
		\cexpect\{Y_{t_+}(0)\} = \cexpect\{q^*(Z_{t_-},X_{t_-}) Y_{t_-}\}
		\label{eq: identify EY0 treatment2}
	\end{equation}
	and thus
	$\phi^*(t_+)=\cexpect\{Y_{t_+}-q^*(Z_{t_-},X_{t_-}) Y_{t_-}\}$
	for all $t_- \leq T_0$ and $t_+ > T_0$. In addition,
	$P_{W_{t_+},X_{t_+}}$ is dominated by $P_{W_{t_-},X_{t_-}}$
	and the treatment confounding bridge function $q^*$ is a solution to
	\begin{equation}
		\cexpect\{q(Z_{t_-},x) \mid W_{t_-}=w,X_{t_-}=x\}= \frac{\intd P_{W_{t_+},X_{t_+}}}{\intd P_{W_{t_-},X_{t_-}}}(w,x) \quad \text{for $P_{W_{t_-},X_{t_-}}$-a.e. $(w,x) \in \mathcal{W} \times \mathcal{X}$}
		\label{eq: identify q2}
	\end{equation}
	in
	$q: \mathcal{Z} \times \mathcal{X} \rightarrow \real$.
    Also assuming Condition~\ref{cond: treatment complete2}, \eqref{eq: identify q2} has a unique solution up to probability zero sets.
\end{theorem}

\begin{theorem}
\label{thm: identify ATT DR2}
    Let
	$h: \mathcal{W} \times \mathcal{X} \rightarrow \real$
	and
	$q: \mathcal{Z} \times \mathcal{X} \rightarrow \real$
	be any square-integrable functions. Under Conditions~\ref{cond: proxy independence2}, \ref{cond: outcome bridge2}, \ref{cond: covariate shift2}, \ref{cond: posttreatment identical distribution2} and \ref{cond: treatment bridge2}, it holds that
	\begin{align}
	    \begin{split}
	        \cexpect\{Y_{t_+}(0)\} &= \cexpect \left[ q(Z_{t_-},X_{t_-}) \{Y_{t_-} - h(W_{t_-},X_{t_-})\} + h(W_{t_+},X_{t_+}) \right], \\
	        \phi^*(t_+) &= \cexpect \left[ Y_{t_+} - q(Z_{t_-},X_{t_-}) \{Y_{t_-} - h(W_{t_-},X_{t_-})\} - h(W_{t_+},X_{t_+}) \right] \label{eq: DR identify ATT2}
	    \end{split}
	\end{align}
	for all $t_- \leq T_0$ and $t_+ > T_0$, if (i) Condition~\ref{cond: outcome bridge2} holds and $h=h^*$,
	or (ii) Condition~\ref{cond: treatment bridge2} holds and $q=q^*$.
\end{theorem}

\section{Generalized method of moments implementation in simulations \& data analysis} \label{sec: implementation of GMM}

In this appendix, we describe some implementation details for the generalized method of moments used in our simulations and data analysis with the \texttt{R} package \texttt{gmm}. 
One challenge with the generalized method of moments in our methods is that the treatment confounding bridge function $q^*$ is often parameterized as an exponential function, and thus the value of the moment equation is sensitive to small changes in the coefficients. 
This seems to have led to numerical instability in our experience. We expect better software implementation of the generalized method of moments to resolve these issues by default, but with the current software implementation, we describe some non-default options that we have found to alleviate numerical issues for our methods.

We recommend providing the analytic gradient function to \texttt{gmm} function. Compared to the default numerical gradients, in our experiments, the analytic gradient functions could improve stability and speed in numerical optimization in the generalized method of moments.

Because the generalized method of moments can be highly nonlinear, the associated optimization problem can have several local optima. Thus, proper initialization can be crucial to obtaining an estimator that is close to the truth. 
When feasible, we recommend fitting a GLM assuming no latent confounders and taking the fitted coefficients as the initial values. For example, when $h^*$ is specified as a linear function of $W_t$, we can fit an ordinary least squares regression with outcome $Y_t$ and covariates $W_t$ for $t \leq T_0$, and take the fitted coefficients as the initial coefficient for $h^*$. 

When $q^*$ is specified as a log-linear model, namely $q^*: z \mapsto \exp \left(\beta_0 + \sum_{k=1}^K \beta_k z_k \right)$, we may fit a GLM as follows. If we treat time period $t$ as a random variable uniform over the observed time points, by Bayes' theorem, we find that
\begin{align*}
    \frac{\intd P_{Z_{t_+}}}{\intd P_{Z_{t_-}}}(z) = \frac{\intd P_{Z_{t \mid t>T_0}}}{\intd P_{Z_t \mid t \leq T_0}}(z) = \frac{\Prob(t>T_0 \mid Z_t=z)}{\Prob(t \leq T_0 \mid Z_t=z)} \frac{\Prob(t \leq T_0)}{\Prob(t > T_0)},
\end{align*}
where $\Prob(t \leq T_0)$ and $\Prob(t \leq T_0)$ may be interpreted as the proportions of pre- and post-treatment periods. We set the above to be equal to $\exp \left(\beta_0 + \sum_{k=1}^K \beta_k z_k \right)$ for initialization. Then, with $A_t$ defined to be $\ind(t>T_0)$, we have that
\begin{equation}
    \log \mathrm{OR}(A_t \mid Z_t=z) = \log \frac{\Prob(A_t=1)}{\Prob(A_t=0)} + \beta_0 + \sum_{k=1}^K \beta_k z_k \label{eq: initialize q with GLM}
\end{equation}
where $\mathrm{OR}$ stands for odds ratio, and $\Prob(A_t=1)/\Prob(A_t=0)=(T-T_0)/T_0$. Therefore, we can run a logistic regression with outcome $A_t$ and covariate $Z_t$ for all $t$, and derive an initial value for the coefficient $\beta=(\beta_0,\ldots,\beta_K)$ with an adjustment in the intercept $\beta_0$ according to \eqref{eq: initialize q with GLM}. In simulations, we have found this initialization approach effective in removing a potentially large proportion of estimates that are far from the truth, thus achieving consistency and approximately correct confidence interval coverage.
In practice, a sensitivity analysis can also be conducted by, for example, randomly initializing parameter values or optimizing the objective function via stochastic methods like stochastic gradient descent.

Many choices of the weight matrix $\Omega_T$ lead to asymptotically normal generalized method of moments estimators since the only requirement is that the probability limit $\Omega$ of $\Omega_T$ is positive definite (see Condition~\ref{cond: positive weight matrix}). In addition to setting $\Omega_T$ to be a fixed weight matrix, many choices of $\Omega_T$ have been proposed, including the continuously updated efficient generalized method of moments, the two-step generalized method of moments estimator, and the iteratively updated generalized method of moments estimator.
These three data-adaptive choices all lead to asymptotically efficient generalized method of moments estimators \citep{Hansen1996}, while a fixed weight matrix is generally asymptotically inefficient. In finite samples, the continuously updated efficient generalized method of moments estimator has been found to perform better than the two-step generalized method of moments estimator and the iteratively updated estimator \citep{Hansen1996}, which are preferable to using a fixed weight matrix \citep{Hall2007}. Nevertheless, we have found more numerical issues with these more complicated methods for our nonlinear generalized method of moments in finite samples. We thus recommend using a fixed weight matrix because it still leads to asymptotically valid inference despite its asymptotic inefficiency,
and it has led to fewer numerical issues in our experience. In particular, we chose the weight matrix to be the identity matrix throughout; that is, we set \texttt{wmatrix=``ident''} in the \texttt{gmm} function.

In addition, the option \texttt{vcov=``HAC''} in the \texttt{gmm} function in the \texttt{R} package \texttt{gmm} has led to numerical errors in our experiments. The issue seems to be caused by prewhitening. Since our theory does not require prewhitening, we recommend not prewhitening and instead setting \texttt{vcov=``iid''} in the \texttt{gmm} function, and then using \texttt{vcovHAC} with the default option \texttt{prewhite=FALSE} from the \texttt{sandwich} package to estimate the variance for dependent time series data.
This has led to significantly fewer numerical issues and appears to have correct asymptotic behavior in our simulations. 
It is also feasible to use the \texttt{NeweyWest} function with the option \texttt{prewhite=FALSE} from the \texttt{sandwich} package to estimate the variance.

In some applications, the magnitude of the variables can be large. As an example, the number of hospitalizations in Section~\ref{sec: vaccine analysis} is of order $10^2$--$10^3$. Directly using these values in the GMM involving estimation of the treatment confounding bridge function $q^*$ can lead to severe numerical instability.
The reason seems that, when $Z_t$ has large magnitudes and the treatment confounding bridge function is parameterized by a log-linear model as in our simulation and data analysis, even a tiny change in the parameter value will lead to a drastic change in the confounding bridge function value. The above issue with the magnitude of $Z_t$ appears to be another challenge caused by strong nonlinearity and large gradients.
Scaling all outcomes to be of a unit order before feeding the data into GMMs appears to significantly improve the numerical stability.

\section{Additional simulation details and results} \label{sec: more sim}

\subsection{Data-generating mechanism for just-identified seeting} \label{sec: sim DGP}

We let $T \in \{500,1000,2000,4000\}$ and $T_0=T/2$. For each $T$, we run 200 Monte Carlo simulations.

We let the number of latent confounders, the number of donors and of the other control units all be  $K$. 
We generate random vectors $U_t=(U_{t,1},\ldots,U_{t,K})^\top$, $W_t=(W_{t,1},\ldots,W_{t,K})^\top$ and $Z_t=(Z_{t,1},\ldots,Z_{t,K})^\top$. 
We consider $K \in \{2,3,4,5\}$. 
We generate serially correlated latent confounders $U_t$ via a Gaussian copula \citep[e.g., ][]{Jaworski2010} as follows:
\begin{align*}
    & \epsilon_t \overset{i.i.d.}{\sim} \mathrm{N}(0,I_K), \qquad \undertilde{U}_1 = \epsilon_1, \qquad \undertilde{U}_t = 0.1 \undertilde{U}_{t-1} + 0.9 \epsilon_t, \quad t \geq 2, \\
    & U_{t_-,k} = F_1^{-1}\{ \Phi(\undertilde{U}_{t_-,k}) \}, \qquad U_{t_+,k} = F_2^{-1} \{ \Phi(\undertilde{U}_{t_+,k}) \}, 
    \quad t_- \leq T_0,\, t_+ > T_0,\, k=1,\ldots,K,
\end{align*}
where $I_K$ is the $K \times K$ identity matrix, $\Phi$ is the cumulative distribution function of the standard normal distribution, and $F_\lambda$ is the cumulative distribution function of the $\mathrm{Exponential}(\lambda)$ distribution with rate $\lambda>0$. 
Subsequently, for each $t\ge 1$, the observed outcomes $(Y_t,W_t,Z_t)$ are generated independently conditional on $U_t$ as follows:
\begin{align*}
    & Y_t \mid U_t \sim \mathrm{Unif} \left( 2 \ind(t > T_0) + 2 \sum_{k=1}^K U_{t,k} - 1, 2 \ind(t > T_0) + 2 \sum_{k=1}^K U_{t,k} + 1 \right) \\
    & W_{t,k} \mid U_t \sim \mathrm{Unif} \left( 2 U_{t,k} - 1, 2 U_{t,k} + 1 \right), \qquad Z_{t,k} \mid U_t \sim \mathrm{Unif} \left( 2 U_{t,k} - 1, 2 U_{t,k} + 1 \right)
\end{align*}
for $k=1,\ldots,K$.
Since $U_t$ contains serial correlation, so do the observed variables $(W_t,Z_t,Y_t)$.
As an example, $(Y_t,W_t,Z_t)$ may be numbers of hospitalizations due to various causes up to shifting and scaling, while $U_t$ may be the unobserved confounding factors, such as the overall infection level and the overall immune status.
From the above formula for $Y_t \mid U_t$, we see that the true ATT equals two.

Therefore, the outcomes are generated from a factor model \citep{Abadie2010} and $h^*: w \mapsto \alpha_0 + \sum_{k=1}^K \alpha_k w_k$ satisfies Condition~\ref{cond: outcome bridge} for some coefficients $\alpha_0,\ldots,\alpha_K \in \real$. 
By the properties of the copula generating $U_t$, the likelihood ratio in Condition~\ref{cond: treatment bridge} is
$C \exp( - \sum_{k=1}^K u_k)$
for some constant $C>0$. 
By a simple calculation, we can check that $q^*: z \mapsto \exp \left( \beta_0 + \sum_{k=1}^K \beta_k z_k \right)$  satisfies Condition~\ref{cond: treatment bridge} for some coefficients $\beta_0,\ldots,\beta_K \in \real$.
For methods that involve a correctly specified outcome confounding bridge function $h^*$, we parameterize $h^*$ as above and choose $g_h: z \mapsto (1,z^\top)^\top$. When $q^*$ is correctly specified, we parameterize $q^*$ as above and choose $g_q: w \mapsto (1,w^\top)^\top$. When $h^*$ is misspecified, we omit the last several proxies: we parameterize $h^*$ as $h: w \mapsto \alpha_0 + \alpha_1 w_1$ with coefficients $\alpha_0,\alpha_1 \in \real$, and choose $g_h: z \mapsto (1,z_1)^\top$. 
Similarly, when $q^*$ is misspecified, we parameterize $q^*$ as $q: z \mapsto \exp(\beta_0 + \beta_1 z_1)$ with coefficients $\beta_0,\beta_1 \in \real$, and choose $g_q: w \mapsto (1,w_1)^\top$.

\subsection{Over-identified setting} \label{sec: sim2}

The data generating mechanism is similar to the just identified setting, and we abbreviate the description. There are three latent confounders $U_t \in \real^3$, 
five donors $W_t \in \real^5$ and ten other control units $Z_t \in \real^{10}$. 
The latent confounders are generated as follows:
\begin{align*}
    & \epsilon_t \overset{iid}{\sim} \mathrm{N}(0,I_3), \qquad \undertilde{U}_1 = \epsilon_1, \qquad \undertilde{U}_t = 0.1 \undertilde{U}_{t-1} + 0.9 \epsilon_t \quad (t \geq 2), \\
    & U_{t_-,k} = F_1^{-1}\{\Phi(\undertilde{U}_{t_-,k})\}, \qquad U_{t_+,k} = F_2^{-1}\{\Phi(\undertilde{U}_{t_+,k})\}, \quad t_- \leq T_0,\, t_+ > T_0,\, k=1,\ldots,3.
\end{align*}
The observed outcomes are generated as follows:
\begin{align*}
    & Y_t \mid U_t \sim \mathrm{Unif} \left( 2 \ind(t>T_0) + 2 \sum_{k=1}^3 U_{t,k} - 1, 2 \ind(t>T_0) + 2 \sum_{k=1}^3 U_{t,k} + 1 \right), \\
    & W_{t,k} \mid U_t \sim \mathrm{Unif}(a_k^\top U_t - 1, a_k^\top U_t + 1), \quad k=1,\ldots,5, \\
    & Z_{t,k} \mid U_t \sim \mathrm{Unif}(b_k^\top U_t - 1, b_k^\top U_t + 1),\, \quad k=1,\ldots,10,
\end{align*}
where
\begin{align*}
    &a_1 = (1,0,0)^\top, \quad a_2 = (0,1,0)^\top, \quad a_3 = (0,0,1)^\top, \quad a_4 = (1,1,0)^\top, \quad a_5 = (1,0,1)^\top, \\
    & b_1 = (2,0,0)^\top, \quad b_2 = (0,2,0)^\top, \quad b_3 = (0,0,2)^\top, \quad b_4 = (-3,0,0)^\top, \quad b_5 = (0,-3,0)^\top, \\
    & b_6 = (0,0,-3)^\top, \,\, b_7 = (1,-1,0)^\top, \,\, b_8 = (1,0,-1)^\top, \,\, b_9 = (0,1,-1)^\top, \, b_{10} = (2,-0.5,-0.5)^\top.
\end{align*}
From the form of $Y_t \mid U_t$, the true average treatment effect for the treated unit equals two.

One can verify that there is a
valid outcome confounding bridge function takes of form $h^*: w \mapsto \alpha_0 + \sum_{k=1}^5 \alpha_k w_k$ for some coefficients $\alpha_0,\alpha_1,\ldots,\alpha_5 \in \real$. 
Further, there is a valid treatment confounding bridge function of the form $q^*: z \mapsto \exp(\beta_0 + \sum_{k=1}^{10} \alpha_k z_k)$ for some coefficients $\beta_0,\beta_1,\ldots,\beta_{10} \in \real$. When $h^*$ is correctly specified, we parameterize $h^*$ as above and take $g_h: z \mapsto (1,z^\top)^\top$. When $q^*$ is correctly specified, we also parameterize $q^*$ as above and take $g_q(w)$ to be the vector in $\real^{21}$ consisting of $w_{k_1}^{\gamma_{k_1}} w_{k_2}^{\gamma_{k_2}}$ with $k_1 \neq k_2$, $\gamma_{k_1}+\gamma_{k_2} \leq 2$, and both of $\gamma_{k_1}$ and $\gamma_{k_2}$ are non-negative integers. 
This can be concisely expressed as \texttt{cbind(1,poly(w,degree=2,raw=TRUE))} in \texttt{R}.

We did not encounter numerical issues in this simulation. The sampling distributions of the estimated average treatment effect for the treated unit and the 95\%-Wald confidence interval coverage are presented in Figure~\ref{fig: sim2 phi hat} and \ref{fig: sim2 CI cover}, respectively. As shown in these figures, the performance of the methods is similar to the just identified setting from Section~\ref{sec: sim}.

\begin{figure}[bt]
    \centering
    \includegraphics[scale=0.7]{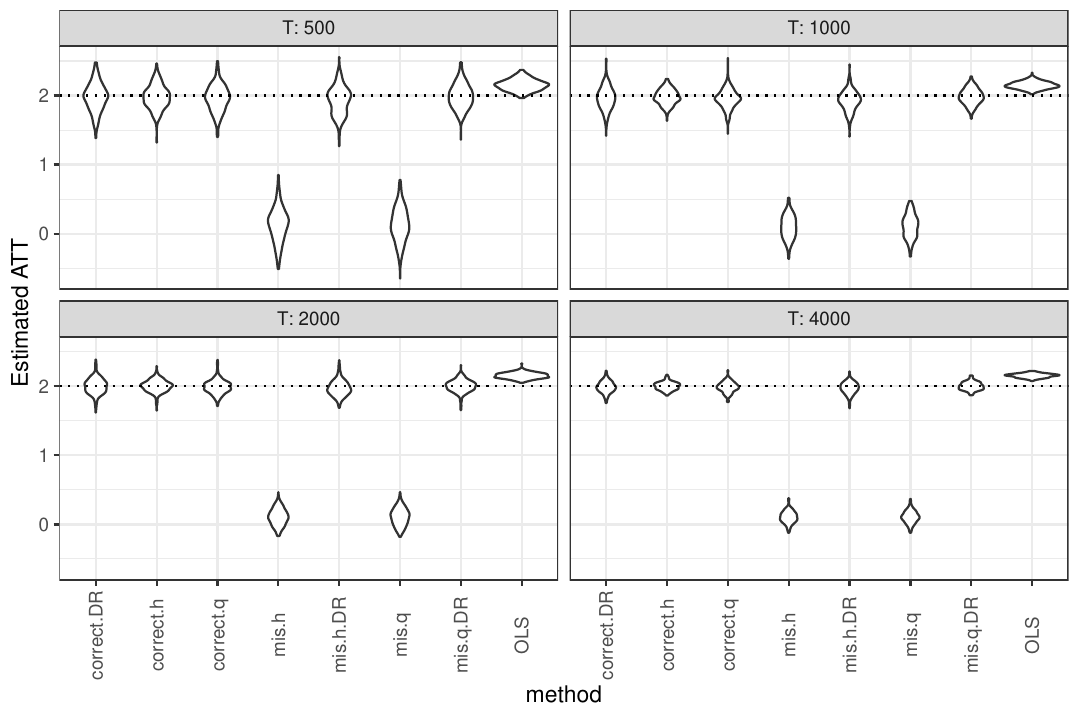}
    \caption{Figures similar to Figure~\ref{fig: sim phi hat} for the over-identified setting.}
    \label{fig: sim2 phi hat}
\end{figure}

\begin{figure}[bt]
    \centering
    \includegraphics[scale=0.7]{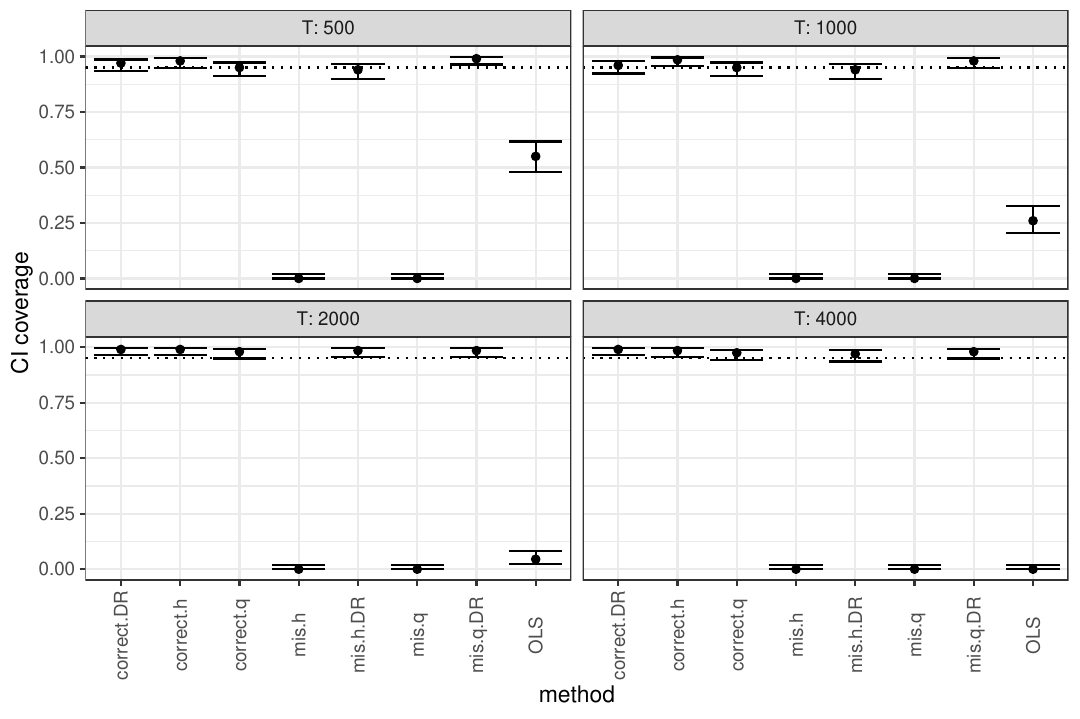}
    \caption{Figures similar to Figure~\ref{fig: sim CI cover} for the overidentified setting.}
    \label{fig: sim2 CI cover}
\end{figure}

\subsection{Smaller sample size $T$} \label{sec: sim smallT}

We investigate the performance of our methods when the sample size $T$ is relatively small. In this simulation, we consider a similar setting as in Section~\ref{sec: sim} except that $T \in \{80,100\}$, $T-T_0=20$, and we restrict to $K=2$. The choice of $T$ and $T-T_0$ mimics the values in the analysis of Brazil all-cause pneumonia hospitalization in Section~\ref{sec: vaccine analysis}.
The simulation results are presented in Figures~\ref{fig: sim smallT phi hat} and \ref{fig: sim smallT CI cover}.
We observe similar phenomena as in Section~\ref{sec: sim}.
However, due to smaller sample sizes $T$, the estimators' distributions are further from Gaussian than in Section~\ref{sec: sim}, suggesting that the asymptotic results (Theorems~\ref{thm: DR ATT asymptotic normality}, \ref{thm: DR ATT asymptotic normality nonstationary}, \ref{thm: treatment ATT asymptotic normality}, and \ref{thm: treatment ATT asymptotic normality nonstationary}) might not approximate the sampling distributions well at such small sample sizes.
Thus, the somewhat outlying estimate from \texttt{DR} in Section~\ref{sec: vaccine analysis}, which uses a very simple model for the treatment confounding bridge function under a relatively small sample size ($T=108, T_0=84$), might not be surprising.

\begin{figure}[bt]
    \centering
    \includegraphics[scale=0.7]{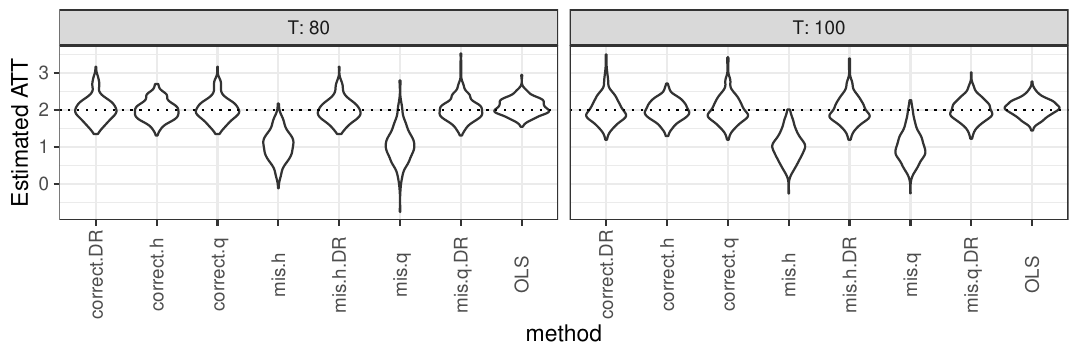}
    \caption{Figures similar to Figure~\ref{fig: sim phi hat} for small sample sizes $T$.}
    \label{fig: sim smallT phi hat}
\end{figure}

\begin{figure}[bt]
    \centering
    \includegraphics[scale=0.7]{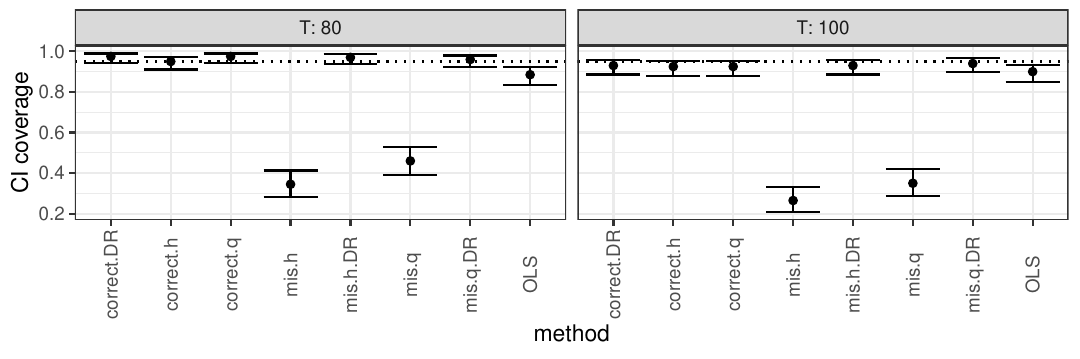}
    \caption{Figures similar to Figure~\ref{fig: sim CI cover} for small sample sizes $T$.}
    \label{fig: sim smallT CI cover}
\end{figure}

\subsection{Short post-treatment period} \label{sec: sim shortpost}

We also investigate the performance of our methods when the post-treatment period is relatively short. The simulation setting is also similar to Section~\ref{sec: sim}, except that $T_0=T- \lfloor \sqrt{T} \rfloor$ and we restrict to $K=2$, where $\lfloor \cdot \rfloor$ denotes the floor function. Although the number of post-treatment time periods $T-T_0$ still grows to infinity as $T \to \infty$, the growth has a slower rate and thus $T_0/T \to 1$. The results are presented in Figures~\ref{fig: sim shortpost phi hat} and \ref{fig: sim shortpost CI cover}. Unlike Section~\ref{sec: sim} where $T_0/T=1/2$, proximal causal inference methods whose validity relies on the correct weighting function (\texttt{corect.q} and \texttt{mis.h.DR}) perform poorly with large bias, even if the treatment confounding bridge function is correctly specified. In contrast, methods whose validity can be justified by correct outcome bridge functions (\texttt{correct.DR}, \texttt{correct.h} and \texttt{mis.q.DR}) still appear to perform well as in Section~\ref{sec: sim}. Thus, when $T_0/T \to 1$, our proposed method is not robust against misspecification of the outcome bridge function, even if the treatment confounding bridge function is correctly specified. Having a sufficient proportion of post-treatment periods appears necessary for our proposed method to be doubly robust.
One explanation for this phenomenon is that, to obtain a stable estimator of the treatment confounding bridge function capturing the likelihood ratio in \eqref{eq: define q}, a comparable amount of data from pre- and post-treatment periods is needed. Having too little data from either period could lead to substantially larger variances for the treatment confounding bridge function estimator.

\begin{figure}[bt]
    \centering
    \includegraphics[scale=0.7]{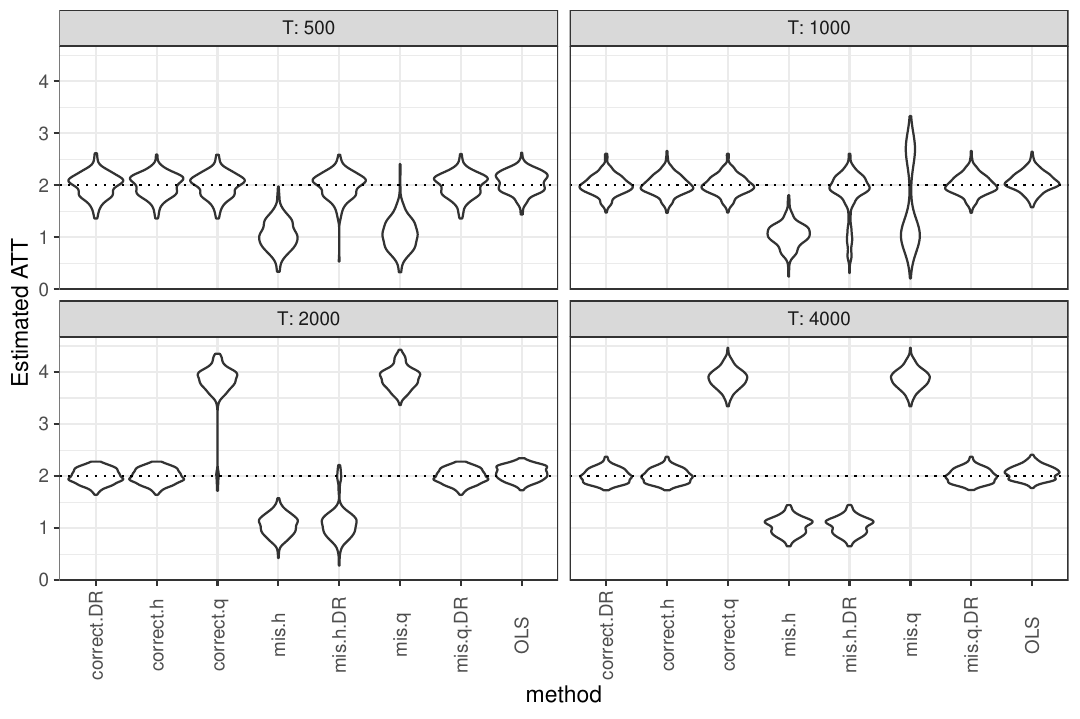}
    \caption{Figures similar to Figure~\ref{fig: sim phi hat} for short post-treatment periods.}
    \label{fig: sim shortpost phi hat}
\end{figure}

\begin{figure}[bt]
    \centering
    \includegraphics[scale=0.7]{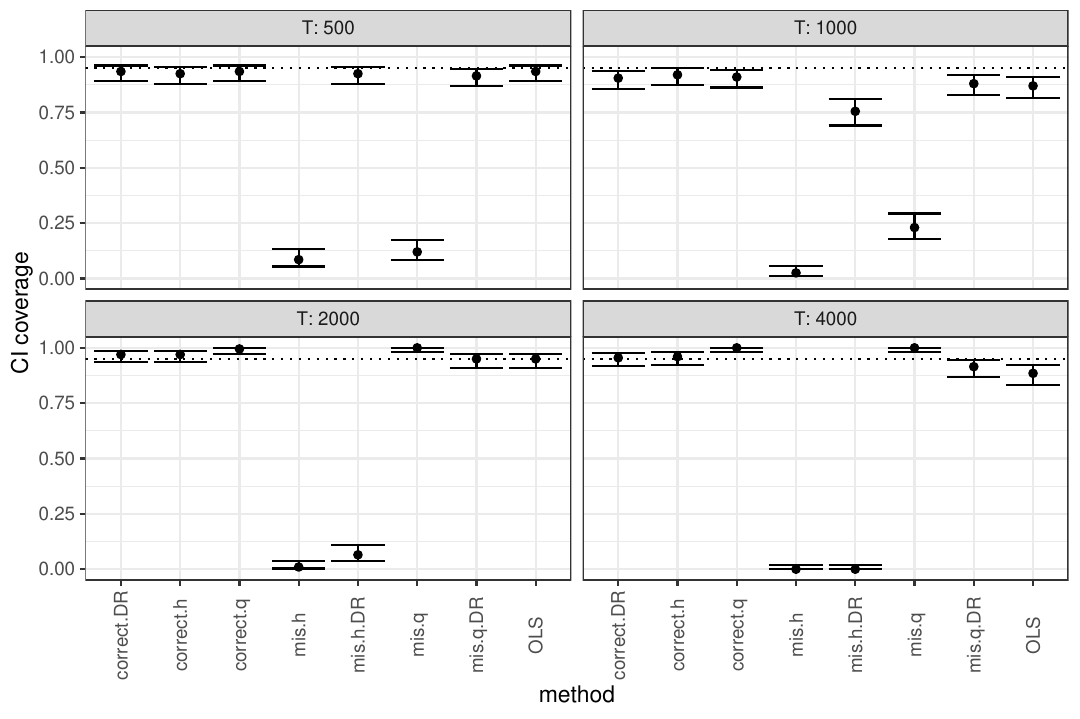}
    \caption{Figures similar to Figure~\ref{fig: sim CI cover} for short post-treatment periods.}
    \label{fig: sim shortpost CI cover}
\end{figure}

\subsection{Autoregressive data-generating model} \label{sec: sim normal}

We consider the following autoregressive data-generating model:
\begin{align*}
    & \epsilon_t \overset{i.i.d.}{\sim} \mathrm{N}(0,I_K), \qquad U_1 = \epsilon_1, \qquad U_t = 0.1 U_{t-1} + 0.9 \epsilon_t, \quad t \geq 2, \\
    & Y_t \mid U_t \sim \mathrm{N} \left( 2 \ind(t > T_0) + 2 \sum_{k=1}^K U_{t,k},1 \right), \\
    & W_{t,k} \mid U_t \sim \mathrm{N}(2 U_{t,k},1), \qquad Z_{t,k} \mid U_t \sim \mathrm{N}(2 U_{t,k},1).
\end{align*}
We also restrict to $K=2$ and consider the seven methods as in Section~\ref{sec: sim}. The only difference is the forms of the confounding bridge functions because $U_t$ is marginally Gaussian, rather than exponentially, distributed (see Examples~\ref{ex: normal} and \ref{ex: exponential}). In particular, when the outcome confounding bridge function is correctly specified, we parameterize $h^*$ as $w \mapsto \alpha_0 + \alpha_1 w_1 + \alpha_2 w_2 + \alpha_3 w_1^2 + \alpha_4 w_1 w_2 + \alpha_5 w_2^2$ and choose $g_h: z \mapsto (1,z_1,z_2,z_1^2,z_1 z_2,z_2^2)^\top$; when the treatment confounding bridge function is correctly specified, we parameterize $q^*$ as $z \mapsto \exp( \beta_0 + \beta_1 z_1 + \beta_2 z_2 + \beta_3 z_1^2 + \beta_4 z_1 z_2 + \beta_5 z_2^2 )$ and choose $g_q: w \mapsto (1,w_1,w_2,w_1^2,w_1 w_2,w_2^2)^\top$. We also consider misspecified confounding bridge functions similar to Section~\ref{sec: sim}.

The simulation results are shown in Figures~\ref{fig: sim normal phi hat} and \ref{fig: sim normal CI cover}. We overall performance of all methods are similar to Section~\ref{sec: sim}. The only exception is that \texttt{mis.h} and \texttt{OLS} happen to also perform reasonably well. However, as we show in Section~\ref{sec: sim}, these two methods are not guaranteed to produce consistent and asymptotically normal estimators.

\begin{figure}[bt]
    \centering
    \includegraphics[scale=0.7]{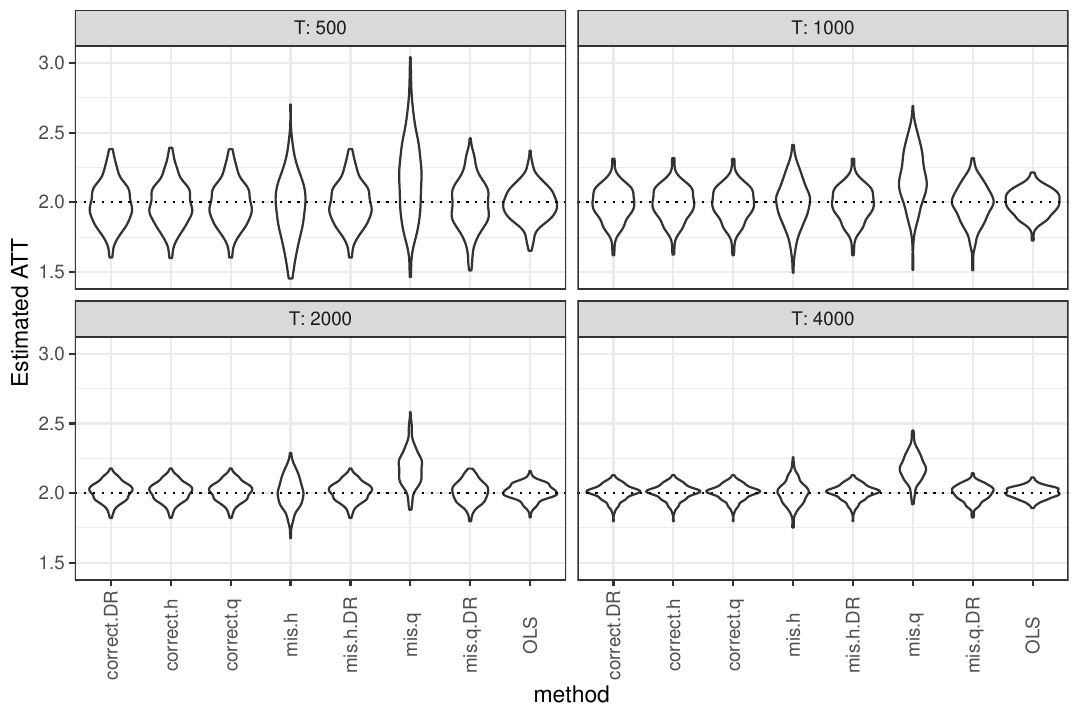}
    \caption{Figures similar to Figure~\ref{fig: sim phi hat} for an autoregressive data-generating model.}
    \label{fig: sim normal phi hat}
\end{figure}

\begin{figure}[bt]
    \centering
    \includegraphics[scale=0.7]{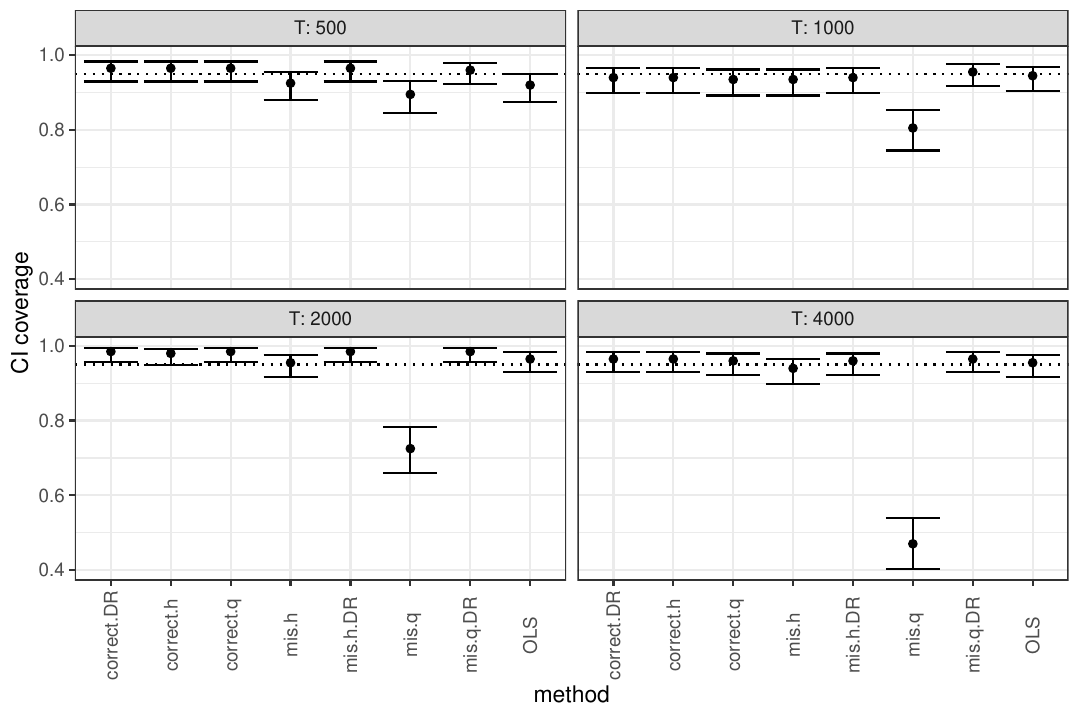}
    \caption{Figures similar to Figure~\ref{fig: sim CI cover} for an autoregressive data-generating model.}
    \label{fig: sim normal CI cover}
\end{figure}

\subsection{Detrending before analysis} \label{sec: sim detrend}

We investigate the effect of detrending the panel data before applying our methods. We consider a scenario similar to that in Section~\ref{sec: sim normal}, except that we add a time trend $t \mapsto 10 t/T + 10 (t/T)^2$ when generating the observed data $(W_t,Z_t,Y_t)$. In all methods, we detrend the data as described in the following Section~\ref{sec: data analysis2} before applying all methods. We focus on the methods with correctly specified confounding bridge functions and the case $K=2$. We run 1,000 Monte Carlo simulations for each $T$ to obtain more accurate inference about the CI coverage.

The results are shown in Figures~\ref{fig: sim detrend phi hat} and \ref{fig: sim detrend CI cover}.
Although the estimators still appear normally distributed, the CI coverage is lower than the nominal level 95\% by around 2--5\%. Such anti-conservativeness might be caused by the fact that the uncertainty in estimating the trend is not accounted for in this procedure.

\begin{figure}[bt]
    \centering
    \includegraphics[scale=0.7]{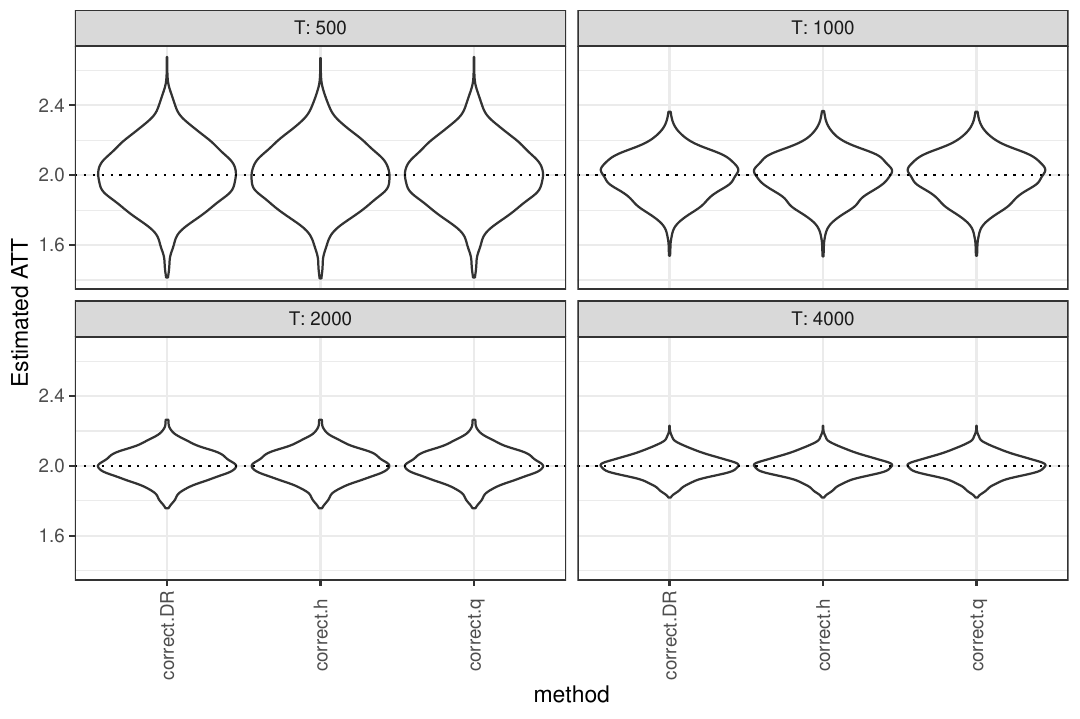}
    \caption{Figures similar to Figure~\ref{fig: sim phi hat} with detrending.}
    \label{fig: sim detrend phi hat}
\end{figure}

\begin{figure}[bt]
    \centering
    \includegraphics[scale=0.7]{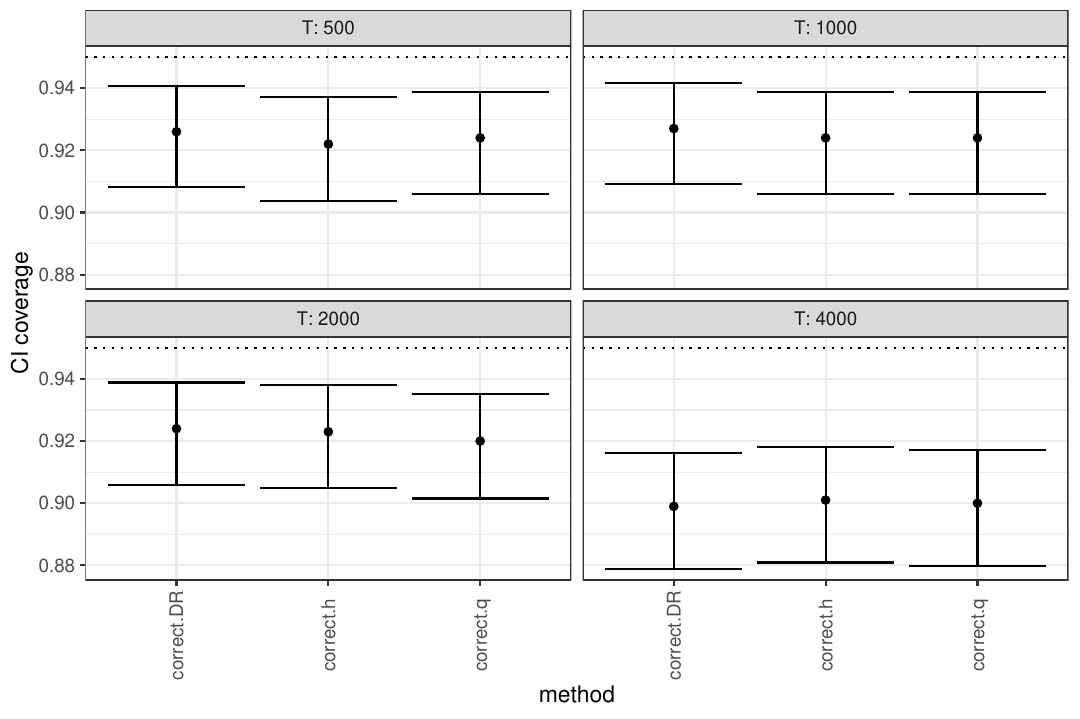}
    \caption{Figures similar to Figure~\ref{fig: sim CI cover} with detrending.}
    \label{fig: sim detrend CI cover}
\end{figure}

\subsection{Post-selection} \label{sec: sim post-select}

We investigate the effect of post-selecting $W$ and $Z$ before applying our methods. We consider a scenario similar to that in Section~\ref{sec: sim normal}, except that we first use Abadie's synthetic control method on all control units, and then select those with top $K$ weights to be donors $W$ and the others to be supplemental proxies $Z$. We also focus on the methods with correctly specified confounding bridge functions and the case $K=2$.

The results are shown in Figures~\ref{fig: sim detrend phi hat} and \ref{fig: sim detrend CI cover}.
Among the 2,400 estimates, 202 estimates (8.4\%) are outside the range $(1, 3)$ shown in Figure~\ref{fig: sim detrend phi hat}.
The estimators' distributions appear further from normal distributions compared to the case without post-selection, but the CI coverage still appears close to the nominal level 95\%. With post-selection, proximal methods appear to produce estimates with extremely large magnitudes and standard errors more often.
This is likely due to the fact that post-selection sometimes selects wrong units for proxies $W$ and $Z$, and thus wrong models for both confounding bridge functions. However, in general, we do not expect the CI coverage to be close to the nominal level.

\begin{figure}[bt]
    \centering
    \includegraphics[scale=0.7]{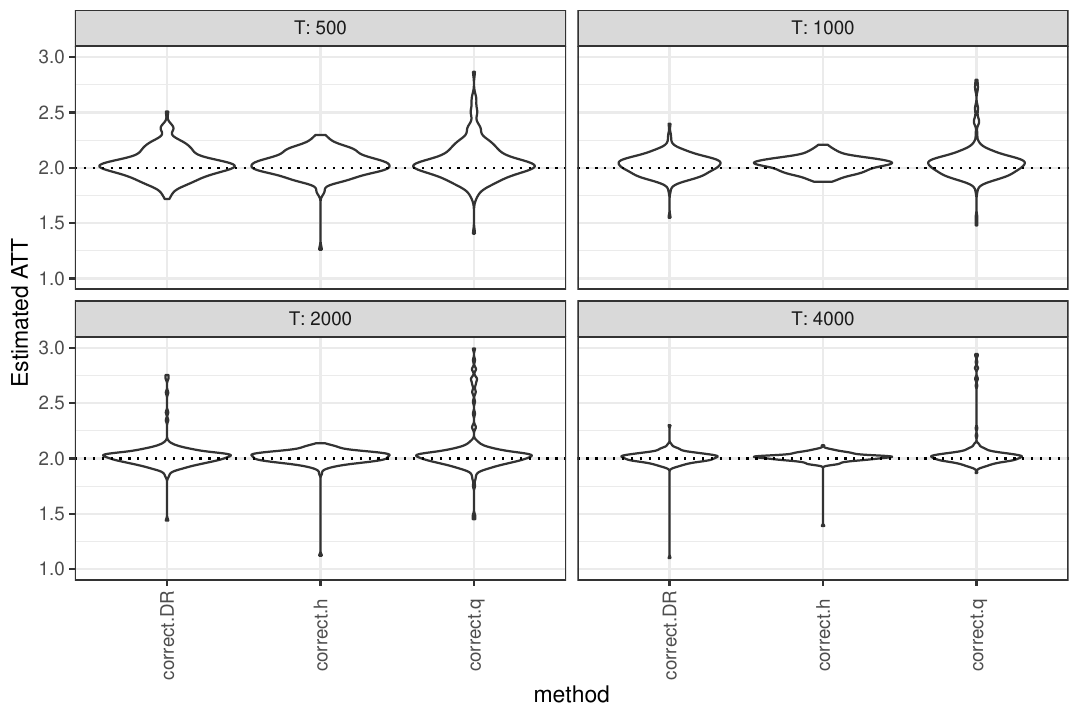}
    \caption{Figures similar to Figure~\ref{fig: sim phi hat} with post-selection.}
    \label{fig: sim post-select phi hat}
\end{figure}

\begin{figure}[bt]
    \centering
    \includegraphics[scale=0.7]{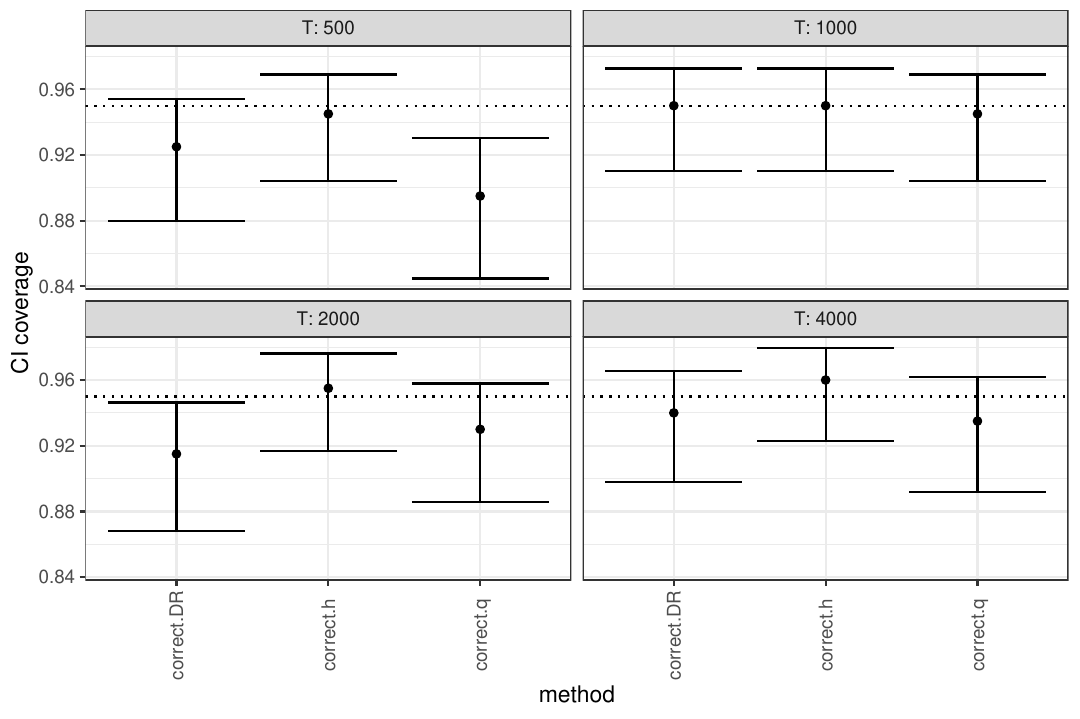}
    \caption{Figures similar to Figure~\ref{fig: sim CI cover} with post-selection.}
    \label{fig: sim post-select CI cover}
\end{figure}

\subsection{Nonstationarity and time-varying treatment effect} \label{sec: sim nonstationarity}

We also investigate the performance of our proposed method with nonstationarity and time-varying treatment effects. We consider a scenario similar to that in Section~\ref{sec: sim normal}, except that (i) $(W_t,Z_t,Y_t(0))$ has a seasonal trend $1.5 \sin(20 \pi t/T)$, and (ii) the true treatment effect at time $t$ is $(8t)/(3T)$, and thus the treatment effect averaged over post-treatment periods approaches the value two as $T \rightarrow \infty$. This scenario corresponds to Section~\ref{sec: nonstationary}. 
We focus on the proximal methods based on the generalized method of moments involving weighting, and we expect them to be consistent, with conservative inference.

The results are shown in Figures~\ref{fig: sim nonstationary phi hat}--\ref{fig: sim nonstationary SE SD}. As discussed in Section~\ref{sec: estimate ATT nonstationary}, the estimates appear consistent and asymptotically normal, but the CI coverage is above the nominal level.
Although the CI coverage is as high as 100\% for large sample sizes ($T \geq 2000$), the standard error overestimates the estimator's standard deviation by around 50--60\%, and thus the CI may still be meaningful.

\begin{figure}[bt]
    \centering
    \includegraphics[scale=0.7]{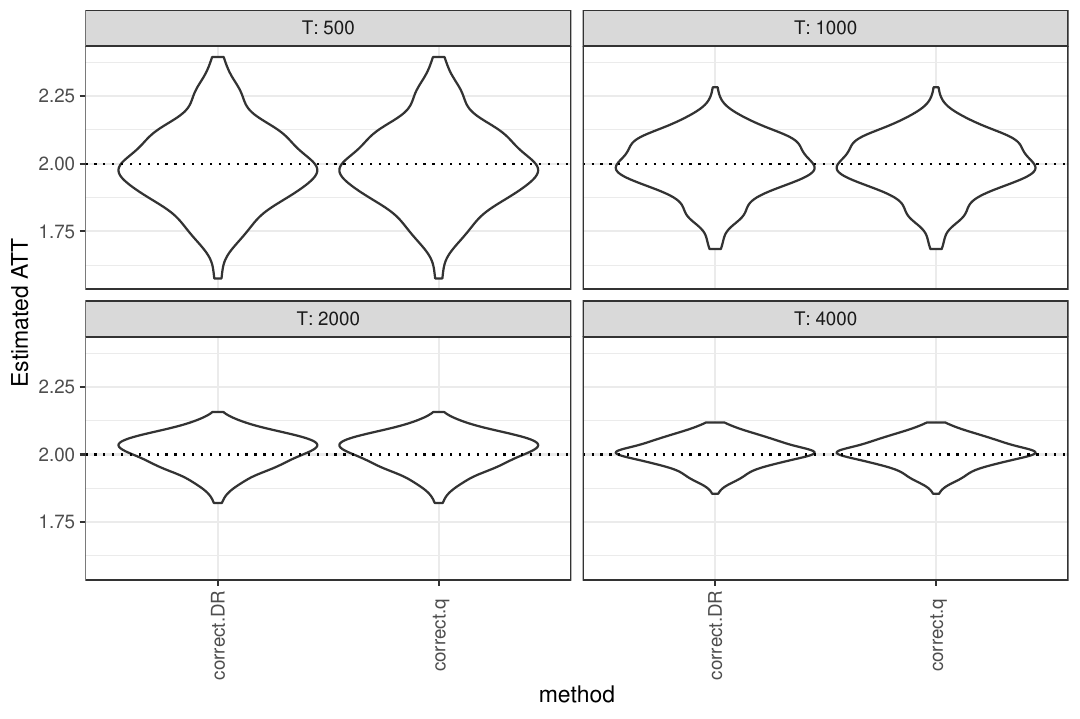}
    \caption{Figures similar to Figure~\ref{fig: sim phi hat} with nonstationarity and time-varying treatment effects.}
    \label{fig: sim nonstationary phi hat}
\end{figure}

\begin{figure}[bt]
    \centering
    \includegraphics[scale=0.7]{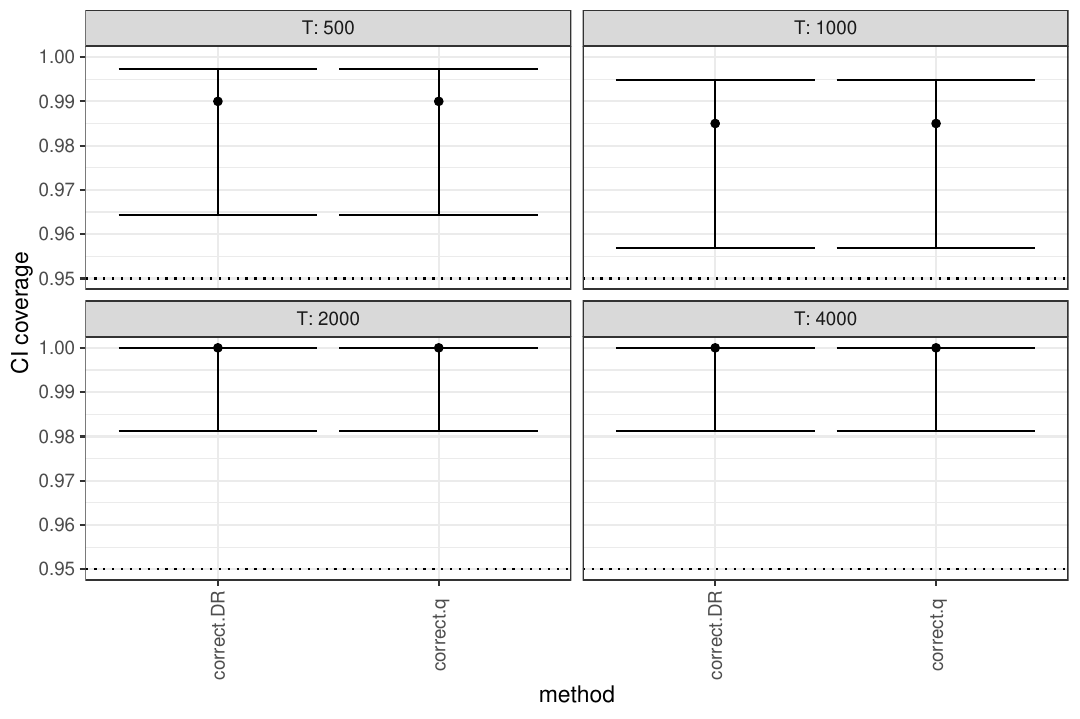}
    \caption{Figures similar to Figure~\ref{fig: sim CI cover} with nonstationarity and time-varying treatment effects.}
    \label{fig: sim nonstationary CI cover}
\end{figure}

\begin{figure}[bt]
    \centering
    \includegraphics[scale=0.7]{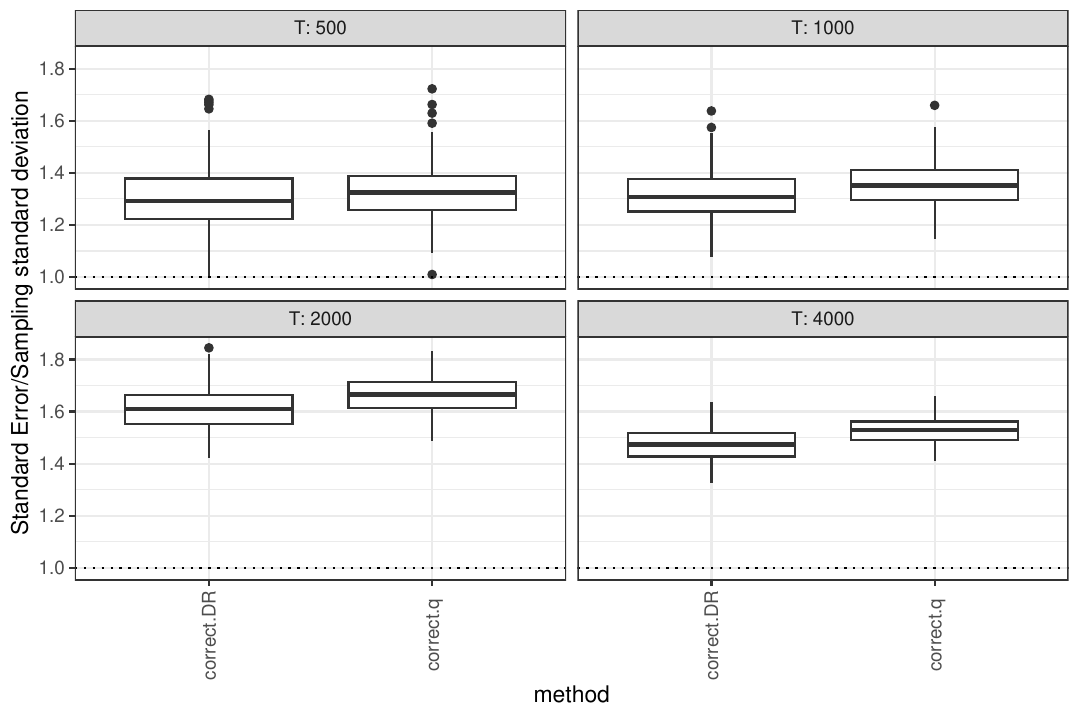}
    \caption{Ratio of the standard error to the standard deviation of the estimator's sampling distribution with nonstationarity and time-varying treatment effects.}
    \label{fig: sim nonstationary SE SD}
\end{figure}

\section{Additional examples} \label{sec: data analysis2}

\subsection{Analysis of Florida homicide rates} \label{sec: Florida analysis}

We apply our methods to study the effect of Florida's ``stand your ground'' law implemented in October 2005 on homicide rates \citep{Bonander2021}. In these data, monthly public safety measures are collected in 16 states including Florida from January 1999 to December 2014.

Condition~\ref{cond: treatment bridge} requires that the distribution of $U_{t_+}$ is dominated by $U_{t_-}$, and \eqref{eq: identify q} can only be solved if the distribution of $W_{t_+}$ is dominated by $W_{t_-}$. In practice, however, there might be a time trend in the time series $(W_t,Z_t,Y_t)$, in which case the treatment confounding bridge function would fail to exist. 
Therefore, we recommend preprocessing the data by detrending to account for the time trend, so that Condition~\ref{cond: treatment bridge} may be plausible.
One may fit a regression model $f: \{1,\ldots,T\} \to \real$ with covariate being time $t$ and outcome being the observed outcomes $W_t$ and $Z_t$ in all control units, use this model to predict an outcome $f(t)$ for each time period $t$, and finally
consider 
the residuals $(\tilde{W}_t,\tilde{Z}_t,\tilde{Y_t}) = (W_t-f(t),Z_t-f(t),Y_t-f(t))$ as the new detrended outcomes.

We then use these detrended outcomes for all subsequent analyses. This does not impact the classical synthetic control method of \cite{Abadie2010}, because, if we find $Y_t(0) \approx \sum_i \alpha_i W_{ti}$ where $W_t$ is the collection of $W_{ti}$ and $\sum_i \alpha_i=1$, then the same holds for detrended outcomes: $\tilde{Y}_t(0)= Y_t(0) - f(t) \approx \sum_i \alpha_i \tilde{W}_{ti}$ where $\tilde{W}_{ti} = W_{ti} - f(t)$. If $\sum_i \alpha_i \neq 1$, we still have $\tilde{Y}_t(0) \approx \sum_i \alpha_i \tilde{W}_{ti} + (\sum_i \alpha_i -1) f(t)$, that is, one may include an additional time-varying intercept to explicitly capture the time trend. Under the linear factor model of \cite{Abadie2010}, this intercept vanishes, and thus the model remains correctly specified. Therefore, the treatment effect on the residual can still be interpreted as the treatment effect on the original outcome.
The above discussion is meant to illustrate the interpretation of the detrended outcomes, but we do not make this assumption throughout the analysis.
However, as shown in Section~\ref{sec: sim detrend}, detrending might lead to slightly to moderately anti-conservative inference.

We detrend the Florida data by fitting a quadratic function of time to homicide rates. We model the time-varying average treatment effect $\phi_\lambda(t)$ for the treated unit as a constant. 
We also model the outcome bridge function as a linear function $h_\alpha: w \mapsto (1,w^\top)^\top \alpha$, and select the homicide rates in the states of Arkansas, Maryland, New Jersey, and New York as the proxy $W_t$.
This choice is motivated by Abadie's synthetic control method \citep{Abadie2010}, 
for which weights are negligible for all but these four states.
Our inference might be non-informative due to increased uncertainty associated with empirical selection of the proxy $W_t$. Proxy selection remains an open problem we are actively working on. Similar challenges are known to impact IV analysis with potentially invalid IVs.
We use all other states as control units in $Z_t$.
We consider the following two parametrizations of the treatment confounding bridge function, where the second has a larger set of states included in the model:
\begin{align*}
    & q_\beta(Z_t) = \exp \left\{ \beta_0 + \beta_{\mathrm{Delaware}} Z_{t,\mathrm{Delaware}} \right\}; \\
    & q_\beta(Z_t) = \exp \left\{ \beta_0 + \beta_{\mathrm{Delaware}} Z_{t,\mathrm{Delaware}} + \beta_{\mathrm{Ohio}} Z_{t,\mathrm{Ohio}} \right\}.
\end{align*}
We considered these two states in the treatment confounding bridge model due to their relative proximity and similarities to Florida. Thus, we hypothesize that these two states capture variation in the unmeasured confounder $U_t$. More states similar to Florida could have been included, had data from more control units been available.

The point estimates and 95\% confidence intervals for the average treatment effect for the treated unit produced by the above methods are presented in Table~\ref{tab: Florida results}. The trajectories of synthetic controls and the actual residuals of log GDP per capita in Kansas are presented in Figure~\ref{fig: Florida main analysis}.
We also conduct a falsification analysis of the placebo effect in October 2002. The results are presented in Table~\ref{tab: Florida results} and Figure~\ref{fig: Florida analysis}.

Because this time series is quite noisy, compared to typical data from econometric applications (e.g., data from Section~\ref{sec: Kansas analysis}), \texttt{OLS} does not yield a good pre-treatment fit. 
\texttt{Abadie's SC} appears to yield a reasonable pre-treatment fit because of the regularization, but the high noise level might still be a concern. The pre-treatment fits for proximal causal inference methods also appear reasonable and qualitatively similar to that of \texttt{Abadie's SC}. 
The point estimates of the effect of ``stand your ground'' law on homicide rates from all methods are positive. However, methods based on proximal causal inference involving the outcome confounding bridge function report an insignificant effect and much smaller point estimates than \texttt{Abadie's SC}, while other methods yield similar conclusions to \texttt{Abadie's SC}. 
Because of the concern about the noise level, the results from \texttt{DR}, \texttt{DR2} and \texttt{Outcome bridge} might be more reliable. 
The methods based on weighting might not be reliable due to a lack of data from control states that are similar to Florida.
These results suggest that, with a high noise level present, doubly robust methods may outperform \texttt{OLS} and \texttt{Abadie's SC} because doubly robust methods do not rely on a near-perfect pre-treatment fit.

\begin{table}
    \centering
    \begin{tabular}{l|c|c}
        Method & Florida's ``stand your ground'' law & placebo \\
        \hline\hline
        \texttt{Abadie's SC} & 0.083 & -0.025 \\
        \texttt{OLS} & 0.066 (0.045, 0.086) & -0.032 (-0.062, -0.002) \\
        \texttt{DR} & 0.024 (-0.061, 0.108) & -0.075 (-0.239, 0.088) \\
        \texttt{DR2} & 0.006 (-0.099, 0.110) & -0.042 (-0.086, 0.002) \\
        \texttt{Outcome bridge} & 0.012 (-0.087, 0.112) & -0.094 (-0.262, 0.075) \\
        \texttt{Treatment bridge} & 0.058 (0.038, 0.078) & -0.033 (-0.061, -0.005) \\
        \texttt{Treatment bridge2} & 0.057 (0.036, 0.078) & -0.039 (-0.092, 0.013)
    \end{tabular}
    \caption{Table similar to Table~\ref{tab: vaccine results} for Florida's ``stand your ground'' law.}
    \label{tab: Florida results}
\end{table}

\begin{figure}
    \centering
    \begin{subfigure}{0.495\textwidth}
    \includegraphics[width=\textwidth]{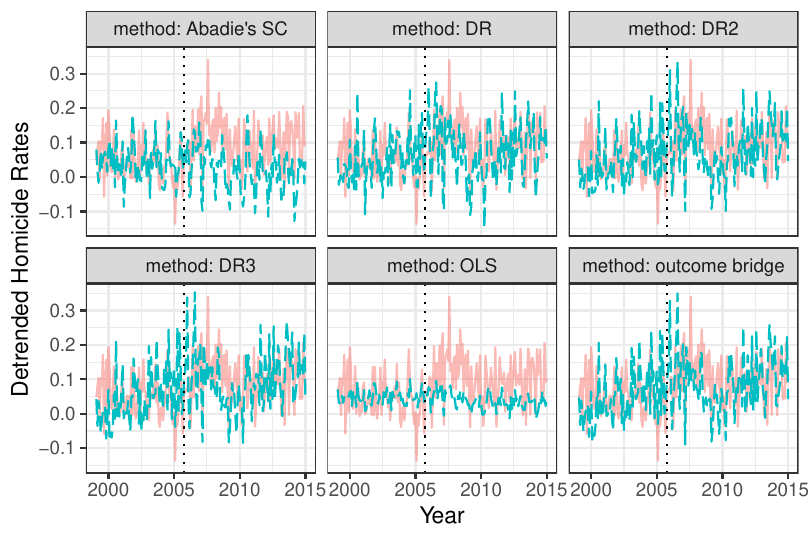}
    \caption{Florida ``stand your ground'' law}
    \label{fig: Florida main analysis}
    \end{subfigure}
    \hfill
    \begin{subfigure}{0.495\textwidth}
        \includegraphics[width=\textwidth]{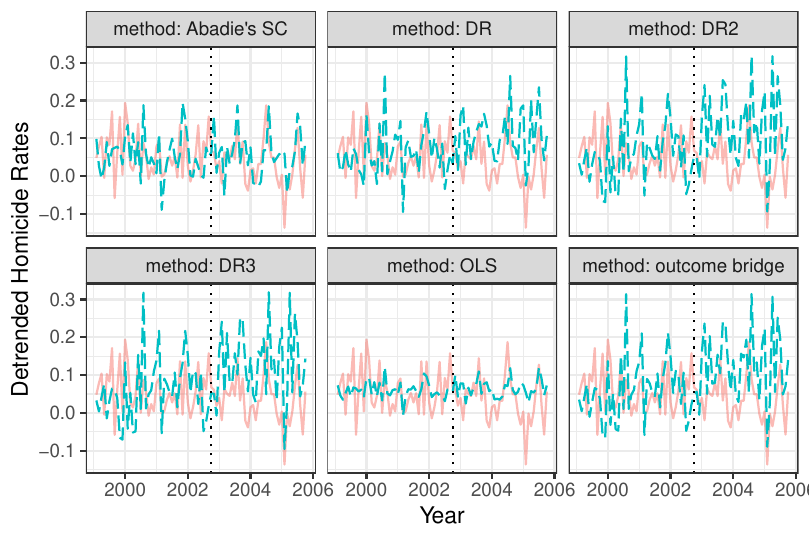}
        \caption{Placebo}
        \label{fig: Florida placebo analysis}
    \end{subfigure}
    \caption{Figures similar to Figure~\ref{fig: vaccine analysis} for Florida ``stand your ground'' law.}
    \label{fig: Florida analysis}
\end{figure}

\subsection{Analysis of Kansas GDP} \label{sec: Kansas analysis}

We also illustrate our methods by studying the effect of an aggressive tax cut in Kansas in the first quarter of 2012 on the logarithm of GDP per capita. 
This question has been studied in \cite{Rickman2018} using the classical synthetic control method proposed by \cite{Abadie2010,Abadie2015}. In these data, economic outcomes are measured every quarter from 1990 to 2016 for all US states. GDP was measured in millions of U.S. dollars.

We preprocess the data, choose proxies $(W_t,Z_t)$ and specify the confounding bridge functions $(h^*,q^*)$ similarly to Section~\ref{sec: Florida analysis}. Specifically, the donors forming the proxy $W_t$ are the states of North Dakota, South Carolina, Texas, and Washington as donors in $W_t$; the quarterly log GDP per capita of all other control states form the proxy $Z_t$.
We consider the following three parametrizations of the treatment confounding bridge function with increasing sets of states included in the model:
\begin{align*}
    & q_\beta(Z_t) = \exp \left\{ \beta_0 + \beta_{\mathrm{Iowa}} Z_{t,\mathrm{Iowa}} \right\}; \\
    & q_\beta(Z_t) = \exp \left\{ \beta_0 + \beta_{\mathrm{Iowa}} Z_{t,\mathrm{Iowa}} + \beta_{\mathrm{South\ Dakota}} Z_{t,\mathrm{South\ Dakota}} \right\}; \\
    & q_\beta(Z_t) = \exp \left\{ \beta_0 + \beta_{\mathrm{Iowa}} Z_{t,\mathrm{Iowa}} + \beta_{\mathrm{South\ Dakota}} Z_{t,\mathrm{South\ Dakota}} + \beta_{\mathrm{Oklahoma}} Z_{t,\mathrm{Oklahoma}} \right\}.
\end{align*}
These states are chosen because they are similar to Kansas and likely to capture the effect of the unmeasured confounder $U_t$.

The analysis results are presented in Table~\ref{tab: Kansas results} and Figure~\ref{fig: Kansas analysis}.
All methods support a decrease in log GDP per capita due to the tax cut; 
the decrease is significant for all methods providing 95\% confidence intervals (namely all methods except \texttt{Abadie's SC}); the only exception is the weighted method based on a single state included in the treatment confounding bridge model, \texttt{treatment bridge}. This result may be due to model misspecification, and the doubly robust method with the same treatment confounding bridge model is still significant.
Point estimates of the average treatment effect for the treated unit from doubly robust methods are almost twice those of \texttt{Abadie's SC} and weighted methods. Doubly robust estimates are well within sampling variability based on the proximal outcome bridge estimates, suggesting that the outcome confounding bridge function may be correctly specified. 
The synthetic controls projected potential outcome trajectories of all methods have similar fits in the pre-treatment period; however, \texttt{Abadie's SC} and \texttt{OLS} appear to lead to a lower synthetic control trajectory in the post-treatment period.

\begin{table}
    \centering
    \begin{tabular}{l|c|c}
        Method & Kansas tax cut & placebo \\
        \hline\hline
        \texttt{Abadie's SC} & -0.048 & 0.029 \\
        \texttt{OLS} & -0.069 (-0.087, -0.050) & 0.026 ($2.6 \times 10^{-6}$, 0.052) \\
        \texttt{DR} & -0.077 (-0.126, -0.028) & 0.004 (-0.068, 0.077) \\
        \texttt{DR2} & -0.095 (-0.147, -0.043) & -0.005 (-0.039, 0.030) \\
        \texttt{DR3} & -0.103 (-0.228, -0.021) & -0.007 (-0.059, 0.046) \\
        \texttt{Outcome bridge} & -0.104 (-0.150, -0.058) & 0.012 (-0.069, 0.093) \\
        \texttt{Treatment bridge} & -0.031 (-0.087, 0.024) & -0.028 (-0.063, 0.008) \\
        \texttt{Treatment bridge2} & -0.017 (-0.032, -0.002) & -0.042 (-0.056, -0.0027) \\
        \texttt{Treatment bridge3} & -0.016 (-0.029, -0.003) & -0.048 (-0.097, 0.001)
    \end{tabular}
    \caption{Table similar to Table~\ref{tab: vaccine results} for tax cut in Kansas.}
    \label{tab: Kansas results}
\end{table}

\begin{figure}
    \centering
    \begin{subfigure}{0.495\textwidth}
    \includegraphics[width=\textwidth]{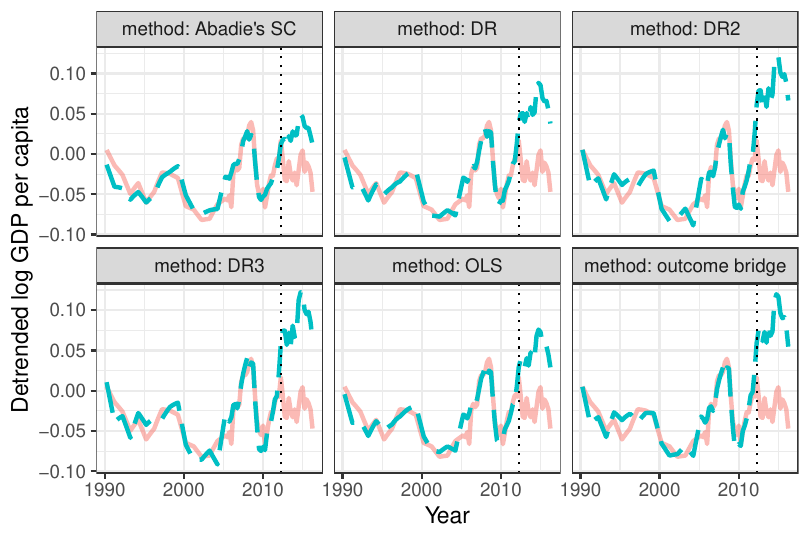}
    \caption{Kansas tax cut}
    \label{fig: Kansas main analysis}
    \end{subfigure}
    \hfill
    \begin{subfigure}{0.495\textwidth}
        \includegraphics[width=\textwidth]{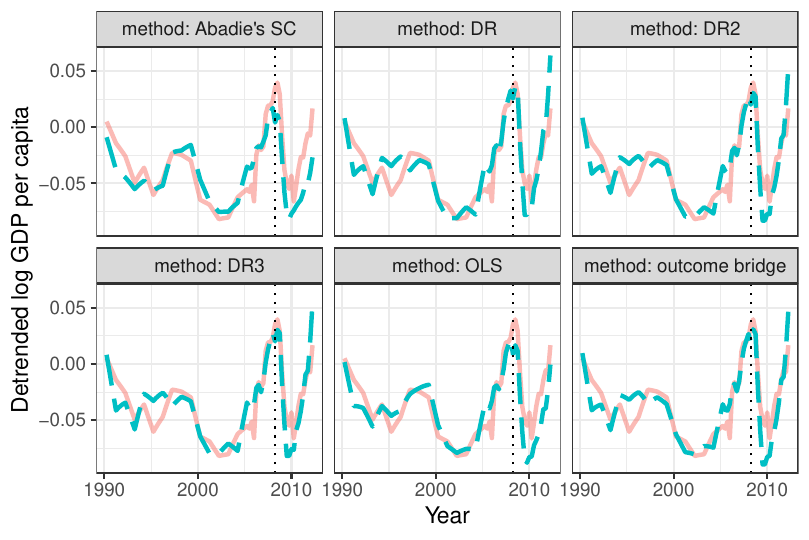}
        \caption{Placebo}
        \label{fig: Kansas placebo analysis}
    \end{subfigure}
    \caption{Figures similar to Figure~\ref{fig: vaccine analysis} for tax cut in Kansas.}
    \label{fig: Kansas analysis}
\end{figure}

We further conduct a falsification analysis of the placebo effect in the first quarter of 2008.
The analysis results are presented in Table~\ref{tab: Kansas results} and Figure~\ref{fig: Kansas placebo analysis}. 
The 95\%-confidence intervals from all methods cover zero, correctly indicating a non-significant effect in the placebo period, with the exception of \texttt{OLS} and the weighted method with two states only included in the treatment confounding bridge model. The point estimates of doubly robust methods are also closer to zero. These results suggest a potentially superior performance of doubly robust methods and a potential bias of \texttt{OLS}.

\section{Existence of treatment confounding bridge function} \label{sec: bridge exists}

In this appendix, we briefly state examples of sufficient conditions for Condition~\ref{cond: treatment bridge}. Sufficient conditions for Conditions~\ref{cond: treatment bridge nonstationary} and \ref{cond: treatment bridge nonstationary limit} are similar. Sufficient conditions for existence of various confounding bridge functions have been considered in \cite{Miao2018}, \cite{Cui2020} and \cite{Shi2021}. In particular, sufficient conditions for Condition~\ref{cond: outcome bridge} are presented in \cite{Shi2021}. We refer readers to these works for more details. These conditions are only sufficient and not necessary.

\subsection{Discrete random variables}

Fix any $t_-\le T_0$ and $t_+>T_0$.
Suppose that, $(Z_{t_-},U_{t_-})$ is a discrete random vector with finite support. Suppose that the supports of $U_{t_-}$ and $Z_{t_-}$ are $\{u_1,\ldots,u_I\}$ and $\{z_1,\ldots,z_J\}$, respectively. Let $P_{Z \mid U}$ be a $J \times I$ matrix with the $(j,i)$-th entry being $\Prob(Z_{t_-}=z_j \mid U_{t_-}=u_i)$. Let $R_U$ be the $I$-dimensional vector of likelihood ratios $\Prob(U_{t_+}=u_i)/\Prob(U_{t_-}=u_i)$. 
For any function $q: \mathcal{Z} \rightarrow \real$, we may equivalently represent $z \mapsto q(z)$ by a $J$-dimensional vector $Q$ with entries $q(z_j)$, $1\le j\le J$. Condition~\ref{cond: treatment bridge} is then equivalent to the existence of a solution in $Q$ to
\begin{equation}
    P_{Z \mid U}^\top Q = R_{U}. \label{eq: discrete completeness to solve}
\end{equation}
Therefore, a sufficient condition for Condition~\ref{cond: treatment bridge} is that $P_{Z \mid U}$ has full column rank, namely $\rank(P_{Z \mid U})=I$, and the solution to \eqref{eq: discrete completeness to solve} is identical for all $t_- \leq T_0$, which is implied by stationarity of $(Z_{t_-},U_{t_-})$ ($t_- \leq T_0$).

\subsection{Continuous random variables}

We first fix any $t_-\le T_0$ and $t_+>T_0$ in this appendix.
Suppose that $(Z_{t_-},U_{t_-})$ is a continuous random vector. For any distribution $P$, recall that $L^2(P)$ denotes the space of all square-integrable functions with respect to $P$, which is a Hilbert space equipped with inner product $\langle f,g \rangle = \int f g \  \intd P$.
Let $K_{t_-}:L^2(P_{Z_{t_-}}) \rightarrow L^2(P_{U_{t_-}})$ be the operator defined pointwise by $K_{t_-} q : u \mapsto \cexpect\{q(Z_{t_-}) \mid U_{t_-}=u\}$ for $q \in L^2(P_{Z_{t_-}})$.
We assume the following regularity conditions:

\begin{condition} \label{reg cond: compact operator}
    For all $t_- \leq T_0$, $\iint p_{Z_{t_-} \mid U_{t_-}}(z \mid u) p_{U_{t_-} \mid Z_{t_-}}(u \mid z) \ \intd z \intd u < \infty$.
\end{condition}

By Example~2.3 in \cite{Carrasco2007} (pages~5656 and 5659), under Condition~\ref{reg cond: compact operator}, the operator $K_{t_-}$ is compact \citep[see, e.g., Definition~2.17, ][]{Kress2014}. Let $(\lambda_{t_-,m},\varphi_{t_-,m},\psi_{t_-,m})_{m=1}$ be a singular system of $K_{t_-}$ \citep[see, e.g., Theorem~15.16 of][for more about singular systems]{Kress2014}.

\begin{condition} \label{reg cond: square-integrable}
    For all $t_- \leq T_0$,
    $$\int \left\{ \frac{\intd P_{U_{t_+}}}{\intd P_{U_{t_-}}}(u) \right\}^2 p_{U_{t_-}}(u) \ \intd u < \infty.$$
\end{condition}

\begin{condition} \label{reg cond: series}
    For all $t_- \leq T_0$,
    $$\sum_{m=1}^\infty \lambda_{t_-,m}^{-2} \left| \left\langle \frac{\intd P_{U_{t_+}}}{\intd P_{U_{t_-}}}, \psi_{t_-,m} \right\rangle \right|^2 < \infty.$$
\end{condition}

The next assumption is a completeness condition, which is similar to Conditions~\ref{cond: outcome complete} and \ref{cond: treatment complete}.

\begin{condition} \label{reg cond: complete}
    Let $g \in L^2(P_{U_{t_-}})$. For all $t_- \leq T_0$, the following two statements are equivalent: (i) $g(U_{t_-})=0$, (ii) $\cexpect\{g(U_{t_-}) \mid Z_{t_-}\} = 0$.
\end{condition}

Under Conditions~\ref{reg cond: compact operator}--\ref{reg cond: complete}, Condition~\ref{cond: treatment bridge} holds for the fixed $t_-$; that is, a solution to \eqref{eq: define q} exists for the fixed $t_-$. This claim can be proved by Picard's Theorem \citep[Theorem~15.18 of][]{Kress2014}. We next sketch the proof. The orthogonal complement $N(K_{t_-}^*)^\perp$ of the nullspace of the adjoint of $K_{t_-}$ equals $L^2(P_{Z_{t_-}})$ by Condition~\ref{reg cond: complete}. By Condition~\ref{reg cond: square-integrable}, $\intd P_{U_{t_+}}/\intd P_{U_{t_-}}$ lies in $N(K_{t_-}^*)^\perp$. The desired existence result then follows by directly applying Picard's Theorem. Condition~\ref{cond: treatment bridge} then holds if the solution to \eqref{eq: define q} is identical for all $t_- \leq T_0$.

\section{Proofs} \label{sec: proof}

\subsection{Proof of Theorem~\ref{thm: identify ATT treatment}}
\label{pf: thm: identify ATT treatment}

We first prove \eqref{eq: identify EfY0 treatment}. Let $f:\real \rightarrow \real$ be any square-integrable function.
Then,
\begin{align*}
	&\cexpect[q^*(Z_{t_-}) f\{Y_{t_-}(0)\}]
	= \cexpect(\cexpect[q^*(Z_{t_-}) f\{Y_{t_-}(0)\} \mid U_{t_-}]) \\
	&= \cexpect(\cexpect\{q^*(Z_{t_-}) \mid U_{t_-}\} \cexpect[f\{Y_{t_-}(0)\} \mid U_{t_-}]) & \text{(Condition~\ref{cond: proxy independence})} \\
	&= \cexpect \left( \frac{\intd P_{U_{t_+}}}{\intd P_{U_{t_-}}}(U_{t_-}) \cexpect[f\{Y_{t_-}(0)\} \mid U_{t_-}] \right) & \text{(Conditions~\ref{cond: posttreatment identical distribution} \& \ref{cond: treatment bridge})} \\
	&= \int_{\mathcal{U}} \frac{\intd P_{U_{t_+}}}{\intd P_{U_{t_-}}}(u) \cexpect[f\{Y_{t_-}(0)\} \mid U_{t_-}=u] \intd P_{U_{t_-}}(u) \\
	&= \int_{\mathcal{U}} \cexpect[f\{Y_{t_+}(0)\} \mid U_{t_+}=u] \intd P_{U_{t_+}}(u) & \text{(Condition~\ref{cond: covariate shift})} \\
	&= \cexpect[f\{Y_{t_+}(0)\}].
\end{align*}
Thus, \eqref{eq: identify EfY0 treatment} has been proved and \eqref{eq: identify EY0 treatment} follows immediately.
We next prove \eqref{eq: identify q}. Let $f: \mathcal{W} \rightarrow \real$ be any square-integrable function.
\begin{align*}
    & \cexpect\{q^*(Z_{t_-}) f(W_{t_-})\}
    = \cexpect[ \cexpect\{q^*(Z_{t_-}) f(W_{t_-}) \mid U_{t_-}\} ] \\
    &= \cexpect[ \cexpect\{q^*(Z_{t_-}) \mid U_{t_-}\} \cexpect\{f(W_{t_-}) \mid U_{t_-}\} ] & \text{(Condition~\ref{cond: proxy independence})} \\
    &= \cexpect \left[ \frac{\intd P_{U_{t_+}}}{\intd P_{U_{t_-}}}(U_{t_-}) \cexpect\{f(W_{t_-}) \mid U_{t_-}\} \right] & \text{(Conditions~\ref{cond: posttreatment identical distribution} \& \ref{cond: treatment bridge})} \\
    &= \int_{\mathcal{U}} \frac{\intd P_{U_{t_+}}}{\intd P_{U_{t_-}}} (u) \cexpect\{f(W_{t_-}) \mid U_{t_-}=u\} \, \intd P_{U_{t_-}}(u) \\
    &= \int_{\mathcal{U}} \cexpect\{f(W_{t_+}) \mid U_{t_+}=u\} \intd P_{U_{t_+}}(u) =
    \cexpect\{f(W_{t_+})\}.
    & \text{(Condition~\ref{cond: covariate shift})} 
\end{align*}
For a Borel set $B \subseteq \mathcal{U}$,
take $f(w)=\ind(w \in B)$. 
If $P_{W_{t_-}}(B)=0$, the above equality implies that $P_{W_{t_+}}(B)=\cexpect\{q^*(Z_{t_-}) \ind(W_{t_-} \in B)\}=0$. Therefore, $P_{W_{t_+}}$ is dominated by $P_{W_{t_-}}$ and the Radon-Nikodym derivative $\intd P_{W_{t_+}}/\intd P_{W_{t_-}}$ is well defined. Thus, for any integrable function $f$,
$$\cexpect\{q^*(Z_{t_-}) f(W_{t_-})\} = \cexpect \left\{ \frac{\intd P_{W_{t_+}}}{\intd P_{W_{t_-}}} (W_{t_-}) f(W_{t_-}) \right\},$$
that is,
$$\cexpect \left[ \left\{ q^*(Z_{t_-}) - \frac{\intd P_{W_{t_+}}}{\intd P_{W_{t_-}}} (W_{t_-}) \right\} f(W_{t_-}) \right]=0.$$
Since $f$ is arbitrary, we have that
$$\cexpect \left\{ q^*(Z_{t_-}) - \frac{\intd P_{W_{t_+}}}{\intd P_{W_{t_-}}} (w) \mid W_{t_-}=w \right\}=0$$
for $P_{W_{t_-}}$-a.e. $w \in \mathcal{W}$. Equation~\ref{eq: identify q} follows.

We finally prove the uniqueness of $q^*$ under Condition~\ref{cond: treatment complete}. Suppose that two functions $q_1$ and $q_2$ both solve \eqref{eq: identify q}. Then, $\cexpect\{q_1(Z_{t_-}) - q_2(Z_{t_-}) \mid W_{t_-}\}=0$, and thus $q_1(Z_{t_-})=q_2(Z_{t_-})$ by Condition~\ref{cond: treatment complete}.

\subsection{Proof of Theorem~\ref{thm: identify ATT DR}}

We study the two cases where $h=h^*$ and $q=q^*$ separately. If Condition~\ref{cond: outcome bridge} holds and $h=h^*$, then, by 
\eqref{eq: identify EY0 outcome} and
\eqref{eq: identify h},
\begin{equation}\label{d}
\cexpect[q(Z_{t_-}) \{Y_{t_-} - h(W_{t_-})\} + h(W_{t_+})] = \cexpect\{Y_{t_+}(0)\}
\end{equation}
and
\begin{equation}\label{e}
    \cexpect[Y_{t_+} - q(Z_{t_-}) \{Y_{t_-} - h(W_{t_-})\} - h(W_{t_+})] = \cexpect\{Y_{t_+}(1)-Y_{t_+}(0)\} = \phi^*(t_+),
\end{equation}
as desired.

If Condition~\ref{cond: treatment bridge} holds and $q=q^*$, then by a similar argument as in the proof of Theorem~\ref{thm: identify ATT treatment},
\begin{align*}
	&\cexpect[q^*(Z_{t_-}) \{ Y_{t_-}(0) - h(W_{t_-}) \}]
	= \cexpect(\cexpect[q^*(Z_{t_-}) \{ Y_{t_-}(0) - h(W_{t_-})\} \mid U_{t_-}]) \\
	&= \cexpect[\cexpect\{q^*(Z_{t_-}) \mid U_{t_-}\} \cexpect\{Y_{t_-}(0)-h(W_{t_-}) \mid U_{t_-}\}] & \text{(Condition~\ref{cond: proxy independence})} \\
	&= \int_{\mathcal{U}} \frac{\intd P_{U_{t_+}}}{\intd P_{U_{t_-}}}(u) \cexpect\{Y_{t_-}(0) - h(W_{t_-}) \mid U_{t_-}=u\} P_{U_{t_-}}(\intd u) & \text{(Conditions~\ref{cond: posttreatment identical distribution} \& \ref{cond: treatment bridge})} \\
	&= \int_{\mathcal{U}} \cexpect\{Y_{t_+}(0) - h(W_{t_+}) \mid U_{t_+}=u\} P_{U_{t_+}}(\intd u) 
	= \cexpect\{Y_{t_+}(0) - h(W_{t_+})\}.
	&\text{(Condition~\ref{cond: covariate shift})} 
\end{align*}
Therefore, \eqref{d} and \eqref{e} hold, 
as desired. We have proved Theorem~\ref{thm: identify ATT DR}.

\subsection{Proof of Theorems~\ref{thm: identify ATT treatment nonstationary}--\ref{thm: identify ATT DR nonstationary limit}}

The proofs of these theorems are similar to those of Theorems~\ref{thm: identify ATT treatment} and \ref{thm: identify ATT DR} and thus we abbreviate the presentation.

\begin{proof}[Proof of Theorem~\ref{thm: identify ATT treatment nonstationary}]
    By Condition~\ref{cond: covariate shift}, for any square-integrable functions $f$ and $g$, $\cexpect[f\{Y_t(0)\} \mid U_t=u]$ and $\cexpect[g(W_t) \mid U_t=u]$ do not depend on $t$. We first prove \eqref{eq: identify EfY0 treatment nonstationary}:
    \begin{align*}
        & \sum_{t_+=T_0+1}^T \ell_T(t_+) \cexpect[f\{Y_{t_+}(0)\}] 
        = \sum_{t_+=T_0+1}^T \ell_T(t_+) \cexpect( \cexpect[f\{Y_{t_+}(0)\} \mid U_{t_+}] ) \\
        &= \frac{1}{T_0} \sum_{t_-=1}^{T_0} \cexpect[\undertilde{q}^*_T(Z_{t_-}) \cexpect\{f(Y_{t_-}) \mid U_{t_-}\} ] 
        = \frac{1}{T_0} \sum_{t_-=1}^{T_0} \cexpect\{\undertilde{q}^*_T(Z_{t_-}) f(Y_{t_-})\}. & \text{(Conditions~\ref{cond: treatment bridge nonstationary} and \ref{cond: proxy independence})}
    \end{align*}
    Then, \eqref{eq: identify EY0 treatment nonstationary} follows by taking $f$ to be the identity function in \eqref{eq: identify EfY0 treatment nonstationary}.
    We next show that $\undertilde{q}^*_T$ is a solution to \eqref{eq: identify q nonstationary}. For any square-integrable function $g$,
    \begin{align*}
        &\frac{1}{T_0} \sum_{t_-=1}^{T_0} \cexpect\{\undertilde{q}(Z_{t_-}) g(W_{t_-})\} 
        = \frac{1}{T_0} \sum_{t_-=1}^{T_0} \cexpect[ \undertilde{q}(Z_{t_-}) \cexpect\{g(W_{t_-}) \mid U_{t_-}\} ] \text{(Condition~\ref{cond: proxy independence})} \\
        &= \sum_{t_+=T_0+1}^T \ell_T(t_+) \cexpect[ \cexpect\{g(W_{t_+}) \mid U_{t_+}\} ]
        = \sum_{t_+=T_0+1}^T \ell_T(t_+) \cexpect\{g(W_{t_+})\}.
    \end{align*}
    Hence, $\undertilde{q}^*_T$ is a solution to \eqref{eq: identify q nonstationary}.
    We finally prove the uniqueness of the solution to \eqref{eq: identify q nonstationary}. Suppose that two functions $\undertilde{q}_{T,1}$ and $\undertilde{q}_{T,2}$ are solutions to \eqref{eq: identify q nonstationary}. Then, for any square-integrable function $g$,
    $$\frac{1}{T_0} \sum_{t_-=1}^{T_0} \cexpect[ \{ \undertilde{q}_{T,1}(Z_{t_-}) - \undertilde{q}_{T,2}(Z_{t_-}) \} g(W_{t_-})] = 0.$$
    By Condition~\ref{cond: treatment complete nonstationary}, $\undertilde{q}_{T,1}(Z_{t_-}) - \undertilde{q}_{T,2}(Z_{t_-})=0$ for all $t_- \leq T_0$ and thus the solution is unique almost surely.
\end{proof}

\begin{proof}[Proof of Theorem~\ref{thm: identify ATT DR nonstationary}]
    First consider the case where Condition~\ref{cond: outcome bridge} holds and $h=h^*$. In this case, using the identification result in \cite{Shi2021} (Theorem~4), we have that
    \begin{align*}
        &\frac{1}{T_0} \sum_{t_-=1}^{T_0} \cexpect \left[ \undertilde{q}(Z_{t_-}) \{Y_{t_-} - h(W_{t_-})\} \right] + \sum_{t_+=T_0+1}^T \ell_T(t_+) \cexpect\{h(W_{t_+})\} \\
        &= \frac{1}{T_0} \sum_{t_-=1}^{T_0} \cexpect \left[ \cexpect\{\undertilde{q}(Z_{t_-}) \mid U_{t_-}\} \cexpect\{Y_{t_-} - h(W_{t_-}) \mid U_{t_-}\} \right] + \sum_{t_+=T_0+1}^T \ell_T(t_+) \cexpect\{h(W_{t_+})\} & \\
        &= 0 + \sum_{t_+=T_0+1}^T \ell_T(t_+) \cexpect\{Y_{t_+}(0)\}.
    \end{align*}
    The second line follows by Condition~\ref{cond: proxy independence}.
    Therefore, $\sum_{t_+=T_0+1}^T \ell_T(t_+) \cexpect\{Y_{t_+}(0)\}$ is identified as in \eqref{eq: DR identify ATT nonstationary}.
    Next suppose that Condition~\ref{cond: treatment bridge nonstationary} holds and $\undertilde{q}=\undertilde{q}^*_T$. Using Theorem~\ref{thm: identify ATT DR nonstationary}, we have that
    \begin{align*}
        & \frac{1}{T_0} \sum_{t_-=1}^{T_0} \cexpect \left[ \undertilde{q}(Z_{t_-}) \{Y_{t_-} - h(W_{t_-})\} \right] + \sum_{t_+=T_0+1}^T \ell_T(t_+) \cexpect[h(W_{t_+})] \\
        &= \frac{1}{T_0} \sum_{t_-=1}^{T_0} \cexpect \left\{ \undertilde{q}(Z_{t_-}) Y_{t_-} \right\} - \frac{1}{T_0} \sum_{t_-=1}^{T_0} \cexpect \left\{ \undertilde{q}(Z_{t_-}) h(W_{t_-}) \right\} + \sum_{t_+=T_0+1}^T \ell_T(t_+) \cexpect\{h(W_{t_+})\} \\
        &= \sum_{t_+=T_0+1}^T \ell_T(t_+) \cexpect\{Y_{t_+}(0)\} + 0.
    \end{align*}
    Therefore, $\sum_{t_+=T_0+1}^T \ell_T(t_+) \cexpect\{Y_{t_+}(0)\}$ is identified as in \eqref{eq: DR identify ATT nonstationary}.
    Finally, the identification formula for $\undertilde{\phi}^*_T$ in \eqref{eq: DR identify ATT nonstationary} follows by noting that
    $$\undertilde{\phi}^*_T = \sum_{t_+=T_0+1}^T \ell_T(t_+) \cexpect\{Y_{t_+} - Y_{t_+}(0)\}.$$
\end{proof}

The proof of Theorems~\ref{thm: identify ATT treatment nonstationary limit} and \ref{thm: identify ATT DR nonstationary limit} is almost identical and thus omitted.

\subsection{Proof of Theorems~\ref{thm: DR ATT asymptotic normality} \& \ref{thm: treatment ATT asymptotic normality}}

Theorem~\ref{thm: DR ATT asymptotic normality} follows immediately from standard estimation theory of the generalized method of moments, for example, Theorem~7.1 and 7.2 in \cite{Wooldridge1994}, along with Theorem~\ref{thm: identify ATT DR}.
The proof of Theorem~\ref{thm: treatment ATT asymptotic normality} is similar.

\subsection{Proof of Theorems~\ref{thm: DR ATT asymptotic normality nonstationary} \& \ref{thm: treatment ATT asymptotic normality nonstationary}}

The proofs of 
Theorems~\ref{thm: DR ATT asymptotic normality nonstationary} \& \ref{thm: treatment ATT asymptotic normality nonstationary} are almost identical. 
Therefore, we present the proof of Theorem~\ref{thm: DR ATT asymptotic normality nonstationary} and omit the proof of Theorem~\ref{thm: treatment ATT asymptotic normality nonstationary}. 
The argument is inspired by the theory of  \citep[see, e.g.,][]{Wooldridge1994,Hall2007} with adaptations to our case where the moment equation might not be centered at each time period.

\begin{proof}[Proof of Theorem~\ref{thm: DR ATT asymptotic normality nonstationary}]
    Under Conditions~\ref{cond: proxy independence}, \ref{cond: covariate shift} and \ref{cond: misspecify h or q nonstationary}, we have that $\undertilde{\phi}_\infty=\undertilde{\phi}^*$ by the definition of $\undertilde{G}_{T,t}$ in \eqref{eq: Gt nonstationary} and Theorem~\ref{thm: identify ATT DR nonstationary limit}. We first prove consistency. By Conditions~\ref{cond: positive weight matrix} and \ref{cond: uniform law of large numbers}, we have that, as $T \rightarrow \infty$,
    letting $V_T(\theta) = \sum_{t=1}^T \undertilde{G}_{T,t}(\theta)/T$
    and
    $\bar V_{T'}(\theta) = \sum_{t=1}^{T'} \cexpect[\undertilde{G}_{T',t}(\theta)]/T'$,
    \begin{align}
        \begin{split}
            \sup_{\undertilde{\theta} \in \undertilde{\Theta}} \Bigg| &V_T(\undertilde{\theta})^\top \undertilde{\Omega}_T V_T(\undertilde{\theta}) 
            - \lim_{T' \rightarrow \infty} \bar V_{T'}(\undertilde{\theta})^\top \undertilde{\Omega} \lim_{T' \rightarrow \infty} \bar V_{T'} \Bigg|
            \text{ converge to zero in probability}
        \end{split} \label{eq: uniform convergence of quadratic form}
    \end{align}
    Let $\epsilon>0$ be an arbitrary positive constant. By \eqref{eq: uniform convergence of quadratic form} and the definition of $\hat{\undertilde{\theta}_T}$, we have that, with probability tending to one,
    \begin{align*}
        &\Bigg| V_T(\undertilde{\theta}_\infty) ^\top \undertilde{\Omega}_T V_T(\undertilde{\theta}_\infty) - \lim_{T' \rightarrow \infty} \bar V_{T'}(\undertilde{\theta}_\infty)^\top \undertilde{\Omega} \lim_{T' \rightarrow \infty} \bar V_{T'}(\undertilde{\theta}_\infty) \Bigg| < \epsilon/2, \\
        &\Bigg| V_T(\undertilde{\theta}_T) ^\top \undertilde{\Omega}_T V_T(\undertilde{\theta}_T) 
        - \lim_{T' \rightarrow \infty} \bar V_{T'}(\hat{\undertilde{\theta}}_T)^\top \undertilde{\Omega} \lim_{T' \rightarrow \infty} \bar V_{T'}(\hat{\undertilde{\theta}}_T)  \Bigg| < \epsilon/2, \\
        &V_T(\undertilde{\theta}_T) ^\top \undertilde{\Omega}_T V_T(\undertilde{\theta}_T)  \leq V_T(\undertilde{\theta}_\infty) ^\top \undertilde{\Omega}_T V_T(\undertilde{\theta}_\infty) .
    \end{align*}
    Combining these three inequalities, we have that, with probability tending to one,
    \begin{align}
    \begin{split}
        &\lim_{T' \rightarrow \infty} \bar V_{T'}(\hat{\undertilde{\theta}}_T)^\top \undertilde{\Omega} \lim_{T' \rightarrow \infty} \bar V_{T'}(\hat{\undertilde{\theta}}_T)
        < \lim_{T' \rightarrow \infty} \bar V_{T'}(\undertilde{\theta}_\infty)^\top \undertilde{\Omega} \lim_{T' \rightarrow \infty} \bar V_{T'}(\undertilde{\theta}_\infty) + \epsilon = \epsilon,
    \end{split} \label{eq: GMM consistency inequality}
    \end{align}
    where the last equality follows from Condition~\ref{cond: misspecify h or q nonstationary}.
    
    Let $N \subseteq \undertilde{\Theta}$ be an arbitrary open set containing $\undertilde{\theta}_\infty$. By Condition~\ref{cond: param compact}, $\undertilde{\Theta} \setminus A$ is compact. By Conditions~\ref{cond: misspecify h or q nonstationary}, \ref{cond: positive weight matrix} and \ref{cond: reg conditions for estimating function},
    $$\inf_{\undertilde{\theta} \in \undertilde{\Theta} \setminus N} \lim_{T' \rightarrow \infty} \bar V_{T'}(\undertilde{\theta})^\top \undertilde{\Omega} \lim_{T' \rightarrow \infty} \bar V_{T'}$$
    exists and is strictly positive. Taking $\epsilon$ in \eqref{eq: GMM consistency inequality} to be the above this infimum, we have that, with probability tending to one,
    \begin{align*}
        &\lim_{T' \rightarrow \infty} \bar V_{T'}(\hat{\undertilde{\theta}}_T)^\top \undertilde{\Omega} \lim_{T' \rightarrow \infty} \bar V_{T'}(\hat{\undertilde{\theta}}_T)
        < \inf_{\undertilde{\theta} \in \undertilde{\Theta} \setminus N} \lim_{T' \rightarrow \infty} \bar V_{T'}(\undertilde{\theta})^\top \undertilde{\Omega} \lim_{T' \rightarrow \infty} \bar V_{T'}.
    \end{align*}
    This event implies that $\hat{\undertilde{\theta}}_T \in N$. Since $N$ is arbitrary, we have shown that 	$\hat{\undertilde{\theta}}_T$ converges to $\undertilde{\theta}_\infty$ in probability as $T \rightarrow \infty$.
    We next prove the asymptotic normality of $\hat{\undertilde{\theta}}_T$. Under Condition~\ref{cond: uniform law of large numbers derivative}, by a first-order Taylor expansion and the above consistency result, we have that
    \begin{equation}
        \frac{1}{T} \sum_{t=1}^T G_{T,t}(\hat{\theta}_T) = \frac{1}{T} \sum_{t=1}^T G_{T,t}(\theta_\infty) + \frac{1}{T} \sum_{t=1}^T \nabla_{\theta} G_{T,t}(\theta) |_{\theta=\theta_\infty} (\hat{\theta}_T-\theta_\infty) + \smallo_p(\| \hat{\theta}_T-\theta_\infty \|). \label{eq: GMM Taylor}
    \end{equation}
    The fact that $\hat{\undertilde{\theta}}_T$ is a minimizer together with Condition~\ref{cond: uniform law of large numbers derivative} implies that
    \begin{align}
        0 &= \nabla_{\undertilde{\theta}} \left. \left\{ V_T(\undertilde{\theta})^\top \undertilde{\Omega}_T V_T(\undertilde{\theta}) \right\} \right|_{\undertilde{\theta} = \hat{\undertilde{\theta}}_T} 
        = 2 \left\{ \frac{1}{T} \sum_{t=1}^T \nabla_{\undertilde{\theta}} \undertilde{G}_{T,t}(\undertilde{\theta})|_{\undertilde{\theta} = \hat{\undertilde{\theta}}_T} \right\}^\top \undertilde{\Omega}_T V_T(\undertilde{\theta}_T) . \label{eq: GMM derivative 0}
    \end{align}
    By Conditions~\ref{cond: uniform law of large numbers}--\ref{cond: CLT weighted moment}, using \eqref{eq: GMM Taylor} and \eqref{eq: GMM derivative 0}, we have that
    \begin{align*}
        0 &= T^{1/2} \left\{ \frac{1}{T} \sum_{t=1}^T \nabla_{\undertilde{\theta}} \undertilde{G}_{T,t}(\undertilde{\theta})|_{\undertilde{\theta} = \hat{\undertilde{\theta}}_T} \right\}^\top \undertilde{\Omega}_T V_T(\undertilde{\theta}_T)  \\
        &= \undertilde{R}^\top \undertilde{\Omega} \frac{1}{T^{1/2}} \sum_{t=1}^T \undertilde{G}^q_{T,t}(\undertilde{\theta}_\infty) + \undertilde{R}^\top \undertilde{\Omega} \undertilde{R} T^{1/2} (\hat{\undertilde{\theta}}_T-\undertilde{\theta}_\infty) + \smallo_p(T^{1/2} \| \hat{\theta}_T-\theta_\infty \| + 1).
    \end{align*}
    By Condition~\ref{cond: CLT weighted moment} and Slutsky's Theorem, we have that
    \begin{align*}
        T^{1/2} (\hat{\undertilde{\theta}}_T-\undertilde{\theta}_\infty) &= -(\undertilde{R}^\top \undertilde{\Omega} \undertilde{R})^{-1} \undertilde{R}^\top \undertilde{\Omega} \frac{1}{T^{1/2}} \sum_{t=1}^T \undertilde{G}^q_{T,t}(\undertilde{\theta}_\infty) + \smallo_p(1) \\
        &= -\undertilde{A}^{-1} \undertilde{R}^\top \undertilde{\Omega} \frac{1}{T^{1/2}} \sum_{t=1}^T \undertilde{G}^q_{T,t}(\undertilde{\theta}_\infty) + \smallo_p(1),
    \end{align*}
    which converges in distribution to $\mathrm{N}(0, \undertilde{A}^{-1} \undertilde{B} \undertilde{A}^{-1})$. Here, the existence of $\undertilde{A}^{-1}$ follows from Conditions~\ref{cond: positive weight matrix} and \ref{cond: full rank matrix}.
\end{proof}

\bibliographystyle{plainnat}
\bibliography{ref}

\end{document}